\documentclass[a4paper,normalheadings]{scrreprt}
\pdfoutput=1

\usepackage[utf8]{inputenc}
\usepackage[T1]{fontenc}

\usepackage{enumerate}
\usepackage{graphicx}
\usepackage{epsfig}
\usepackage{subfigure}
\usepackage{latexsym}
\usepackage{amsmath}
\usepackage{amssymb}
\usepackage{multirow}
\usepackage{dsfont}
\usepackage{wrapfig}
\usepackage[T1]{fontenc}
\usepackage[ruled,vlined,linesnumbered]{algorithm2e}

\usepackage{stmaryrd}
\usepackage{dcolumn}
\usepackage{url}

\usepackage{pgf}
\usepackage{tikz}
\usetikzlibrary{decorations.pathreplacing,automata,arrows}

\usepackage{a4}
\usepackage{stmaryrd}
\usepackage{dcolumn}

\usepackage{amsthm}

\usepackage[small,bf]{caption}
\DeclareCaptionFormat{myformat}{#1#2#3}
\newtheorem{theorem}{Theorem}
\newtheorem{lemma}{Lemma}[chapter]

\newtheorem{corollary}{Corollary}
\newtheorem{example}{Example}

\newcommand{\A}{\mathcal{A}}
\newcommand{\B}{\mathcal B}
\newcommand{\T}{\mathcal{T}}
\newcommand{\prt}[2]{{#1}/{#2}} 
\newcommand{\tr}{\delta} 
\newcommand{\labset}{\Sigma}
\newcommand{\ltsq}{Q} 
\newcommand{\xtr}{\xrightarrow} 
\newcommand{\pre}[1]{\mathit{pre_{#1}}}
\newcommand{\wordsim}{\preccurlyeq}

\newcommand{\wordsimequiv}{\cong}
\newcommand{\init}{I}

\newcommand{\Rel}{\mathit{Rel}} 
\newcommand{\Split}{\mathit{Split}}
\newcommand{\remove}{\mathit{Remove}}
\newcommand{\ass}{\leftarrow}
\newcommand{\prev}{\mathsf{prev}}
\newcommand{\parent}{\mathsf{parent}}
\newcommand{\card}[1]{|#1|}
\newcommand{\trset}{\delta} 
\newcommand{\nltr}[1]{\stackrel{#1}{\nrightarrow}}
\newcommand{\ind}[1]{\mathrel{R_{#1}}} 
\newcommand{\ltr}{\xtr}
\newcommand{\relc}{\mathit{Count}}
\newcommand{\ine}{\mathit{in}}
\newcommand{\ltrset}[1]{\delta_{#1}}
\newcommand{\Sim}{\mathit{sim}}
\renewcommand{\O}{\mathcal{O}}
\newcommand{\fin}{\mathit{fin}}
\newcommand{\Start}{\mathit{Start}}

\newcommand{\maxrank}[1]{\hat r}
\renewcommand{\part}[2]{{#1}/{#2}}

\newcommand{\eq}[1]{\equiv_{#1}}

\newcommand{\rel}[1]{#1}

\newcommand{\nat}{\mathbb{N}}

\newcommand{\vect}[2]{(#1_1,\dotsc,#1_{#2})}
\newcommand{\trans}[4]{\vect{#1}{#2}\xtr{#3}{#4}}
\newcommand{\hole}{\square}

\newcommand{\ds}{\preceq}
\newcommand{\db}{\simeq}
\newcommand{\ubdb}{\stackrel{\bullet}{\simeq}}
\newcommand{\ubds}{\mathrel{\stackrel{\bullet}{\raisebox{0pt}[0.99ex][0pt]{\ensuremath{\ds}}}}}
\newcommand{\ubid}{\stackrel{\bullet}{=}}
\newcommand{\usds}{\mathrel{\stackrel{\circ}{\raisebox{0pt}[0.99ex][0pt]{\ensuremath{\ds}}}}}
\newcommand{\usdb}{\stackrel{\circ}{\simeq}}
\newcommand{\usid}{\stackrel{\circ}{=}}

\newcommand{\env}{\mathit{Env}}

\newcommand{\node}{p}

\newcommand{\hide}[1]{}

\newcommand{\K}{\mathcal K}

\newcommand{\eqn}{P_\mathit{rel}}
\newcommand{\eqm}{\delta_\mathit{rel}}

{
\newcommand{\Processed}[0]{\mathit{Processed}}
\newcommand{\Next}[0]{\mathit{Next}}
\newcommand{\lang}[1]{L(#1)}

\newcommand{\dist}{\mathit{Dist}}
\newcommand{\minimize}{\mathit{Minimize}}
\newcommand{\univ}{\mathit{Univ}}
\newcommand{\initialize}{\mathit{Initialize}}
\newcommand{\incl}{\mathit{Incl}}
\newcommand{\ps}[2]{(#1,#2)}

\def\centerframe#1#2#3{%
\framebox{%
\parbox{0pt}{\rule{0pt}{#2}}\parbox{#1}{\makebox[#1]{#3}}}}

\renewcommand{\node}{v}
\renewcommand{\L}{\mathcal L}

\newcommand{\tup}[2]{(#1_1,\ldots,#1_{#2})}

\newcommand{\post}{\mathit{Post}}
\newcommand{\mstates}{\mathit{MStates}}
\newcommand{\pstates}{\mathit{PStates}}

\renewcommand{\rel}[1]{#1^\subseteq}

\newcommand{\false}{\mathsf{FALSE}}
\newcommand{\true}{\mathsf{TRUE}}

\begin{document}
\title{Simulations and Antichains for Efficient Handling of Finite Automata}
\author{Luk\'a\v s Hol\'ik}
\maketitle


\pagenumbering{roman}
\chapter*{Abstract} 
This thesis is focused on techniques for finite automata and their
use in practice, with the main emphasis on nondeterministic tree automata. This
concerns namely techniques for size reduction and language inclusion testing,
which are two problems that are crucial for many applications of tree automata.
For size reduction of tree automata, we adapt the simulation quotient technique that is well
established for finite word automata. We give efficient algorithms for
computing tree automata simulations and we also introduce a new type of
relation that arises from a combination of tree automata downward and upward
simulation and that is very well suited for quotienting. The combination
principle is relevant also for word automata. We then generalise the so called
antichain universality and language inclusion checking technique developed
originally for finite word automata for tree automata.  Subsequently, we
improve the antichain technique for both word and tree automata by combining it
with the simulation-based inclusion checking techniques, significantly
improving efficiency of the antichain method. We then show how the developed
reduction and inclusion checking methods improve the method of abstract regular
tree model checking, the method that was the original motivation for starting
the work on tree automata. Both the reduction and the language inclusion
methods are based on relatively simple and general principles that can be
further extended for other types of automata and related formalisms. An example
is our adaptation of the reduction methods for alternating B\"uchi automata,
which results in an efficient alternating automata size reduction technique.

\noindent
\begin{minipage}{\textwidth}
\chapter*{Keywords} 
Finite automata, finite tree automata, alternating
B\"uchi automata, nondeterminism, simulation, bisimulation, universality, language inclusion, antichain, quotienting,
regular tree model checking.
\end{minipage}

\chapter*{Abstrakt} 
Cílem této práce je vývoj technik umožňujících praktické využití
nedetermi\-ni\-stických konečných automatů, zejména nedeterministických stromových
auto\-matů. Jde zvláště o techniky pro redukci velikosti a testování jazykové
inkluze, jež hrají zásadní roli v mnoha oblastech aplikace konečných automatů.
V oblasti redukce velikosti vycházíme z dobře známých metod pro slovní automaty
které jsou založeny na relacích simulace.  Navrhli jsme efektivní algoritmy pro
výpočet stromových variant simulačních relací a identifikovali jsme nový typ
relace založený na kombinaci takzvaných horních a dolních simulací nad
stromovými automaty. Tyto kombinované relace jsou zvláště vhodné pro redukci
velikosti automatů slučováním stavů. Navržený princip kombinace relací simulace
je re\-le\-vantní i pro slovní automaty.  Náš přínos v oblasti testování jazykové
inkluze je dvojí. Nejprve jsme zobecnili na stromové automaty takzvané
protiřetězcové algoritmy, které byly původně navrženy pro slovními automaty.
Dále se nám podařilo použitím simulačních relací výrazně zefektivnit
protiřetězcové algoritmy pro testování jazykové inkluze jak pro slovní, tak pro
stromové automaty. Re\-le\-vanci našich technik pro praxi jsme demonstrovali jejich
nasazením v rámci regulárního stromového model checkingu, což je verifikační
metoda založená na stromových automatech. Použití našich algoritmů zde vedlo k
výraznému zry\-chle\-ní a zvětšení škálovatelnosti celé metody. Základní myšlenky
našich algoritmů pro redukci velikosti automatů a testování jazykové inkluze
jsou apli\-kovatelné i na jiné typy automatů. Příkladem jsou naše redukční
techniky pro alternující B\"uchiho automaty prezentované v poslední části
práce.

\noindent
\begin{minipage}{\textwidth}
\chapter*{Klíčová slova} 
Konečný automat, konečný stromový automat, alternující Büchiho automat, nedeterminismus, univerzalita,
jazyková inkluze, protiřetězec, simulace, bisimulace, redukce velikosti, regulární stromový model checking.
\end{minipage}
\vfill

\noindent
\begin{minipage}{\textwidth}
\chapter*{Citace} 
Lukáš Holík, Simulations and Antichains for Efficient Handling of Finite Automata, disertační práce, Brno, FIT VUT v Brně, 2010
\end{minipage}

\chapter*{Simulations and Antichains for
Efficient Handling of Finite
Automata} 

\section*{Prohlášení}
Prohlašuji, že jsem tuto disertační práci vypracoval samostatně pod vedením
doc. Tomáše Vojnara. Uvedl jsem všechny literární prameny a publikace, ze
kterých jsem čerpal.

\vspace{1.5cm}
\hfill
\begin{minipage}{2.5cm}
\begin{center}
\dotfill\\
Lukáš Holík\\
26.\,října 2010
\end{center}
\end{minipage}

\vfill
\noindent
\copyright Lukáš Holík, 2010.\\
{\it 
Tato práce vznikla jako školní dílo na Vysokém učení technickém v Brně, Fakultě
informačních technologií. Práce je chráněna autorským zákonem a její užití bez
udělení oprávnění autorem je nezákonné, s~vý\-jim\-kou zákonem definovaných
případů.
}

\chapter*{Acknowledgements}
\label{acknoledgements}
I am most grateful to my advisor Tom\'a\v s Vojnar for his thoughtful approach and
the enormous effort he spent when teaching me what it means to do 
research in computer science.
I appreciate his trust that this investment would eventually pay off, which was
a great source of motivation for me.
I~must also thank him for the opportunity to meet great people from our field,
especially prof. Bouajjani, prof. Abdulla, doc.  Habermehl, doc. Mayr. and also
younger colleagues Dr. Kaati (the queen of tree automata), Dr. Chan., and Dr. Rogalewicz.
I~was continuously learning from them during our discussions, especially about
the importance of talking and carefully listening to others.
They deserve my thanks for always patiently listening to me (it was not always easy). 
I wish to express my gratitude to prof. \v Ce\v ska for his support and for his
contribution towards creating an environment where a work such as mine is possible. 
I~also thank my family for standing by me and for believing that the
things I do make sense. 
I thank Marie for her love and patience.

\vfill
{\it
The work presented in this thesis was supported by
the Czech Science Foundation (projects 102/07/0322, 102/09/H042, 103/10/0306), 
the Czech institutional project MSM 0021630528,
the Barrande projects MEB 020840 and 2-06-27, 
the Czech COST project OC10009 associated with the ESF COST action IC0901,
the internal BUT FIT grant FIT-S-10-1,
and the ESF project Games for Design and Verification.
}

\tableofcontents
\pagenumbering{arabic}
\chapter{Introduction}
\label{chapter:introduction}

Finite automata on finite words (FA) are one of the basic concepts of computer
science.  Besides classical applications of FA such as compiler construction or
text searching, FA are widely used in modelling and verification, which are the
application domains of our interest.  Tree automata (TA) are a natural
generalisation of FA that accepts ordered trees/terms.  TA share most of the
good properties of FA, from closure to decidability and complexity (even though
complexities of many tree automata problems are higher, they are still
comparable with the complexities of the corresponding FA ones). This makes tree
automata a convenient tool for modelling and reasoning about various kinds of
structured objects such as syntactical trees, structured documents,
configurations of complex systems, algebraic term representations of data or
computations, etc.  (see, e.g., \cite{tata97}). One of the main motivations for
this work is in particular the use of tree automata in verification, mainly in the method of
regular tree model checking
\cite{Shahar:01:ToolsTechVerParmSys:01,bouajjani:extrapolating,abdulla:simulation,bouajjani:abstractTree},
an infinite-state system verification method where tree automata are used for
representing sets of reachable states of a system.

In the above context, checking language equivalence/inclusion and reducing size
of automata while preserving the language are fundamental issues, and
performing these operations efficiently is crucial in practice. The language inclusion problem and the
minimisation problem for (nondeterministic) automata are PSPACE-complete for FA
and even EXPTIME-complete for TA.  A classical approach to cope with these
problems is determinisation. Both FA as well as TA can be determinised and
minimised in a canonical way.  Testing language inclusion of deterministic
minimal automata is then easy. However, since even the canonical minimal
deterministic automaton can still be exponentially larger than the original
nondeterministic one, its computation easily becomes a major bottleneck of any
automata-based method.

A reasonable and pragmatic approach to the size reduction and language
inclusion problem is to consider some relation on states of an automaton that
respects language inclusion on states, but which can be checked efficiently,
using a~polynomial algorithm. Such a relation can then be used for
approximating language inclusion between two automata by checking whether each initial state of one automaton is related to an initial state of other automaton. This method is sound but incomplete in the case when the relation is a proper subset the language inclusion on states. Such a relation can be also used for reducing the
size of an automaton by collapsing equivalent states.  Here, a natural
trade-off between the strength of the considered relation and the cost of
its computation arises.  In the case of word automata, a relation which is
widely considered as a good trade-off in this sense is simulation preorder.  It can be checked in polynomial time, and efficient
algorithms have been designed for this purpose (see, e.g.,
\cite{gentiliny:fromBisimulation,henzinger:computing,ranzato:new,crafa:saving}).
These algorithms make the computation of simulation preorder quite
affordable even in comparison with the one of bisimulation equivalence, which is cheaper
\cite{hopcroft:nlogn,paige:three,valmari:bisimilarity}, but which is also
stronger, and therefore leads to less significant reductions of automata and
also its capability of approximating language inclusion is limited.

As for what concerns language inclusion and universality problem, apart from
the classical determinisation-based methods and simulation-based approximation
technique, there has recently been proposed the so called antichain
universality and inclusion testing method for FA \cite{wulf:antichains}. It is
essentially an optimisation of the classical method based on subset
construction (i.e., on determinisation), it is still of an exponential worst
case complexity, but it behaves very well in practice. 

In the case of tree automata, the only methods for size reduction that were
previously studied (apart from deterministic minimisation) are  based on
bisimulation relations \cite{abdulla:bisimulation, hogberg:backward} and
concerning language inclusion testing, the only methods formerly available are
the classical ones based on explicit determinisation. However, these methods
are not efficient enough. The former ones are rather weak since bisimulation
relations are usually relatively sparse and the latter ones suffer from the
problem of state space explosion too often. 

\section{Goals of the Thesis} 
The lack of efficient methods for reducing size and testing language inclusion
of nondeterministic tree automata described above has significantly limited
their practical usability. Therefore, this thesis is aimed at adapting
techniques that work well for word automata to tree automata, which in
particular concerns the size reduction methods based on simulations and the
language inclusion testing algorithms based on the antichain principle. Then,
apart from generalising existing methods from word automata to tree automata,
we also focus on improving the existing methods themselves. This concerns
introduction of new types of relations suitable for reducing the size of word
as well as tree automata and interconnecting the antichain principle with the
simulation techniques into new language inclusion testing algorithms.
Additionally, we show that the proposed methods are applicable to other kinds
of automata too by designing a simulation-based reduction method for
alternating B\"uchi automata that is similar to the one we proposed for tree
automata.

\section{An Overview of Achieved Results} 
Here we summarise the contributions that we have achieved within the particular
areas marked out by the goals of this work. 

\paragraph{Tree Automata Reduction Methods.}
%
Our tree automata reduction methods are build on the notions of downward and
upward tree automata simulations (proposed first in \cite{abdulla:simulation})
that are the tree automata counterparts the forward and backward FA
simulations. 

We design efficient algorithms for computing tree automata simulations. A
deep examination of the structure of the TA simulations reveals that both
upward and downward TA simulations can be computed by the same algorithmic
pattern.  More specifically, the problems of computing a TA simulation can be
reduced to a problem of computing a common FA simulation (a tree automaton is
translated into an FA and then a common FA simulation algorithm is used).
Moreover, tree automata bisimulations can also be computed efficiently this way
using the same translations (instead of a simulation algorithm, an FA
bisimulation algorithm is run on the FA obtained by translating the input TA). The resulting tree automata bisimulation algorithms are simple
and competitive with the previously known algorithms from
\cite{hogberg:backward}. This results in a uniform and elegant framework for
computing tree automata simulations and bisimulations that can utilise the
best FA simulation and bisimulation algorithms.

We have identified a principle of combining upward and downward TA simulations
and forward and backward FA simulations that yields an equivalence, called
mediated equivalence, suitable for reducing automata by collapsing their states
while preserving the language. Mediated equivalence is coarser than downward
resp. forward simulation equivalence and thus gives a better reduction. The
principle of mediated minimisation of FA generalises the principle of forward
simulation minimisation. Two forward simulation equivalent states can be safely
collapsed since they have the same forward languages (symmetrically for
backward simulation).  In contrary, the property that allow collapsing two
mediated equivalent states $p$ and $q$ is the following.  Whenever there is a
computation under a word $u$ starting in an initial state that ends in a state
$p$, and another computation under a word $v$ starting in a state $q$ and
ending in a final state, then there is a computation under $uv$ from an initial
to a final state.  Therefore, collapsing the two states $p,q$ does not
introduce any new behaviour since every word accepted via the new state was
accepted also before collapsing.  The case of TA mediated equivalence can be
explained analogically. It may be seen from the above that unlike simulations,
mediated equivalences approximate neither forward nor backward language
equivalence on states, and similarly the tree automata mediated equivalence is
not compatible with any notion of language of a state of a tree automaton.  The
combination principle allows to build a mediated equivalence from any
downward/backward relation (simulation, bisimulation or identity relation) and
any upward/forward relation (simulation, bisimulation, identity). This yields a
scale of mediated equivalences offering a fine choice between the computation
cost and reduction power, as confirmed by our experimental results.

\paragraph{Language Inclusion Checking for TA and FA.}
Our  universality and language inclusion algorithms for tree and word automata
build on the antichain based method for FA proposed first in
\cite{wulf:antichains}. It is a complete method that optimises the classical
subset construction based algorithms. We first briefly review its main idea.  

Consider a nondeterministic FA $\A$. In the simpler case of universality
checking, the method is based on a search for a nonaccepting state of the
determinised version $\A'$ of $\A$ reachable from an initial state of $\A'$.
Such a state is a counterexample to universality of $\A$.  When a
counterexample is reached, the algorithm may terminate even before all states
of $\A'$ are constructed.  The states of $\A'$, called macro-states, have the
form of subsets of the set of states of $\A$.  The key idea is that some
macro-states have a better chance of finding a counterexample than other ones since
they have provably smaller languages (in our terminology, we say that they
subsume the states with larger languages). Therefore, one can safely continue
searching only from the generated macro-states that have minimal languages, and
simply discard any generated macro-state that is subsumed by another one. In
\cite{wulf:antichains}, the subsumption relation is just set inclusion, and
already this simple solution gives a fundamental speedup. 

We first adapt the FA antichain technique for tree automata. The adaptation is
quite straightforward, and similarly as in the case of FA, it has a major
impact on efficiency of the TA language inclusion and universality tests. We
then improve the antichain technique for both FA and TA by interconnecting it
with the simulation approximation technique. Simply speaking, we improve accuracy
of the subsumption relation on macro-states by employing simulations on states
of the original automaton. In the case of universality checking, a macro-state
$p$ subsumes a macro state $q$ if all states in $p$ are simulated by some state
in $q$. Moreover, even the internal structure of macro-states can be
simplified by keeping only simulation maximal states of $\A$ inside the macro-states.
In the case of testing inclusion between two automata $\A$ and $\B$,
macro-states have a more complicated structure, and it is possible to utilise
simulation on states of $\A$, on states of $\B$, and also use simulation between
states of $\A$ and $\B$.  It can be said that this method combines advantages
of both simulation approximation of language inclusion and the original antichain technique. It
also behaves very well on our experimental data. 

\paragraph{Simulations and Antichains in Abstract Regular Tree Model Checking.}
We have shown practical applicability of our tree automata reduction and
inclusion testing methods in the framework of abstract regular tree model
checking (ARTMC), an infinite state verification method where the two problems
play a crucial role.
In regular model checking (RMC), we start with an FA $\A_I$ representing a set
of initial configurations $I$ of a system and iteratively apply transition
relation $\tau$ (symbolically, on the structure of the automaton) until a
fixpoint is reached, thus computing an FA representing the set $\tau^*(\A_I)$
of all configurations reachable from the initial configurations. Then, it is checked whether this set satisfies the verified properties. 
In abstract regular model
checking \cite{bouajjani:abstract}, abstraction (together with a counterexample
guided refinement) is used to accelerate the computation.  Checking the
fixpoint condition means to decide whether
$\tau^i(A_I)\subseteq\tau^{i+1}(\A_I)$, which requires an efficient language
inclusion algorithm.  During the computation, the intermediate automata
typically grow quickly, therefore it is needed to reduce their size.  Tree
automata are used instead of FA when configurations of the system being
verified are better represented by trees than by words, e.g., certain
parametrised communication protocols, pointer programs manipulating tree-like
data structures etc. In that case, we speak about abstract regular tree model
checking (ARTMC)
\cite{bouajjani:extrapolating,abdulla:regular,bouajjani:abstractTree,bouajjani:abstractComplex}.
This method was originally based on deterministic tree automata, involving
implicit determinisation after each step.  Our reduction and inclusion testing
methods allowed us to redesign the method on top of nondeterministic tree
automata, which led to a major increase of scalability and efficiency.

\paragraph{Simulations and Antichains for Other Types of Automata.}
The principles of our simulation-based reduction methods are relatively simple
and general which allows extensions of the methods also for other types of
automata. We have done this for alternating B\"uchi automata (ABA), for which
we have designed simulation-based reduction method analogical to the one
proposed for tree automata.  ABA are acceptors of infinite words with the same
expressive power as B\"uchi automata, but may be exponentially more succinct.
Their applications can be found for instance in automata-based LTL model checking
within a B\"uchi automata complementation procedure (e.g.,
\cite{kupferman:weak}).  Alternating B\"uchi automata are similar to tree
automata in the sense that runs of both types of automata have a form of trees
(ordered trees for TA and unordered trees for ABA).  Therefore, the
definitions of simulations look similar for the two types of automata.  Forward
simulation over alternating B\"uchi automata have been already studied (see
\cite{fritz:state,fritz:simulation}).  It may bee seen as an analogy of the
tree automata downward simulation.  We have introduced the notion of ABA
backward simulation, which is an analogy of TA upward simulation.  We also show
that it is possible to combine the ABA simulations in the same way as the TA
simulation into a mediated equivalence suitable for collapsing states while
preserving language.  This equivalence gives better reductions than sole
forward simulation, which we confirm also by experiments.

Generalisations of our universality and language inclusion algorithms are also
possible. We are currently exploring ways of applying these techniques at
deciding B\"uchi automata universality and language inclusion. Our first result
has been published as \cite{abdulla:simulationsubsumption} where we use the
simulation subsumption technique to improve the so called Ramsey-based B\"uchi
universality and inclusion test (see, e.g.,
\cite{sistla:complementation,fogarty:buchi}). However, this work is already beyond the scope of this thesis.


\section{Plan of the Thesis} 
Chapter~\ref{chapter:preliminaries} contains preliminaries on automata, simulations, and regular tree model checking.
Chapter~\ref{chapter:LTS_simulation} presents an algorithm for computing
simulations over labelled transition systems used within most of the
algorithms presented further. In Chapter~\ref{chapter:ta_reduction}, we
describe our simulation and bisimulation-based framework for reducing tree
automata and the algorithms for computing the TA simulations and
bisimulations.  Chapter~\ref{chapter:fa_ta_inclusion} deals with the language
inclusion and universality problems for FA and TA.  Alternating B\"uchi
automata simulation-based reduction methods are discussed in
Chapter~\ref{chapter:aba_reduction} and Chapter~\ref{chapter:conclusions}
concludes the thesis.

\include{preliminaries}
\chapter{Computing Simulations over Labelled Transition Systems}
\label{chapter:LTS_simulation}

This chapter is devoted to an algorithm for computing simulations on labelled transition systems.  As discussed in the previous
chapter, simulation is a good candidate for reducing transition
systems by collapsing equivalent states and also for approximating
language/trace inclusion.  It strongly preserves logics like  $ACTL^*$,
$ECTL^*$, and $LTL$ \cite{dams:generation,grumberg:model,henzinger:computing},
and with respect to its reduction power and computation cost, it offers a
desirable compromise among the other common candidates, such as bisimulation
equivalence \cite{paige:three,sawa:behavioural} and language equivalence.  Our
main motivation for presenting the algorithm here is that computing simulation
over an LTS is a crucial step of almost all algorithms presented later in this
thesis, namely algorithms for computing simulations over tree automata,
alternating B\"uchi automata, and for checking language inclusion and
universality of finite word and tree automata.

Our LTS simulation algorithm is a relatively straightforward modification of
the algorithm by Ranzato and Tapparo from \cite{ranzato:new} (referred to as RT
in the following) for computing simulations over Kripke structures (a Kripke
structure associate labels with states while an LTS attaches labels to
transitions). Given a Kripke structure $\K$ with a set of states $Q$ and a
transition relation $\delta$ such that $P_\Sim$ is the partition of $Q$
according to simulation equivalence, RT runs in time $\O(|P_\Sim||\delta|)$
and space $\O(|P_\Sim||Q|)$. 
This algorithm refines the algorithm \cite{henzinger:computing} by Henzinger,
Henzinger, and Kopke (referred to as HHK) with running time $\O(|Q||\delta|)$
and space $\O(|Q|^2)$. The main difference between HHK and RT is that instead
of manipulating individual states, RT works on the level of iteratively refined
equivalence classes of a relation that finally converges to simulation
equivalence. We have chosen RT since it is the fastest known simulation
algorithm. However, there are other algorithms that are slower but more space
efficient. The algorithm with the lowest space complexity among all known
simulation algorithms is the one by Gentiliny, Piazza, and Policriti
\cite{gentiliny:fromBisimulation}. It runs in time $\O(|P_\Sim|^2|\delta|)$ and
space $\O(|P_\Sim|^2+|Q|\log|P_\Sim|)$.  Then, there is a recent algorithm
\cite{crafa:saving} by  Crafa, Ranzato, and Tapparo, which improves on space
complexity of RT, reducing it to $\O(|P_\Sim||\eqn|)$, which is very close to
the space complexity of the algorithm by Gentiliny et al., however, the price
of this is a worse time complexity $\O(|P_\Sim||\delta| +
|P_\Sim|^2|\eqm|)$.  Here, $\eqn$ is a certain partition of the set of states of
$\K$ such that $|P_\Sim| \leq |\eqn| \leq |Q|$ and  $\eqm$ is a partition of
the set of transitions where $|\eqm|\leq |\delta|$.

In fact, any algorithm computing simulation over Kripke
structures can be used for computing simulations on labelled transition systems.  Every LTS $\T$ with $n$ states and $m$ transitions can be easily
translated into a Kripke structure $\K_\T$ with $m+n$ states and $2m$
transitions (we turn every transition $q\ltr a r$ of $\T$ into the two
transitions $q\rightarrow (q,a,r) \rightarrow r$ where $(q,a,r)$ is a new
state with label $a$) such that the simulation on states of $\K_\T$ directly
gives simulation on $\T$. However, observe that this increase in the number of
states significantly affects complexity of the overall procedure. In the case
of RT, the time and space complexity of computing simulation on $\T$ this way
(running RT on $\K_\T$) would be almost the square of $m$, which is much worse
than for Kripke structures.

We design our version of RT that runs
directly on an LTS to eliminate this increase of complexity. This basically requires augmenting most of the data structures
of RT by alphabet symbols and iterating certain subprocedures for all
incoming/outgoing symbols of a state or a set of states.  We obtain an
algorithm that runs in time $\O(|P_\Sim||Q| + |\Sigma||P_\Sim||\delta|)$ and space $\O(|\Sigma|
|P_\Sim||Q|)$ where $\Sigma$ is the alphabet. The modifications of RT are
rather easy, nevertheless, notice that the dominating factor $|P_\Sim||\delta|$ of the
time complexity formula is not multiplied by the size of the alphabet, which
requires a sensitive approach when manipulating certain data structures. Apart
from that,  we provide a more straightforward (and abstract
interpretation free) proof of correctness of the algorithm than the one in \cite{ranzato:new}.

We also note that in \cite{holik:optimizing}, we present an improved version of
our LTS simulation algorithm where we to a large degree eliminate the
multiplicative effect of the size of the alphabet in the complexity formulas.
This algorithm can even turn nonuniformity of input and output symbols of
states into an advantage.  However, since the improvements described in
\cite{holik:optimizing} are not essential for the rest of this work and are
rather technical, we present only the original simpler version of the algorithm
here.

\section{Preliminaries}


We first introduce some additional notation used within
the chapter and the notion of partition-relation pair. 

Given an LTS {$\T=(\labset,\ltsq,\tr)$}, we
define the set of \emph{$a$-predecessors} of a state $r$ as $\pre a(r) = \{q\in
\ltsq\mid q\xtr a r\}$. Given $X,Y\subseteq \ltsq$, we use $\pre a(X)$ to denote the
set $\bigcup_{q\in X}\pre a(q)$, we write $q\xtr a X$ iff $q\in\pre a(X)$, and
$Y\xtr a X$ iff $Y\cap\pre a(X)\neq\emptyset$.

\paragraph{Partition-Relation Pairs.}  
A \emph{partition-relation pair} over a set $X$ is a pair $\langle P,Rel
\rangle$ where (1)~$P \subseteq 2^X$ is a partition of $X$ (i.e., $X =
\bigcup_{B \in P} B$, and for all $B,C \in P$, if $B\neq C$, then $B \cap C =
\emptyset$), and (2) $Rel \subseteq P \times P$. We say that a
partition-relation pair $\langle P,Rel \rangle$ over $X$ \emph{induces} (or
defines) the relation $\ind{\langle P,Rel \rangle} = {\bigcup_{(B,C)\in Rel} B
\times C}$. 

A partition-relation pair $\langle P,\Rel\rangle$ over $X$ inducing a relation
$R$ is the \emph{coarsest} iff there is no other partition-relation pair inducing
$R$ with the partition coarser than $P$. This means that $P = \{\{y\in X\mid
R(x) = R(y) \wedge R^{-1}(x) = R^{-1}(y)\}\mid x\in X\}$---two elements of $X$
are in the same block of $P$ iff they are related by $R$ with 
elements of $X$ in the same way. Notice that in the case when $R$ is a
preorder, $P$ is the set of equivalence classes of $R\cap R^{-1}$ and $\Rel$ is a partial order.

\section{The LTS Simulation Algorithm}
\label{lts:section} 
We now describe an algorithm to compute simulation over LTS.
For the rest of this chapter, we assume that we are given an LTS
$\T=(\Sigma,Q,\tr)$ and the coarsest partition-relation pair $\langle
P_\init,\Rel_\init \rangle$ inducing an initial preorder $I \subseteq Q \times
Q$.  Our algorithm takes $\T$ and $\langle P_\init,\Rel_\init \rangle$ as the
input and outputs the coarsest partition-relation pair $\langle
P_\Sim,\Rel_\Sim\rangle$ inducing the simulation preorder $\wordsim^I$ on $\T$
included in $\init$. Algorithm~\ref{algorithm:LRT} describes the algorithm in
pseudocode. 
Before we discuss it in detail and analyse its correctness and
complexity, we give a brief outline.


The algorithm propagates the
negative information about which pair of states are not related by simulation.
It iteratively refines a partition-relation pair $\langle P,\Rel \rangle$
(strengthening the induced relation) initialised as $\langle
P_{\init},\Rel_{\init} \rangle$.
The induced relation is always superset of the target simulation, the states
belonging to a block $B\in P$ are those which are currently assumed as being
possibly simulated by states from $\bigcup \Rel(B)$.
%
%
When the algorithm terminates, $\langle P,\Rel \rangle$ equals $\langle
P_{\Sim},\Rel_{\Sim} \rangle$.

The pair $\langle P,\Rel
\rangle$ is refined by splitting the blocks of the partition
in $P$ and pruning the relation $\Rel$.
For this purpose, the
algorithm maintains a set $\remove_a(B)$ for each $a \in \labset$ and $B\in P$. 
$\remove_a(B)$ contains states that was recently identified as not having an $a$-transition leading into $\bigcup\Rel(B)$.
Clearly, a state in $\remove_a(B)$ cannot
simulate states that have an $a$-transition going into $B$. 
Therefore, for a set $\remove_a(B) \neq \emptyset$ chosen at the beginning of an iteration, the algorithm splits each block $C \in P$ to 
$C \cap \remove_a(B)$ and $C \setminus \remove_a(B)$ (states not capable and states possibly capable of simulating states from $\pre a(B)$).
This is done using the function $\Split$ on~line~6.\ \ \ 

After performing  the $\Split$ operation, we update the relation $\Rel$ and the
$\remove$ sets.
This is carried out in two steps.
First, the algorithm refines the values of $\Rel$ and $\remove$ to be consistent with the new value of the partition $P$ refined by the $\Split$.
All $\Rel$ relations between the original ``parent'' blocks of states are
inherited to their ``children'' blocks into which the parents were split (line
8)---the notation $\parent_{P_\prev}(C)$ refers to the parent block of which $C$ was a part before the $\Split$. 
On line 10, the $\remove$ sets are inherited from parent blocks to their
children.
In the second step, the algorithm performs the actual refinement of the
relation induced by $\langle P,\Rel\rangle$.
On line 14, $\Rel$ is being pruned to reflect that states that
have an $a$-transition going into $B$ cannot be simulated by
states which do not have an $a$-transition going into $\bigcup\Rel(B)$. This is done by removing the relation between blocks included in $\remove_a(B)$ and blocks with states leading to $B$ via $a$. 
Refinement of $\Rel$ is then propagated further to $\remove$ sets.  Removing a
pair of blocks $(C,D)$ from $\Rel$ may cause that a state that has a
$b$-transition into $D$ (therefore, it had a $b$-transition into
$\bigcup\Rel(C)$ before removing $(C,D)$ from $\Rel$) now does not have any $b$-transition into
$\bigcup\Rel(C)$. Such a state is freshly identified as not being
capable of simulating states from $\pre a(C)$. We add it into $\remove_b(C)$
on line 17, which ensures propagation of the negative information.
\begin{algorithm}[t]

\caption{Computing simulation on an LTS}
\label{algorithm:LRT}

\KwIn{An LTS $\T=(Q,\labset,\tr)$, the coarsest partition-relation pair
$\langle P_{\init},\Rel_{\init} \rangle$ on $Q$ inducing a preorder $I\subseteq
Q\times Q$.}

\KwData{A partition-relation pair $\langle P,\Rel \rangle$ on $Q$, and for each
$B\in P$ and $a\in\labset$, a set $\remove_a(B)\subseteq Q$.}

\KwOut{The coarsest partition-relation pair $\langle P_\Sim,\Rel_\Sim \rangle$
inducing $\wordsim^I$.}

\BlankLine
\tcc{initialisation}

$\langle P,\Rel \rangle \ass \langle P_{\init},\Rel_{\init} \rangle$\;

\lForAll{$a\in\labset,B\in P$}{
        $\remove_a(B)\ass Q\setminus\pre a( \bigcup\Rel(B))$}\;

\BlankLine
\tcc{computation}
\While{$\exists a\in\labset.\ \exists B\in P.\  \remove_a(B)\neq\emptyset$}{
        $\remove\ass\remove_a(B);\remove_a(B)\ass\emptyset$\;
        $P_{\prev}\ass P; B_\prev \ass B
        ;\Rel_\prev\ass\Rel
        $\;
        $P\ass\Split(P,\remove)$\;
        \ForAll{$C\in P$}{
                $\Rel(C)\ass\{D\in P\mid D\subseteq\bigcup\Rel_\prev(\parent_{P_\prev}(C))\}$\;
                        \ForAll{$b\in\labset$}{
                                $\remove_b(C)\ass \remove_b(\parent_{P_\prev}(C))$
                        }
        }
        \ForAll{$C\in P.\  C\xtr a B_\prev$}{
                \ForAll{$D\in P.\ D\subseteq \remove$}{
                \If{$(C,D)\in\Rel$}{
                        $\Rel\ass\Rel\setminus\{(C,D)\}$\;
                        \ForAll{$b\in\labset$}{
                                \ForAll{$r\in\pre b(D)$ \KwSty{such that} $r\not \in \pre b( \bigcup\Rel(C))$}{
                                        $\remove_b(C)\ass\remove_b(C)\cup\{r\}$
                                }
                        }
                }
        }
        }
}
\Return $\langle P,\Rel \rangle$\;
\end{algorithm}

\subsection{Correctness of the Algorithm}

The correctness of the algorithm is formalised in Theorem
\ref{theorem:LRT:correctness}.  A similar correctness result is proved in
\cite{ranzato:new} for the algorithm on Kripke structures, using notions from
the theory of abstract interpretation. We provide here an alternative, more
direct proof.  

We will prove termination and partial correctness, this is, that (1)
the final partition-relation pair that we denote $\langle P_\fin,\Rel_\fin
\rangle$ induces $\wordsim^I$;
and (2) that $\langle P_{\fin},\Rel_{\fin}\rangle$ is also the
coarsest. The two points together give $\langle P_{\fin},\Rel_{\fin}\rangle=\langle P_\Sim,\Rel_\Sim \rangle$.

\begin{theorem} 
\label{theorem:LRT:correctness} 
Algorithm~\ref{algorithm:LRT} terminates and returns the partition-relation
pair $\langle P_{\Sim},\Rel_{\Sim} \rangle$.  
\end{theorem}

Let us first introduce some notation that will be needed within the proof of
the theorem. By an \emph{iteration}, we will mean a single iteration of the
while loop of the algorithm.  For an iteration, the block $B$ chosen on line 3
(also referred to as $B_\prev$) will be denoted as the \emph{pivot} of the
iteration. An \emph{ancestor} of a block $C$ is any block $D$ which appears
during the computation and for which $C\subseteq D$, and on the contrary, $C$
is a \emph{descendant} of $D$. Moreover, if $D$ is the immediate ancestor of
$C$ such that $C$ was created while splitting $D$, then $D$ is the
\emph{parent} of $C$ and $C$ is a \emph{child} of $D$.  We will denote by
$q\nltr a r$ the fact that $\neg (q\xtr a r)$. Moreover, for any $B,C\subseteq
Q$, $q\nltr a C$ and $B\xtr a C$ are defined analogously, i.e. provided that
$q\not\in\pre a(C)$ and $B\cap\pre a(C)=\emptyset$.
We will use $\ind{\langle P,\Rel \rangle}$ to denote the relation induced by the partition-relation   pair $\langle P,\Rel \rangle$ in a particular state of a run of the algorithm.

\begin{lemma}\label{shrinking}
On line 3 of Algorithm~\ref{algorithm:LRT}, the pair $\langle P,\Rel \rangle$
is always a partition-relation pair. The partition $P$ can only be refined
during the computation. Moreover, the relation 
$\ind{\langle P,\Rel \rangle}$ is monotonically getting smaller during the
computation.\end{lemma}

\begin{proof} 
The initial value of $\langle P,\Rel\rangle$ is clearly a partition-relation
pair.  After $\Split$ on line 6, $\langle P,\Rel\rangle$ is temporarily not a
partition-relation pair as $\Rel$ is a relation on $P_\prev$, not on $P$.
However, after inheriting all $\Rel$ links of parent blocks by their children
on lines 7--10, $\langle P,\Rel\rangle$ is a partition-relation pair again. The
other two claims of the lemma are also immediate as the algorithm can only
split the classes of $P$ (but never unites them), and can only remove
elements from $\Rel$.\end{proof}


\begin{lemma}\label{invariants}
The following claims are invariants of the while loop (of line 3) of Algorithm~\ref{algorithm:LRT}:
\begin{gather}
\label{remove}\forall B\in P.\ \forall a\in\labset.\
\remove_a(B)\nltr a\bigcup\Rel(B)\\
\label{refl}\forall B\in P.\ B\in\Rel(B)
\end{gather}
\begin{multline}\label{sim}
\forall B,C\in P.\ (B,C)\in\Rel \implies \\
\left(\forall a\in\labset.\ \forall D\in P.\ B\xtr a D \implies
C\subseteq\pre a(\bigcup\Rel(D))\cup\remove_a(D)\right)
\end{multline}
\end{lemma}

\begin{proof} After the initialisation, all the invariants hold. It is
immediate for Invariants~\ref{remove} and \ref{refl}. It is also fairly obvious
for Invariant~\ref{sim}, as after the algorithm passes line 2, for all $q\in
Q,a\in\Sigma,D\in P$, it holds that either $q$ has an $a$ transition leading to
$\bigcup\Rel(D)$ or $q$ is in $\remove_a(D)$.   

\begin{itemize} 
\item Invariant \eqref{remove} can never be broken. After the initialisation it
holds. From there on, it holds because only such a state $r$ can be moved into
the $\remove_b(C)$ which is not in $\pre b(\bigcup\Rel(C))$  (the test on line
16). Moreover, if $r$ is once not in $\pre b(\bigcup\Rel(C))$, then it will
never be there from that moment on (by Lemma~\ref{shrinking}).

\item Invariant \eqref{refl} can never be broken as violating reflexivity of
$\Rel$ requires choosing a pair $(C,D)$ on line 14 such that $C=D$. The $(C,D)$ pair can be chosen on line 14 only if $C\xtr a B$ and $D\subseteq\remove_a(B)$ where
$B$ is the pivot block. Thanks to Invariant \eqref{remove}, this is not
possible for $C=D$. 

\item Invariant \eqref{sim} can be temporarily broken on three places of the
algorithm:  \begin{description}

  \item[lines 6--10:] Let $C$  be a block of $P$ on line 7 and let $C'\in
P_\prev$ be its parent. Then it is easy to see that after finishing the for
loop on line 7, it holds that $\bigcup\Rel(C) = \bigcup\Rel_\prev(C')$ and for
all $a\in\labset$, $\remove_a(C)=\remove_a(C')$. Thus, after finishing the for
loop on line 7, Invariant \eqref{sim} can be broken only for those $(B,C)$
pairs such that it was broken even for their parents on line 6. Therefore, if the invariant holds on line 3, then it also holds after returning from the for loop on line 7.

  \item[line 4:] Assume the invariant holds at the beginning of some iteration and is then violated by emptying the $\remove_a(B)$ set on line 4. Then, there are $C,D\in
P$ which break the invariant and for which it holds that $(C,D)\in \Rel$, $C\xtr a B$,
$D\subseteq\pre a (\bigcup\Rel(B))\cup\remove_a(B)$, and $D\nsubseteq\pre
a(\bigcup\Rel(B))$. The $\Split$ operation on line 6 divides $D$ into
$D_1\subseteq\pre a(\bigcup\Rel(B))$ and $D_2\subseteq\remove$. After that, $\Rel$ and
the $\remove$ sets are inherited on lines 7--10. Now only those $(C',D_2)$ pairs
break the invariant where $C'$ is a child of $C$ such that it leads via $a$ into a
child of $B$. But exactly these pairs will be chosen on line 13 within this iteration for pruning $\Rel$. Hence, after finishing the iteration, the invariant will not be violated from the reason of emptying $\remove_a(B)$.

  \item[line 16:] Pruning $\Rel$ on line 14 lead to breaking the invariant as
there may states $r$ such that $r\xtr b D$ and thus before the update of
$\Rel$, $r\xtr b\bigcup\Rel(C)$, but after the removal of $D$ from $\Rel(C)$,
it can happen that $r\nltr b\bigcup\Rel(C)$. However, exactly these $r$ states
are moved into $\remove_b(C)$, and so Invariant~\eqref{sim} is restored after
finishing the for loop on line 13.  \end{description} \end{itemize} 
\vspace{-5mm}
\end{proof}


\begin{lemma}\label{issimulation} If all the $\remove$ sets are empty when evaluating the condition on line 3, then
$\ind{\langle P,\Rel \rangle}$ is a simulation on $\T$ included in $\init$.
\end{lemma}

\begin{proof} By Lemma~\ref{shrinking}, it is clear that $\ind{\langle
P,\Rel \rangle}$ is always a subset of $I$.
We have to show that $\ind{\langle P,\Rel
\rangle}$ is also a simulation on $\T$. Let $q\ind{\langle P,\Rel \rangle} r$ for some $q\in B,r\in C$ where $B,C\in P$.  From the definition of $\ind{\langle
P,\Rel \rangle}$, $(B,C)\in\Rel$. Let $q\xtr a s$ for some $s\in D,D\in P$. Then $B\xtr a D$.
Therefore, by Invariant \eqref{sim} and since all the $\remove$
sets are empty, we get $C\subseteq\pre a(\bigcup\Rel(D))$. This means that
there is $u\in \bigcup\Rel(D)$ such that $r\xtr a u$.
By the definition of $\ind{\langle P,\Rel \rangle}$, we have $s\ind{\langle P,\Rel \rangle} u$. Therefore, $\ind{\langle P,\Rel \rangle}$ is a simulation on $\T$ and the lemma holds. \end{proof}

During the computation, the relation $\Rel$ is not necessarily always
transitive. We can prove only the following property of $\Rel$ that roughly
resembles transitivity, and which is crucial for our correctness proof.

\begin{lemma}\label{transitivitystub}  
Under the assumption that ${\wordsim^I}\subseteq{\ind{\langle P,\Rel
\rangle}}$, the following invariant always holds on line 3 of
Algorithm~\ref{algorithm:LRT}: For any  $q,r\in Q$ with $q\wordsim^I r$ and
$B,C,D\in P$ such that $q\in C$, $r\in D$ and $(B,C)\in\Rel$, it holds that also
$(B,D)\in\Rel$.
\end{lemma}

\begin{proof} 
Let us recall the relationship between a partition-relation pair
$\langle P,\Rel \rangle$ and its induced relation $\ind{\langle P,\Rel\rangle}$
which is: For any $B,C\in P$ and $q\in B,r\in C$, it holds that $q \ind{\langle
P,\Rel\rangle} r$ iff $(B,C)\in\Rel$. Therefore, if ${\ind{\langle
P,\Rel\rangle}}\subseteq{\wordsim^I}$, then $q\wordsim^I r$ implies $(B,C)\in
\Rel$.
We prove the lemma by induction on the number of iterations of the while loop.

The base case: After the initialisation, the claim holds as $\Rel_{\init}$ is
transitive (the relation ${\init}$ is a preorder). We prove the induction step
by contradiction. 

Let the lemma be broken for the first time at the beginning of the $i$-th
iteration of the while loop. We use $\Start_{i}$ to denote the state of the
algorithm at this moment.
At $\Start_{i}$, we have that ${\wordsim^I}\subseteq{\ind{\langle
P,\Rel\rangle}}$ and there are some $B,C,D\in P$, $q\in C$, and $r\in D$ such that
$q\wordsim^I r$,
$(B,C)\in\Rel$, and
$(B,D)\not\in\Rel$. From $q\wordsim^I r$ and
${\wordsim^I}\subseteq{\ind{\langle P,\Rel\rangle}}$, we have $(C,D)\in\Rel$.
Because the induced relation is shrinking only (Lemma~\ref{shrinking}), we have
that at each moment of the computation preceding $\Start_i$, the relation
$\wordsim^I$ was a subset of the relation induced by the current partition-relation pair, the ancestor $C'$ of $C$ was
above the ancestor $B'$ of $B$ wrt. the current $\Rel$, and also the ancestor
of $D$ was above the ancestor of $C$. Because of this
and as the lemma is broken for
the first time at $\Start_i$, we know that at the beginning of any iteration prior to the $i$-th one, the ancestor of $D$ was above the ancestor of $B$ wrt. the current state of $\Rel$. 

Let us analyse the moment before $\Start_i$ when $(B,D)$ is going to be removed
from relation this is, we are within the $i-1$-th, just before entering the for
loop on line 11). 
Let $\langle P,\Rel'\rangle$ be the current
partition-relation pair (the current partition $P$ at that moment is
the same as at $\Start_i$, since no splitting will be done until $\Start_i$). 
The situation is such that $(B,C)\in\Rel'$,
$(C,D)\in\Rel'$, $(B,D)\in\Rel'$, and we are going to remove $(B,D)$ from
$\Rel'$ on line 14. However, we keep $(B,C)$ and $(C,D)$ in $\Rel'$ during this
iteration as these two pairs will be in $\Rel$ at $\Start_i$.
Removing $(B,D)$ from $\Rel'$ is caused by processing the $\remove_a(E)$
set where $E\in P_\prev$ is the pivot of the $i-1$-th iteration. Thus, we have that $B\xtr a E, D\subseteq\remove_a(E)$ and $C\cap\remove_a(E)=\emptyset$. 

Let us examine the state of the algorithm at the beginning of the $i-1$-th
iteration, the moment referred to as $\Start_{i-1}$. The current partition relation pair at $\Start_{i-1}$ is $\langle P_\prev,\Rel'_\prev\rangle$.
It holds that ${\wordsim^I}\subseteq{\ind{\langle P_\prev,\Rel'_\prev\rangle}}$.
Let $B',C',D'\in P_\prev$ be the ancestors of $B,C,D$
(therefore $B\subseteq B', C\subseteq C', D\subseteq D'$). We have that $q\in
C\subseteq C'$, $C\cap\remove_a(E)=\emptyset$, $B'\xtr a E$, and $(B',D')\in
\Rel'_\prev$, and therefore, from Invariant \eqref{sim}, we have that
$C'\subseteq \pre a(\bigcup\Rel(E))\cup\remove_a(E)$. This implies that
$C\subseteq \pre a(\bigcup\Rel(E))$.  Thus, there is $F\in \Rel'_\prev(E)$ and $q'\in F$ with $q\xtr a q'$. 
Since $q\wordsim^I r$, there is $r'\in Q$ with $r\xtr a r'$ and $q'\wordsim^I r'$.
Because ${\wordsim^I}\subseteq{\ind{\langle P_\prev,\Rel_\prev\rangle}}$,
the block $G\in P_\prev$ containing $r'$ must be in $\Rel'_\prev(F)$. Finally, because $r\in D\subseteq\remove_b(E)$,
from Invariant \eqref{remove}, we get $(E,G)\not\in\Rel'_\prev$.  

To conclude the proof, observe that the
states $q',r'$, the blocks $E,F,G\in P_\prev$, and the partition-relation pair
$\langle P_\prev,\Rel'_\prev \rangle$ form a situation
that violates the lemma at $\Start_{i-1}$ (to recap, we have that
${\wordsim^I}\subseteq{\ind{\langle P_\prev,\Rel'_\prev\rangle}}$,  $q'\in
F,r'\in G,q'\wordsim^I r'$, and $(E,F)\in \Rel'_\prev$, but $(E,G)\not\in\Rel'_\prev$). This is a contradiction since $\Start_i$
was
supposed to be the first such a  
moment.\end{proof}

\begin{lemma}\label{max} At any point of a run of Algorithm~\ref{algorithm:LRT}, ${\wordsim^I}\subseteq{\ind{\langle P,\Rel \rangle}}$. \end{lemma}

\begin{proof} 
The lemma apparently holds after initialisation.  We will prove that it always
holds by contradiction---we will show that violating this lemma in a run of
Algorithm~\ref{algorithm:LRT} has to be preceded by breaking Lemma
\ref{transitivitystub}. 

Let us choose the moment just before the lemma is violated for the first time.
This is, some $(B,C)$ is going to be removed from $\Rel$
on line 14 such that there are $q\in B$ and $r\in C$ with $q\wordsim^I r$. 
This update of $\Rel$ is caused by processing the set $\remove_a(D)$ where $D\in P_\prev$ is the pivot of the
current iteration of the while loop, $B\xtr a D$, $B\cap\remove=\emptyset$ ($\remove$ is the recorder value of $\remove_a(D)$ which was emptied on line 4 in this iteration), and $C\subseteq\remove$. Let $B',C'\in P_\prev$ be the ancestors
of $B,C$. From Invariant \eqref{refl}, we have that $(B',B')\in\Rel_\prev$. 

Let us examine the state at the beginning of this iteration.
We have that $B'\xtr a D$ because of $B\xtr a D$, which by Invariant
\eqref{sim} gives $B'\subseteq\pre a(\bigcup\Rel_\prev(D))\cup\remove_a(D)$.
Since $q\in B$, $q\not\in\remove_a(D)$, and therefore there are $E\in \Rel_\prev(D)$ and $q'\in E$ with
$q\xtr a q'$. From $q\wordsim^I r$ and from the fact that ${\wordsim^I}$ is a subset of the current induced relation
(the lemma is going to be broken for the first time, it holds so far), we have
that there are $F\in \Rel_\prev(E)$ and $r'\in F$ with $r\xtr a r'$.
However, as $r\in\remove_a(D)$ and because of Invariant \eqref{remove}, we have
$(D,F)\not\in\Rel_\prev$. Hence the states $q',r'$ and the blocks $D,E,F$
violates Lemma~\ref{transitivitystub} at the beginning of this iteration. 
\end{proof}

\begin{lemma}\label{partition} 
At any point of a run of Algorithm~\ref{algorithm:LRT}, any two states $q,r\in Q$ with $q\wordsimequiv^I r$ are in the
same block of~$P$.\end{lemma}

\begin{proof} By contradiction. We will show that breaking this lemma in a run of Algorithm~\ref{algorithm:LRT} has to
be preceded by breaking Lemma~\ref{transitivitystub}.

After the initialisation the lemma holds. Let us choose the first moment when
it is broken. At that moment, states $q,r$ with $q\wordsimequiv^I r$ are separated
from each other by the
$\Split$ operation during processing of some pivot block $B$. Without loss of
generality, we assume that at the beginning of this iteration
$r\in\remove_a(B)$ and $q\not\in\remove_a(B)$. 

Consider now the moment within some of the preceding iterations, just
before entering the for loop on line 11 during which $r$ will be added into
$\remove_a(B')$ where $B'$ is an ancestor of $B$. Let the current
partition-relation pair be $\langle P,\Rel \rangle$, and let $q,r\in C,C\in P$.
There is some block $D\in \Rel(B')$ with $r\xtr a D$ such that
$(B',D)$ will be removed from $\Rel$ and $r$ will be added to $\remove_a(B')$ because of that within this iteration. 

Since sine $r\xtr
a D$, there is $r'\in D$ with $r\xtr a r'$.
From $r\wordsim^I q$, there is $q'\in Q$ with $q\xtr a q'$ and $r'\wordsim^I q'$, and since ${\wordsim^I}\subseteq{\langle P,\Rel\rangle}$
(Lemma~\ref{max}), there is $E\in\Rel(D)$ with $q' \in E$.
Moreover, from Lemma~\ref{transitivitystub} (whose claim holds also
just before entering the for loop on line 11 because lines 4--10 do not influence the
induced relation), $E\in\Rel(B')$. 

We have shown that when entering the for loop on line 11, $q\xtr a \bigcup
\Rel(B')$. Recall that $q$ will not be added into $\remove_a(B')$ during this
iteration. Therefore, it has to hold that $q\xtr a\bigcup \Rel(B')$ also after
finishing the for loop on line 11 (otherwise $q$ would be added into
$\remove_a(B')$).  This is, after finishing the for loop on line 11, there is
still some $E'\in \Rel(B')$ and $q''\in E'$ with $q\xtr a q''$.  Because
$q\wordsim^I r$, there is some $r''\in Q$ with $q''\wordsim^I r''$ and $r\xtr a
r''$. Since ${\wordsim^I}\subseteq{\langle P,\Rel\rangle}$, there is some $F\in
\Rel(E')$ with $r''\in F$.  But at the end of the for loop on line 11 (i.e. the
beginning of the next iteration of the while loop), $(B',F)\not\in\Rel$ as $r$
was be added into $\remove_a(B')$ within the for loop (Invariant
\eqref{remove}). To conclude the proof, observe now that at the beginning of
the next iteration of the while loop,  states $q'',r''$ and blocks $B',E',F$
form a situation contradicting Lemma~\ref{transitivitystub}.  \end{proof}


\begin{lemma}\label{disjunct} Let $B,B'$ be two blocks appearing during a run of Algorithm~\ref{algorithm:LRT} such that
$B'$ is an ancestor of $B$. Let $\remove_a(B)$ and $\remove_a(B')$ be two $\remove$ sets at the (different) moments when
$B$, resp. $B'$, is chosen as the pivot. Then, $\remove_a(B)\cap \remove_a(B')=\emptyset$. \end{lemma}

\begin{proof} If a state $q$ is in $\remove_a(B)$ after the initialisation,
then $q\nltr a \bigcup\Rel_{\init}(B)$. If $q$ is added into $\remove_a(B)$
later on line 17, then it means that the test on line 13 passed, so $q\xtr a
\bigcup\Rel(B)$ was true at that moment\footnote{Note that at that time, $B$ is
referred to via $C$ in the algorithm.}.  Subsequently, after the update of
$\Rel$ on line 14, $q\nltr a \bigcup\Rel(B)$.  From Lemma~\ref{shrinking}, if
once $q\nltr a\bigcup\Rel(B)$, then from that moment on it can never happen
that $q\xtr a\bigcup\Rel(B')$ where $B'$ is a descendant of $B$. It means
that $q$ will never be added to any $\remove_a(B')$ where $B'$ is a descendant
of $B$.
To summarise: when a pivot $B$ with nonempty $\remove_a(B)$ is chosen to be
processed on line 3, $\remove_a(B)$ is always emptied and none of the states
from $\remove_a(B)$ can be added to any $\remove_a(B')$ where $B'$ is a
descendant of $B$ again. Thus whenever later some descendant $B'$ of $B$ with $\remove_a(B')$ is
being processed, $\remove_a(B)\cap\remove_a(B') =\emptyset$.
\end{proof}

We are now ready to prove Theorem~\ref{theorem:LRT:correctness}. 
\begin{proof}[Proof of Theorem~\ref{theorem:LRT:correctness}] Due to Lemma~\ref{disjunct}, for any block $B$ which can arise during the computation, $B$ can be chosen as
a pivot only finitely many times as for any $a\in\labset$, all the $\remove_a(B)$ sets encountered on line 3 are disjoint.
There are finitely many possible blocks and hence the algorithm terminates.

Lemma~\ref{issimulation} implies that the relation $\ind{\langle P_\fin,\Rel_\fin\rangle}$ induced by the final partition-relation pair $\langle P_{\fin},
Rel_{\fin}\rangle$ is a simulation included in $\init$. Lemma~\ref{max} implies that this simulation is the maximal one. Finally,
Lemma~\ref{partition} implies that the resulting partition 
$P_{\fin}$ equals $\prt{Q}{\wordsimequiv^I}$ and thus $\langle P_\fin,\Rel_\fin\rangle = \langle P_\Sim,\Rel_\Sim\rangle$. \end{proof}

\subsection{Implementation and Complexity of the Algorithm}

The complexity of the algorithm is equal to that of the original algorithm from
\cite{ranzato:new}, up to the new factor $\labset$ that is not present in
\cite{ranzato:new} (or, equivalently, $|\labset| = 1$ in \cite{ranzato:new}).
The complexity analysis is based on the similar reasoning as the one in
\cite{ranzato:new}. Time complexity strongly depends on use of certain data
structures and on some particular implementation techniques that we describe
below along the analysis within the proof of Theorem~\ref{LTSSimCompl}.

\begin{theorem} \label{LTSSimCompl} 
Algorithm~\ref{algorithm:LRT} runs in time
$\O(\card{\labset}  \card{P_{\Sim}}  \card{Q}+\card{P_{\Sim}}
\card{\tr})$ and space $\O(\card{\labset}  \card{P_{\Sim}}
\card{Q})$.  
\end{theorem} 

\begin{proof} \ \vspace{-3mm}
\paragraph{Basic Data Structures.} 
We use resizable arrays (and matrices) which double (or quadruple) their size
whenever needed. The insertion operation over these structures takes amortised
constant (linear) time.

The input LTS is represented as a list of records about its states---we call
this representation as the \emph{state-list} representation of the LTS. The
record about each state $q \in Q$ contains a list of nonempty $\pre a(q)$
sets, each of them encoded as a list of its members (we use a list rather than an array having an entry for each $a
\in \labset$ in order to avoid a need to iterate over alphabet symbols for
which there is no transition).
The partition $P$ is encoded as a
doubly-linked list (DLL) of blocks. Each block is represented as a DLL of
(pointers to) states of the block.  The relation $\Rel$ is encoded as a Boolean matrix $P \times P$. 

Each block $B$ contains for  each $a\in\labset$ a list of (pointers on) states
from $\remove_a(B)$. Each time when any set $\remove_a(B)$ becomes nonempty,
block $B$ is moved to the beginning of the list of blocks. Choosing the pivot
block on line 3 then means just scanning the head of the list of blocks.

For each $a\in\labset$, a state $q\in Q$ and a block $B\in P$, we maintain a counter
$\relc_a(q,B)$. Its value within a run of the algorithm records cardinality of the set $\{r\in Q\mid {r\in\bigcup\Rel_a(B)}\wedge {q\xtr a r}\}$.
This counters allow us to test whether $r$ is in $\pre b(\bigcup \Rel(C))$ on
line 16 in constant time---we just ask whether $\relc_b(r,C) = 0$. The
counters are stored as an $P\times Q$ integer matrix per each $a\in\Sigma$.
The way of updating the counters during a computation will be described later.

We attach to each $q\in Q$ an array indexed by
symbols of $\labset$. A cell of the array indexed by $a\in\Sigma$ contains a reference the $\pre a(q)$ list. Using the arrays, we can access the $\pre a(q)$ list for given $a$ and $q$ in constant time (it would be
$\O(|\labset|)$ time without the arrays). 

\paragraph{Space Complexity.}
\noindent The arrays of pointers on the $\pre a$ lists take $\O(|\labset||Q|)$ space, the matrix encoding of $\Rel$ takes
$\O(\card{P_{\Sim}}^2)$ space, and the $\remove$ sets as well as the counters take $\O(|\labset||P_{\Sim}| |Q|)$ space. Thus
the overall asymptotic space complexity is $\O(|\labset||P_{\Sim}| |Q|)$.  

\paragraph{Time Complexity.}
We first introduce some auxiliary notation.  
For $B\subseteq Q$ and $a\in\labset$, we denote by
$\ine_a(B)$ the set $\{{(r,a,q)}\in{\tr}\mid q\in B\}$, and by $\ine(B)$ the
set $\bigcup_{a\in\labset}\ine_a(B)$. Note that $\card{\pre
a(B)}\leq\card{\ine_a(B)}$. We also denote by $\ltrset a$ the set of all
$a$-edges of $\tr$.  We use $Anc(B)$ to denote the set of all ancestors of $B$, including also $B$ itself.

We first analyse the initialisation phase of the algorithm preceding the main while
loop.  The initialisation of the arrays of pointers to the $\pre a$ lists takes
$\O(|\labset||Q|)$ time. The $\relc$ counters are initialised by (1) setting
all $\relc$ to $0$, and then (2) for all $B\in P$, for all $q\in B$, for all
$r\in\pre a(q)$, and for all $C$ such that $(C,B)\in\Rel$, incrementing
$\relc_a(r,C)$. This takes  $\O(|P_{\init}||\trset|)$ time. The $\remove$ sets
are initialised by iterating through all $a\in\labset,q\in Q,B\in P$ and checking whether $\relc_a(q,B)=0$. If so, then we add (append) $q$ to $\remove_a(B)$. This takes
$\O(|\labset| |P_{\init}||Q|)$ time. Overall, the initialisation can be
done in time $\O(|P_{\init}||\trset|+|\labset| |P_{\init}||Q|)$. 

The time complexity analysis of the while loop builds on Lemma~\ref{disjunct} and Lemma~\ref{shrinking} proved within the proof of
correctness of Algorithm~\ref{algorithm:LRT}. The two lemmas allow us to make the following observations:
\begin{description}
\item[Observation 1.] For any $a\in\labset$ and $B\in P_{\Sim}$, the sum of the
cardinalities of the $\remove_a(B')$ sets for all $B'\in Anc(B)$ that are chosen as the pivot is below $\card{Q}$. 
\item[Observation 2.]
If a pair $(C,D)$ once
appears on line 15, then no pair $(C',D')$ such that $C\in Anc(C')$ and $D\in
Anc(D')$ can appear on line 15 again. 
\end{description}

The $\Split(P,\remove)$ operation can be implemented in the following way: Iterate
through all $q\in\remove$. If $q\in B\in P$, add $q$ into a block $B_{child}$
(if $B_{child}$ does not exist yet, create it and add it into $P$) and remove
$q$ from $B$. If $B$ becomes empty, discard~it. This can be done in time
$\O(|\remove|)$. 
From Observation 1, we have that for a fixed block $B\in
P_{\Sim}$ and $a\in\labset$, the sum of cardinalities of all $\remove_a(B')$
sets with $B'\in Anc(B)$ according to which $\Split$ is being
done is below $|Q|$. Therefore, summed over all symbols of $\labset$ and all blocks
of $P_{\Sim}$, the overall time complexity of all $\Split$ operations is
$\O(|\labset||P_{\Sim}||Q|)$.  

The time complexity analysis of lines 7--10 is based on the fact that it can
happen at most $\card{P_{\init}}-\card{P_{\Sim}}$ times that any block $B$ is
split. Moreover, the presented code can be optimised by not having the lines
7--10 as a separate loop (this was chosen just for clarity of the
presentation), but the inheritance of $\Rel$, $\remove$, and the counters can
be done within the $\Split$ function, and only for those blocks that were
really split (not for all the blocks every time).
Whenever a new blocks is generated by $\Split$, we have to do the following:
(1)~For each $a\in\labset$, copy the $\remove_a$ set of the parent block and
attach the copy to the child block. As for all $a\in\labset,B\in P$,
$\remove_a(B)\subseteq Q$, and a new block will be generated at most
$\card{P_{\init}}-\card{P_{\Sim}}$ times, the overall time of this copying is
in $\O(|\labset||P_{\Sim}||Q|)$. (2)~Add a row and a column to the $\Rel$
matrix and copy the entries from those of the parent. This operation takes
$\O(|P_{\Sim}|)$ time for one added block as the size of the rows and columns
of the $\Rel$-matrix is bounded by $|P_{\Sim}|$. Thus. for all newly generated
blocks, we achieve the overall time complexity of $\O(\card{P_{\Sim}}^2)$.
(3)~Add and copy the $\relc$ counters. For one newly generated block, this
operation takes an $\O(|\labset||Q|)$ time and thus for all generated blocks,
it gives time $\O(|\labset||P_{\Sim}||Q|)$.

Lines 13 and 14 are $\O(1)$-time ($\Rel$ is a Boolean matrix). Before we enter
the for loop on line 11 with $B$ being the pivot, we compute a list
$RemoveList_a(B) = \{D\in P\mid D\subseteq \remove\}$. This is an
$\O(\card{\remove})$ operation and by almost the same argument as in the case
of the overall time complexity of $\Split$, we get also exactly the same
overall time complexity for computing all the $RemoveList_a(B)$ lists. On line
$11$, for each $q\in B$, we find the $\pre a(q)$ list (in $\O(1)$ time using
the array of pointers to the $\pre a(q)$ lists), and we iterate through all
elements of $\pre a(q)$ and choose every $C,C\xtr a \{q\}$.  This takes
$\O(|\ine_a(B)|)$ time. For any $B\in P_{\Sim}$, let $RL_a(B)$ be the set of
blocks $\bigcup_{B'\in Anc(B)}RemoveList_a(B')$. Then the overall time
complexity of lines 11--14 is at most $\O(\sum_{a\in\labset}\sum_{B\in
P_{\Sim}}|RL_a(B)||\ine_a(B)|)$. From the initial observations, we can see that
$\card{RL_a(B)}\leq\card{P_{\Sim}}$, and thus we have the overall time
complexity of
lines 11--14 in $\O(\sum_{a\in\labset}\sum_{B\in P_{\Sim}}|P_{\Sim}| |\ine_a(B)|)
= \O(\sum_{a\in\labset}|P_{\Sim}||\ltrset a|) = \O(|P_{\Sim}||\trset|)$.

Lines 15--17 are
implemented as follows. For a single pair $(C,D)$ appearing on line 14, we iterate through all $q\in D$ and through all
nonempty lists $\pre a(q)$, and for each $r\in \pre a(q)$, we decrement
$\relc_a(r,C)$. If $\relc_a(r,C)=0$ after the decrement, we append $r$ to the
$\remove_a(C)$ list. It follows from the initial observations that if any pair
of blocks $(C,D)$ once appears on line 14, then there will never appear any
pair of their descendants on line 14. Thus, if we fix a block $C\in P_{\Sim}$
and a state $q$, then it can happen at most once that a pair
$(C',D)$ such that $q\in D$ and $C'\in Anc(C)$ is being removed from $\Rel$.
on line 14. Thus, the contribution of the pair $C,q$
to the time complexity of lines 15--17 is $\O(\sum_{a\in\Sigma}|\pre a(q)|)$.
Therefore, the contribution of the $C,r$ pairs for all $r\in Q$ is
$\O(|\trset|)$, and hence the overall time complexity of lines 15--17 is
$\O(|P_{\Sim}||\trset|)$.

From the above analysis, it follows that the overall time complexity of the algorithm is $\O(|P_{\Sim}||\trset|+|\labset|
|P_{\Sim}||Q|)$. 
\end{proof}

\section{Conclusions and Future Work} 
We have presented a modification of the currently fastest algorithm RT
\cite{ranzato:new} for computing simulations over Kripke structures, which was
at the time of its publication the fastest algorithm for computing simulations
over LTS (the currently fastest algorithm is its optimised version from
\cite{holik:optimizing}). The algorithm has the time complexity
$\O(\card{\labset} \card{P_{\Sim}} \card{Q}+\card{P_{\Sim}} \card{\tr})$ and
the space complexity $\O(\card{\labset} \card{P_{\Sim}} \card{Q})$, which is
slightly worse than $\O(\card{P_{\Sim}} \card{\tr})$ time and
$\O(\card{P_{\Sim}} \card{Q})$ space of RT.  However, this complexity increase
can be to a large degree lowered as we show in \cite{holik:optimizing}. We have
also presented a proof of correctness that is more straightforward than the one
presented in \cite{ranzato:new}.

We plan to continue the research by the authors of \cite{ranzato:new} and
\cite{crafa:saving}. We have noticed that the algorithm from
\cite{crafa:saving} that refines RT goes in some sense against the spirit of
the original algorithm from \cite{henzinger:computing}, which is the main
reason of its worse time complexity. We believe that this problem can be circumvented and that it is possible to design an algorithm that matches both the time complexity of the fastest simulation algorithm \cite{ranzato:new} and space complexity of the most space efficient algorithm \cite{crafa:saving}.

\include{ta_reduction}
\chapter{Language Inclusion and Universality of Finite (Tree) Automata}
\label{chapter:fa_ta_inclusion}

The language inclusion problem for regular languages is important in many
application domains, e.g., formal verification. Many verification problems can
be formulated as a~language inclusion problem. For example, one may describe
the actual behaviours of an implementation in an automaton $\mathcal{A}$ and
all of the behaviours permitted by the specification in another automaton
$\mathcal{B}$.  Then, the problem of whether the implementation meets the
specification is equivalent to the problem $\lang{\mathcal{A}}\subseteq
\lang{\mathcal{B}}$.  Other applications include checking whether a fixpoint of
a symbolic automata-based incremental reachability computation was reached or
checking implication in automata-based decision procedures.  The universality
problem is a simpler variant of the language inclusion problem.  Even though
it is less useful in practice, it is important from the theoretical point of
view. A good solution for the universality problem often leads to a good
solution for language inclusion problem while the simpler setting of the former
problem makes the principles of the method easier to master. 

Methods for proving language inclusion can be categorised into two types: those
based on \emph{simulation} (e.g., \cite{dill:checking}) and those based on the
\emph{subset construction} (e.g.,
\cite{brzozowski:canonical,hopcroft:nlogn,meyer:equivalence,moller:dkbrics}).
Simulation-based approaches first compute a simulation relation on the states
of two automata $\mathcal{A}$ and $\mathcal{B}$ and then check whether all initial
states of $\mathcal{A}$ can be simulated by some initial state of
$\mathcal{B}$.  Since simulation can be computed in polynomial time,
simulation-based methods are usually very efficient. Their main drawback is
that they are incomplete since simulation implies language inclusion, but
not vice-versa.

On the other hand, methods based on the subset construction are complete but
inefficient because in many cases they will cause an exponential blow up in the
number of states. Recently, De Wulf et al. in~\cite{wulf:antichains} proposed the
\emph{antichain-based} approach for nondeterministic finite word automata.
To the best of our knowledge, it was the most efficient one among all of the
methods based on the subset construction.  Although the antichain-based method
significantly outperforms the classical subset construction, in many cases, it
(unavoidably) still sometimes suffers from the exponential blow up problem.

This chapter presents result that were published in two works, \cite{bouajjani:antichain} and \cite{abdulla:when}.
In \cite{bouajjani:antichain}, we generalise the
results on FA from \cite{wulf:antichains} also for tree automata and we show how a combination of the antichain-based
tree automata inclusion checking with the reduction techniques from
Chapter~\ref{chapter:ta_reduction} allows to greatly improve efficiency of
abstract regular tree model checking method.  In \cite{abdulla:when}, we present a new approach for both word
and tree automata universality and inclusion checking that nicely combines the
simulation-based and the antichain-based approaches. A computed simulation
relation is used for pruning out unnecessary search paths of the
antichain-based method and also to efficiently encode the stored state-space.
To distinguish the approaches from \cite{wulf:antichains,bouajjani:antichain}
from the one of \cite{abdulla:when}, we will refer to the former ones as to the
\emph{pure antichain approach} and to the latter ones as to  the \emph{simulation-enhanced antichain approach}.  
In this chapter, we describe mostly the results from  \cite{abdulla:when}, this is, the simulation enhanced antichain algorithms for FA and TA since the pure antichain TA algorithms that we present in \cite{bouajjani:antichain} can be seen as they simpler instances. As for experimental results, we present both the results from \cite{bouajjani:antichain} and from \cite{abdulla:when}. 

To simplify the presentation, we first consider the
problem of checking universality for a word automaton
 $\mathcal{A}$.
In a similar manner to the classical
subset construction, we start from the set of initial states
and search for sets of states
(here referred to as \emph{macro-states})
which are not accepting (i.e., we search for a counterexample of universality).
The key idea is to define an ``easy-to-check'' ordering $\preceq$  on the
states of $\mathcal{A}$ which implies language inclusion
(i.e.,  $p\preceq q$ implies that the language of the state $p$ is included in
the language of the state $q$).
From $\preceq$, we derive an ordering on macro-states
which we use in two ways
to optimise the subset construction:
(1) searching from a macro-state needs not continue in case
a smaller macro-state has already been analysed; and
(2) a given macro-state is represented by (the subset of) its maximal elements.
In this work, we take the ordering $\preceq$ to be the 
simulation preorder
on the automaton $\mathcal{A}$.
In fact, the antichain algorithms
of \cite{wulf:antichains} coincide  with the special case where
the ordering $\preceq$ is the identity
relation.
Subsequently, we describe how to generalise the above approach to the case of
checking language inclusion between two automata  $\mathcal{A}$ and
$\mathcal{B}$, by extending the ordering to pairs consisting of a state of
$\mathcal{A}$ and a macro-state of $\mathcal{B}$.

We then generalise our algorithms to the case of
tree automata.
We first formally define a notion of a language accepted from
tuples of states of the tree automaton as a set of contexts.
We identify here a new application of the upward simulation relation from
Chapter~\ref{chapter:ta_reduction}. We show that it implies (context)
language inclusion, and we describe how we can use it to optimise existing
algorithms for checking the universality and language inclusion properties.

We have implemented our algorithms and carried out an extensive
experimentation. Particularly, in \cite{bouajjani:antichain}, we compare
performance of the classical tree automata subset construction based algorithms
with the pure antichain-based algorithms (so far not using simulation
optimisations) developed in the spirit of \cite{wulf:antichains}.  We have
tested the algorithms on tree automata generated with a scale of different
settings of a random automata generator designed according to framework by
Tabakov and Vardi \cite{tabakov:experimental}.  We have also considered
tree-automata derived from intermediate steps of abstract regular tree model
checking.  The obtained results are consistent with the ones from
\cite{wulf:antichains} on FA and lead to a definite conclusion that the
antichain tree automata algorithms vastly outperform the classical ones.  Our
inclusion checking algorithms together with the reduction techniques from
Chapter~\ref{chapter:ta_reduction} also greatly improve the overall performance
of the abstract regular tree model checking method.

In \cite{abdulla:when}, we have carried out experiments comparing the pure antichain-based algorithms for both FA and TA with their simulation-improved
variants. In the case of FA, we obtained our experimental data from several
different sources. 
The experiments show that simulation enhanced antichain approach significantly outperforms the
pure antichain-based approach in almost all of the considered cases.

We note that simultaneously with \cite{abdulla:when},
Doyen and Raskin published their recent work \cite{doyen:antichain} where they
present basically the same main idea as is the one of \cite{abdulla:when} (this
is, using simulation to improve the antichain algorithms). However, even though
the two works have significant overlaps, both of them contain original unique
contributions. We will briefly compare the two
works in the following two paragraphs. 

Doyen and Raskin in \cite{doyen:antichain} study more systematically theoretical
aspects of simulation optimisations of antichain algorithms. They present a
framework where they consider also the backward algorithms for FA that were
presented in \cite{wulf:antichains} and show how they can be optimised with backward simulation. These backward algorithms are dual to the forward ones
and they utilise backward simulation instead of forward simulation. They also consider a
conceptually different approach where one utilises forward simulation within
backward algorithms and backward simulation within forward algorithms.  
Apart from that, Doyen and Raskin also show other applications of their
framework to problems such as emptiness of alternating automata. 

On the other hand, our paper \cite{abdulla:when} comes with the following.
Contrary to \cite{doyen:antichain}, we provide extensive experimental results
showing practical applicability of the algorithms. We also design algorithms
that are carefully optimised not to explore unnecessary search paths which also
notably improves their efficiency.  Then, except using simulations to prune
unpromising macro-states, we use them also to reduce the internal
representations of reached macro-states.  We study in detail both universality
and language inclusion problem (while Raskin and Doyen concentrate mostly only
on universality) where not all the optimisations that we propose are covered by
the framework from \cite{doyen:antichain} (in~particular, in~the~case of
inclusion checking, we utilise also simulation between states of the two
automata). Finally, we also present an extension of the technique to tree automata.

\paragraph{Outline.}
The remainder of the chapter is organised as follows.
We begin Section~\ref{section:nfa-ui} by applying our idea to solve the
universality problem for FA. The problem is simpler than the language
inclusion problem and thus we believe that presenting our universality checking
algorithm first makes it easier for the reader to grasp the idea. 
We continue the section by discussing our language inclusion checking algorithm
for FA. In Section~\ref{sec:TA-main}, we present the algorithms for checking
universality and language inclusion for tree automata that are extensions of the FA algorithms from Section~\ref{section:nfa-ui}. 
Section~\ref{SecUniExprm} describes experimental results
on comparing pure antichain-based algorithms for TA with the classical subset construction-based algorithms, and also experiments on testing impact of applying our algorithms in abstract regular tree model checking. In Section~\ref{sec:experiments}, we present experiments on comparing pure antichain-based algorithms for both FA and TA with their versions improved with simulations. 
Finally, in Section~\ref{sec:conclusion}, we conclude the chapter
and discuss further research directions.

\section{Universality and Language Inclusion of FA} \label{section:nfa-ui}
In this section, we describe our simulation improvements of the  antichain
algorithms for testing language inclusion and universality of FA from
\cite{wulf:antichains}. Basically, we will show how to utilise simulation on
states of an automaton (that is computed in advance) within a language
inclusion/universality checking algorithm.

Let $\A = (\Sigma, Q, \delta, I, F,)$ be a finite automaton.
For convenience, we call a set of states in $\mathcal{A}$ a \emph{macro-state}, i.e., a macro-state is a subset of $Q$. A macro-state is \emph{accepting} if it contains at least one accepting state, otherwise it is \emph{rejecting}. For a macro-state $P$, define $\lang{\mathcal{A}}(P) := \bigcup_{p\in P}\lang{\mathcal{A}}(p)$. We say that a macro-state $P$ is universal if $\lang{\mathcal{A}}(P)=\Sigma^*$. For two macro-states $P$ and $R$, we write $P \preceq^{\forall\exists} R$ as a shorthand for $\forall {p \in P}. \exists {r \in R}: {p \preceq r}$. We define the post-image of a macro-state $\mathit{Post}(P):=\{P'\mid \exists a\in\Sigma:P'=\{p'\mid \exists p\in P : (p,a,p')\in \delta\}\}$. We use $\rel{\mathcal{A}}$ to denote the set of relations over the states of $\mathcal{A}$ that imply language inclusion, i.e., if ${\preceq}\in \rel{\mathcal{A}}$, then we have ${p \preceq r} \implies {\lang{\mathcal{A}}(p) \subseteq \lang{\mathcal{A}}(r)}$.

Let $\mathcal{A} = (\Sigma,
Q_\mathcal{A}, \delta_\mathcal{A}, I_\mathcal{A}, F_\mathcal{A})$
and $\mathcal{B} = (\Sigma, Q_\mathcal{B}, \delta_\mathcal{B},
I_\mathcal{B}, F_\mathcal{B})$ be two FA. Define their union
automaton $\mathcal{A}\cup\mathcal{B}:=(\Sigma, Q_\mathcal{A}\cup
Q_\mathcal{B}, \delta_\mathcal{A} \cup \delta_\mathcal{B},
I_\mathcal{A}\cup I_\mathcal{B}, F_\mathcal{A}\cup F_\mathcal{B})$.

\subsection{Universality of FA}\label{sec:universality}

The \emph{universality problem} for an FA $\mathcal{A} = (\Sigma, Q, \delta,
I, F)$ is to decide whether $\lang{\mathcal{A}}=\Sigma^*$. The problem
is PSPACE-complete. The classical algorithm for the problem first determinises
$\mathcal{A}$ with the subset construction and then checks if every reachable
macro-state is accepting. The algorithm is inefficient since in many cases the
determinisation will cause a very fast growth in the number of states. Note
that for universality checking, we can stop the subset construction immediately
and conclude that $\mathcal{A}$ is not universal whenever a rejecting
macro-state is encountered. An example of a run of this algorithm is given in
Fig.~\ref{figure:subset}. The automaton $\mathcal{A}$ used in
Fig.~\ref{figure:subset} is universal because all reachable macro-states are
accepting.

\begin{figure}[t]
\begin{minipage}[b]{0.35\linewidth}
\begin{minipage}[b]{1\linewidth}
    \centering
    \scalebox{1}{
    \begin{tikzpicture}[>=latex',join=bevel]
    \tikzstyle{every state}=     [draw=blue!50,very thick,fill=blue!20, minimum size=5mm, initial text=, initial distance=3mm]%
    \node[state, initial, accepting] (s1) at (0,0) {\scriptsize$s_1$};
    \node[state, initial, accepting](s2) at (1,-1) {\scriptsize$s_2$};
    \node[state, accepting](s3) at (1.5,0) {\scriptsize$s_3$};
    \node[state](s4) at (2.5,-1) {\scriptsize$s_4$};
    \draw [->] (s1) to[bend right] node[auto] {\scriptsize $a$} (s2);
    \draw [->] (s1) to[loop above] node[auto] {\scriptsize $b$} (s1);
    \draw [->] (s1) to node[left = 4mm, above = 1mm] {\scriptsize $b$} (s4);
    \draw [->] (s2) to[loop below] node[auto] {\scriptsize $b$} (s2);
    \draw [->] (s2) to[bend right] node[below = 0.1mm] {\scriptsize a} (s3);
    \draw [->] (s3) to[loop above] node[auto] {\scriptsize $b$} (s3);
    \draw [->] (s3) to[bend right] node[auto] {\scriptsize $a$} (s1);
    \draw [->] (s4) to[bend right] node[auto] {\scriptsize $b$} (s3);
    \draw [->] (s4) to[bend left] node[auto] {\scriptsize $b$} (s2);
    \end{tikzpicture}}\\
    (a) Source FA $\mathcal{A}$\vspace*{2mm}
\end{minipage}
\begin{minipage}[b]{1\linewidth}
    \centering
    \scalebox{1.2}{
    \begin{tikzpicture}[>=latex',join=bevel,]
    \tikzstyle{mstate}=     [draw=green!75,thick,fill=green!30, shape=rectangle, initial text=]%
    \tikzstyle{fstate}=     [draw=green!50,thick, dashed,fill=green!15, shape=rectangle, initial text=]%

    \node[mstate,initial] (s1) at (0,0) {\scriptsize$s_1$};
    \node[fstate] (s2) at (-1,-1) {\scriptsize$s_2$};
    \node[fstate] (s3) at (1,-1) {\scriptsize$s_1$};

    \draw [->, thick] (s1) to node[left=2mm] {\scriptsize $a$} (s2);
    \draw [->, thick] (s1) to node[right=2mm] {\scriptsize $b$} (s3);
    \end{tikzpicture}
    \label{figure:subsetresult_op2}
    }\\
    (c)~Optimisation~1 and~2
\end{minipage}
\end{minipage}
\begin{minipage}[b]{0.7\linewidth}
\begin{minipage}[b]{0.9\linewidth}
    \centering
    \scalebox{0.8}{
    \begin{tikzpicture}[>=latex',join=bevel,]
    \tikzstyle{mstate}=     [draw=green!75,thick,fill=green!30, shape=rectangle, initial text=]%
    \tikzstyle{fstate}=     [draw=green!50,thick, dashed,fill=green!15, shape=rectangle, initial text=]%

    \draw[fill=blue!20, dashed, draw=blue!50] (-2,1.4) --(2,1.4).. controls (3,1.4) and (4,-4.5) ..(3.5,-4.5)--(-3.5,-4.5).. controls (-4,-4.5) and (-3,1.4) ..(-2,1.4);
    \node at (0,1.2) {\scriptsize Classical};

    \draw[fill=red!20, dashed, draw=red!50] (0,1) .. controls (-4,0.5) and (-2.5,-1) .. (-3,-2) .. controls (-3.1,-2.2) and (-3.1,-2.8) .. (-3,-3) .. controls (-3,-3.5) and (0,-3.5) .. (0,-3) .. controls (0, -1.5) and (0, -1.5) .. (2.2,-1.5).. controls (3,-0.5) and (3,1) ..(0,1);
    \node at (0,0.7) {\scriptsize Antichain};

    \draw[fill=yellow!20, dashed, draw=yellow!50] (-0.6,0.55) -- (0.6,0.55) .. controls (2.5, 0.55) and (2.5, -1.3) .. (2.2,-1.3) .. controls (2.2, -1.5) and (-2.2, -1.5) ..(-2.2,-1.3).. controls (-2.5, -1.3) and (-2.5, 0.55) ..(-0.6,0.55);
    \node at (0,0.35) {\scriptsize Optimisation 1};

    \node[mstate,initial] (s1) at (0,0) {\scriptsize$s_1, s_2$};
    \node[mstate] (s2) at (-1.5,-1) {\scriptsize$s_2, s_3$};
    \node[mstate] (s3) at (1.5,-1) {\scriptsize$s_1,s_2,s_4$};
    \node[mstate] (s4) at (-2.3,-2) {\scriptsize$s_1,s_3$};
    \node[fstate] (s5) at (-0.7,-2) {\scriptsize$s_2,s_3$};
    \node[fstate] (s6) at (0.7,-2) {\scriptsize$s_2,s_3$};
    \node[mstate] (s7) at (2.3,-2) {\scriptsize$s_1,s_2,s_3,s_4$};
    \node[fstate] (s8) at (-2.3,-3) {\scriptsize$s_1,s_2$};
    \node[fstate] (s9) at (-0.7,-3) {\scriptsize$s_1,s_3$};
    \node[mstate] (s10) at (0.7,-3) {\scriptsize$s_1,s_2,s_3$};
    \node[fstate] (s11) at (2.5,-3) {\scriptsize$s_1,s_2,s_3,s_4$};
    \node[fstate] (s12) at (0.7,-4) {\scriptsize$s_1,s_2,s_3$};
    \node[fstate] (s13) at (2.5,-4) {\scriptsize$s_1,s_2,s_3,s_4$};

    \draw [->, thick] (s1) to node[right=2mm] {\scriptsize $a$} (s2);
    \draw [->, thick] (s1) to node[left=2mm] {\scriptsize $b$} (s3);
    \draw [->, thick] (s2) to node[right] {\scriptsize $a$} (s4);
    \draw [->, thick] (s2) to node[left] {\scriptsize $b$} (s5);
    \draw [->, thick] (s3) to node[right] {\scriptsize $a$} (s6);
    \draw [->, thick] (s3) to node[left] {\scriptsize $b$} (s7);
    \draw [->, thick] (s4) to node[right] {\scriptsize $a$} (s8);
    \draw [->, thick] (s4) to node[left] {\scriptsize $b$} (s9);
    \draw [->, thick] (s7) to node[right=1mm] {\scriptsize $a$} (s10);
    \draw [->, thick] (s7) to node[left] {\scriptsize $b$} (s11);
    \draw [->, thick] (s10) to node[right=1mm] {\scriptsize $a$} (s12);
    \draw [->, thick] (s10) to node[left] {\scriptsize $b$} (s13);
    \end{tikzpicture}
    \label{figure:subsetresult}
    }

\end{minipage}
\begin{minipage}[b]{0.9\linewidth}
(b) A run of the algorithms. The areas labelled ``Optimisation~1'', ``Antichain'', ``Classical'' are the macro-states generated by our simulation enhanced antichain approach with the maximal simulation and Optimisation 1, the antichain-based approach, and the classical approach, respectively.
\end{minipage}
\end{minipage}
  \caption{Universality Checking Algorithms}
  \label{figure:subset}
\end{figure}
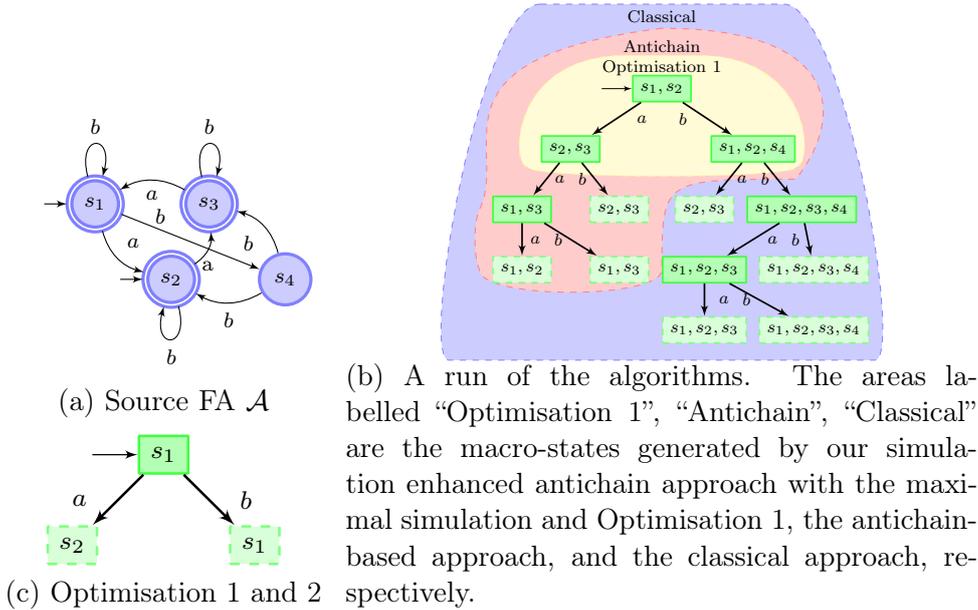

In this section, we propose a more efficient approach to universality checking.
In a~similar manner to the classical algorithm, we run the subset construction
procedure and check if any rejecting macro-state is reachable. However, our
algorithm augments the subset construction with two optimisations, henceforth
referred to as \emph{Optimisation~1} and \emph{Optimisation~2}, respectively.

Optimisation~1 is based on the fact that if the algorithm encounters a
macro-state $R$ whose language is a superset of the language of a visited
macro-state $P$, then there is no need to continue the search from $R$. The
intuition behind this is that if a word is not accepted from $R$, then it is
also not accepted from $P$. For instance, in Fig.~\ref{figure:subset}(b), the
search needs not continue from the macro-state $\{s_2,s_3\}$ since its language
is a superset of the language of the initial macro-state $\{s_1, s_2\}$.
However, in general it is difficult to check if $\lang{\mathcal{A}}(P) \subseteq
\lang{\mathcal{A}}(R)$ before the resulting deterministic FA is completely built. Therefore,
we suggest to use an easy-to-compute alternative based on the following lemma.

\begin{lemma}\label{lemma:language1}
Let $P$, $R$ be two macro-states, $\mathcal{A}$ be an FA, and $\preceq$ be a relation in $\rel{\mathcal{A}}$. Then, $P \preceq^{\forall\exists} R$ implies $\lang{\mathcal{A}}(P) \subseteq \lang{\mathcal{A}}(R)$.
\end{lemma}

Note that in Lemma~\ref{lemma:language1}, $\preceq$ can be any relation on the
states of $\mathcal{A}$ that implies language inclusion. This includes any
simulation relation (Lemma~\ref{lemma:simproperty}). When $\preceq$ is the
maximal simulation or the identity relation, it can be efficiently obtained from
$\mathcal{A}$ before the subset construction algorithm is triggered and used to
prune out unnecessary search paths.

An example of how the described optimisation can help is given in
Fig.~\ref{figure:subset}(b). If $\preceq$ is the identity, the universality
checking algorithm will not continue the search from the macro-state
$\{s_1,s_2,s_4\}$ because it is a superset of the initial macro-state. In fact,
the pure antichain-based approach~\cite{wulf:antichains} can be viewed as a special
case of our simulation enhanced antichain approach when $\preceq$ is the identity. Notice that, in this case,
only 7 macro-states are generated (the classical algorithm generates 13
macro-states). When $\preceq$ is the maximal simulation, we do not need to
continue from the macro-state $\{s_2,s_3\}$ either because $s_1\preceq s_3$ and
hence $\{s_1,s_2\}\preceq^{\forall\exists} \{s_2,s_3\}$. In this case, only 3
macro-states are generated. As we can see from the example, a better reduction
of the number of generated states can be achieved when a weaker relation (e.g.,
the maximal simulation) is used.

Optimisation 2 is based on the observation that $\L(\A)(P)=\L(\A)(P\setminus
\{p_1\})$ if there is some $p_2\in P$ with $p_1\preceq p_2$. This fact is a
simple consequence of Lemma~\ref{lemma:language1} (note that
$P\preceq^{\forall\exists} P\setminus \{p_1\}$). Since the two macro-states $P$
and $P\setminus \{p_1\}$ have the same language, if a word is not accepted from
$P$, it is not accepted from $P\setminus \{p_1\}$ either. On the other hand, if
all words in $\Sigma^*$ can be accepted from $P$, then they can also be accepted
from $P\setminus \{p_1\}$. Therefore, it is safe to replace the macro-state $P$
with $P\setminus \{p_1\}$.

Consider the example in Fig.~\ref{figure:subset}. If $\preceq$ is the maximal
simulation relation, we can remove the state $s_2$ from the initial macro-state
$\{s_1,s_2\}$ without changing its language, because $s_2 \preceq s_1$. This
change will propagate to all the searching paths. With this optimisation, our
approach will only generates 3 macro-states, all of which are singletons. The
result after apply the two optimisations are applied is shown in
Fig.~\ref{figure:subset}(c).

\begin{algorithm}[t]
    \KwIn{An FA $\mathcal{A} = (\Sigma, Q, \delta, I, F)$ and a relation $\preceq\in \rel{\mathcal{A}}$.}
    \KwOut{$\true$ if $\mathcal{A}$ is universal. Otherwise, $\false$.}
  \caption{\textit{Universality Checking}}
  \label{algorithm:universality}
    \lIf{$I$ is rejecting}{\KwRet{$\false$}}\;
    $\Processed$:=$\emptyset$\;
    $\Next$:=$\{\minimize(I)\}$\;
    \While{$\Next\neq \emptyset$}{
        Pick and remove a macro-state $R$ from $\Next$ and move it to $\Processed$\;
        \ForEach{$P \in \{ \minimize(R') \mid R' \in \mathit{Post}(R) \}$}{
            \lIf{$P$ is rejecting}{\KwRet{$\false$}}\;
            \ElseIf{$\neg\exists S \in \Processed \cup \Next$ s.t. $S \preceq^{\forall\exists} P$}{\label{alg:trick2start}
                Remove all $S$ from $\Processed \cup \Next$ s.t. $P \preceq^{\forall\exists} S$\;
                Add $P$ to $\Next$;
            }
        }
    }
    \KwRet{$\true$}
\end{algorithm}

Algorithm~\ref{algorithm:universality} describes our approach in pseudocode. In
this algorithm, the function $\minimize(R)$ implements Optimisation~2.
The function does the following: it chooses a~new state $r_1$ from
$R$, removes $r_1$ from $R$ if there exists a state $r_2$ in $R$ such that $r_1 \preceq r_2$, and
then repeats the procedure until all of the states in $R$ are processed.
Lines~8--10 of the algorithm implement Optimisation~1. Overall, the
algorithm works as follows. Till the set $\Next$ of macro-states waiting to be
processed is non-empty (or a rejecting macro-state is found), the algorithm
chooses one macro-state from $\Next$, and moves it to the $\Processed$ set.
Moreover, it generates all successors of the chosen macro-state, minimises them,
and adds them to $\Next$ unless there is already some
$\preceq^{\forall\exists}$-smaller macro-state in $\Next$ or in $\Processed$. If
a new macro-state is added to $\Next$, the algorithm at the same time removes
all $\preceq^{\forall\exists}$-bigger macro-states from both $\Next$ and
$\Processed$. Note that the pruning of the $\Next$ and $\Processed$ sets
together with checking whether a new macro-state should be added into $\Next$
can be done within a single iteration through $\Next$ and $\Processed$. We
discuss correctness of the algorithm in the next section.

\subsection{Correctness of the Optimised Universality
Checking}\label{sec:correctness}

In this section, we prove correctness of
Algorithm~\ref{algorithm:universality}.  
We first
introduce some definitions and notations that will be used in the proof. For a
macro-state $P$, define $\dist(P)\in \nat \cup \{\infty\}$ as the length of the
shortest word in $\Sigma^*$ that is not in $\lang{\mathcal{A}}(P)$ (if
$\lang{\mathcal{A}}(P)=\Sigma^*$, $\dist(P)=\infty$). For a set of macro-states
$\mathit{MStates}$, the function $\dist(\mathit{MStates})\in \nat \cup
\{\infty\}$ returns the length of the shortest word in $\Sigma^*$ that is not
in the language of some macro-state in $\mathit{MStates}$. More precisely, if
$\mathit{MStates}=\emptyset$, $\dist(\mathit{MStates})=\infty$, otherwise,
$\dist(\mathit{MStates})=\min_{P\in \mathit{MStates}}\dist(P)$. The
predicate $\univ(\mathit{MStates})$ is true if and only if all the macro-states
in $\mathit{MStates}$ are universal, i.e., $\forall P\in \mathit{MStates}:
\lang{\mathcal{A}}(P)=\Sigma^*$.

The lemma bellow follows from the fact that if $\lang{\mathcal{A}}(P) \subseteq \lang{\mathcal{A}}(R)$, then the shortest word rejected by $R$ is also rejected by $P$.
\begin{lemma}\label{lemma:distance} Let $P$ and $R$ be two macro-states such
that $\lang{\mathcal{A}}(P) \subseteq \lang{\mathcal{A}}(R)$. We have
$\dist(P) \leq \dist(R)$.\end{lemma}

Lemma~\ref{lemma:invariant} describes the invariants used to prove the partial
correctness of Algorithm~\ref{algorithm:universality}.

\begin{lemma}\label{lemma:invariant} The below two loop
invariants hold in Algorithm~\ref{algorithm:universality}:\begin{enumerate}
  \item \label{lemma:invariant1}$\neg\univ(\Processed \cup \Next)
  \implies \neg\univ(\{I\})$.
  \item \label{lemma:invariant2}$\neg\univ(\{I\}) \implies
  \dist(\Processed) > \dist(\Next)$.

\end{enumerate}
\end{lemma}

\begin{proof} It is trivial to see that the invariants hold at the entry of the
loop, taking into account Lemma~\ref{lemma:language1} covering the effect of the
$\minimize$ function. We show that the invariants continue to hold when
the loop body is executed from a configuration of the algorithm in which the
invariants hold. We use $\Processed^{\mathit{old}}$ and $\Next^{\mathit{old}}$
to denote the values of $\Processed$ and $\Next$ when the control is on line 4
before executing the loop body and we use $\Processed^{\mathit{new}}$ and
$\Next^{\mathit{new}}$ to denote their values when the control gets back to line
4 after executing the loop body once. We assume that $\Next^{\mathit{old}} \neq
\emptyset$.

Let us start with Invariant 1. Assume first that
$\univ(\Processed^{\mathit{old}} \cup \Next^{\mathit{old}})$ holds.
Then, the macro-state $R$ picked on line 5 must be universal, which holds also for all of its successors and, due
to Lemma~\ref{lemma:language1}, also for their minimised versions, which may be
added to $\Next$ on line 10. Hence, $\univ(\Processed^{\mathit{new}}
\cup \Next^{\mathit{new}})$ holds after executing the loop body, and thus
Invariant 1 holds too. Now assume that $\neg
\univ(\Processed^{\mathit{old}} \cup \Next^{\mathit{old}})$ holds. Then,
$\neg\univ(\{I\})$ holds, and hence Invariant 1 must hold for
$\Processed^{\mathit{new}}$ and $\Next^{\mathit{new}}$ too.

We proceed to Invariant 2 and we assume that $\neg\univ(\{I\})$ holds
(the other case being trivial). Hence, $\dist(\Processed^{\mathit{old}})
> \dist(\Next^{\mathit{old}})$ holds. We distinguish two
cases:\begin{enumerate}

  \item $\dist(R) = \infty$ or $\exists Q \in \Processed^{\mathit{old}}:~
  \dist(Q) \leq \dist(R)$. In this case, $\dist(\Processed)$ will not decrease
  on line 5. From $\dist(\Processed^{\mathit{old}}) >
  \dist(\Next^{\mathit{old}})$, there exists some macro-state $R'$ in
  $\Next^{\mathit{old}}$ s.t. $\dist(R')= \dist(\Next^{\mathit{old}}) <
  \dist(\Processed^{\mathit{old}})\leq \dist(Q)\leq \dist(R)$. Therefore,
  $\dist(\Next)$ will not change on line 5 either. Moreover, for any macro-state
  $P$, removing $Q$ s.t. $P\preceq^{\forall\exists} Q$ from $\Next$ and
  $\Processed$ on line 9 and then adding $P$ to $\Next$ on line 10 cannot
  invalidate $\dist(\Processed^{\mathit{new}}) > \dist(\Next^{\mathit{new}})$
  since $\dist(P) \leq \dist(Q)$ due to Lemmas~\ref{lemma:language1}
  and~\ref{lemma:distance}. Hence, Invariant 2 must hold for
  $\Processed^{\mathit{new}}$ and $\Next^{\mathit{new}}$ too.

  \item $\dist(R) \neq \infty$ and $\neg\exists Q \in
  \Processed^{\mathit{old}}:~ \dist(Q) \leq \dist(R)$. In this case, the value
  of $\dist(\Processed)$ decreases to $\dist(R)$ on line 5. Clearly, $\dist(R)
  \neq 0$ or else we would have terminated before. Then there must be some
  successor $R'$ of $R$ which is either rejecting (and the loop stops without
  getting back to line 4) or one step closer to rejection, meaning that
  $\dist(R') < \dist(R)$. Moreover, $R'$ either appears in
  $\Next^{\mathit{new}}$ or there already exists some $R'' \in
  \Next^{\mathit{old}}$ such that $R'' \preceq^{\forall\exists} R'$, meaning
  that $\dist(\Processed^{\mathit{new}}) > \dist(\Next^{\mathit{new}})$. It is
  impossible that $\exists R'' \in \Processed^{\mathit{old}}: R''
  \preceq^{\forall\exists} R'$, because $\forall R'' \in
  \Processed^{\mathit{old}}:~ \dist(R'') > \dist(R) > \dist(R')$ and from
  Lemmas~\ref{lemma:language1} and~\ref{lemma:distance}, $R''
  \preceq^{\forall\exists} R'$ implies $\dist(R'') < \dist(R')$. Furthermore, if
  some macro-state is removed from $\Processed$ on line 9, $\dist(\Processed)$
  can only grow, and hence we are done.

\end{enumerate}
\end{proof}

Due to the finite number of macro-states, we can show that
Algorithm~\ref{algorithm:universality} eventually terminates.

\begin{lemma}[Termination] \label{lemma:termination}
Algorithm~\ref{algorithm:universality} eventually terminates.
\end{lemma}

\begin{proof}For the algorithm not to terminate, it would have to be the case
that some macro-state is repeatedly added into $\Next$. However, once some
macro-state $R$ is added into $\Next$, there will always be some macro-state $Q
\in \Processed \cup \Next$ such that $Q \preceq^{\forall\exists} R$. This holds
since $R$ either stays in $\Next$, moves to $Processed$, or is replaced by some
$Q$ such that $Q \preceq^{\forall\exists} R$ in each iteration of the loop.
Hence, $R$ cannot be added to $\Next$ for the second time since a macro-state is
added to $\Next$ on line 10 only if there is no $Q \in \Processed \cup \Next$
such that $Q \preceq^{\forall\exists} R$.\end{proof}

We can now easily prove the main theorem.

\begin{theorem}\label{theorem:correctness} Algorithm~\ref{algorithm:universality} always terminates, and returns
$\true$ iff the input automaton $\mathcal{A}$ is universal.
\end{theorem}
\begin{proof}From Lemma~\ref{lemma:termination}, the algorithm eventually
terminates. It returns $\false$ only if either the set of initial states
is rejecting, or the minimised version $R'$ of some successor $S$ of a
macro-state $R$ chosen from $\Next$ on line 5 is found rejecting. In the
latter case, due to Lemma~\ref{lemma:language1}, $S$ is also rejecting. Then
$R$ is non-universal, and hence $\univ(\Processed \cup \Next)$ is false.
By Lemma~\ref{lemma:invariant} (Invariant 1), we have $\mathcal{A}$ is not
universal. The algorithm returns $\true$ only when $\Next$ becomes empty.
When $\Next$ is empty, $\dist(\Processed) > \dist(\Next)$ is not
true. Therefore, by Lemma~\ref{lemma:invariant} (Invariant 2), $\mathcal{A}$
is universal.\end{proof}

\subsection{The FA Language Inclusion Problem}\label{sec:language}

The technique described in Section~\ref{sec:universality} can be generalised to
solve the \emph{language-inclusion problem}. Let $\mathcal{A}$ and $\mathcal{B}$
be two FA. The \emph{language inclusion problem} for $\mathcal{A}$ and
$\mathcal{B}$ is to decide whether $\lang{\mathcal{A}}\subseteq
\lang{\mathcal{B}}$. This problem is also PSPACE-complete. The classical
algorithm for solving this problem builds on-the-fly the product automaton
$\mathcal{A}\times \overline{\mathcal{B}}$ of $\mathcal{A}$ and the complement
of $\mathcal{B}$ and searches for an accepting state. A state in the product
automaton $\mathcal{A}\times \overline{\mathcal{B}}$ is a pair $\ps{p}{P}$ where
$p$ is a state in $\mathcal{A}$ and $P$ is a macro-state in $\mathcal{B}$. For
convenience, we call such a pair $\ps{p}{P}$ a \emph{product-state}. A
product-state is accepting iff $p$ is an accepting state in $\mathcal{A}$ and
$P$ is a rejecting macro-state in $\mathcal{B}$. We use
$\lang{\mathcal{A},\mathcal{B}}\ps{p}{P}$ to denote the language of the
product-state $\ps{p}{P}$ in $\mathcal{A}\times \overline{\mathcal{B}}$. The
language of $\mathcal{A}$ is not contained in the language of $\mathcal{B}$ iff
there exists some accepting product-state $\ps{p}{P}$ reachable from some
initial product-state. Indeed, $\lang{\mathcal{A},\mathcal{B}}\ps{p}{P}
=\lang{\mathcal{A}}(p) \setminus \lang{\mathcal{B}}(P)$, and the language of
$\mathcal{A}\times \overline{\mathcal{B}}$ consists of words which can be used
as witnesses of the fact that $\lang{\mathcal{A}}\subseteq \lang{\mathcal{B}}$
does not hold. In a similar manner to universality checking, the algorithm can
stop the search immediately and conclude that the language inclusion does not
hold whenever an accepting product-state is encountered. An example of a run of
the classical algorithm is given in Fig.~\ref{figure:lan_inclu}. We find that
$\lang{\mathcal{A}}\subseteq \lang{\mathcal{B}}$ is true and the algorithm
generates 13 product-states (Fig.~\ref{figure:lan_inclu}(c), the area labelled
``Classical'').

\begin{figure}[t]
\begin{minipage}[b]{0.15\linewidth}
\begin{minipage}[b]{1\linewidth}
    \centering
    \scalebox{0.9}{
    \begin{tikzpicture}[>=latex',join=bevel]
    \tikzstyle{every state}=     [draw=blue!50,very thick,fill=blue!20, minimum size=5mm, initial text=, initial distance=3mm]%
    \node[state, accepting] (p1) at (1.5,0) {\scriptsize$p_2$};
    \node[initial, state](p2) at (0,0) {\scriptsize$p_1$};
    \draw [->] (p1) to[loop above] node[auto] {\scriptsize $a$} (p1);
    \draw [->] (p1) to[bend left] node[auto] {\scriptsize $a,b$} (p2);
    \draw [->] (p2) to[loop above] node[auto] {\scriptsize $a$} (p2);
    \draw [->] (p2) to[bend left] node[auto] {\scriptsize $a$} (p1);

    \end{tikzpicture}}\\
    (a) FA $\mathcal{A}$
\end{minipage}
\begin{minipage}[b]{1\linewidth}
    \centering
    \scalebox{0.9}{
    \begin{tikzpicture}[>=latex',join=bevel]
    \tikzstyle{every state}=     [draw=blue!50,very thick,fill=blue!20, minimum size=5mm, initial text=, initial distance=3mm]%
    \node[state, accepting] (q1) at (1.5,-2) {\scriptsize$q_2$};
    \node[initial, state](q2) at (0,-2) {\scriptsize$q_1$};
    \draw [->] (q1) to[loop above] node[auto] {\scriptsize $a$} (q1);
    \draw [->] (q1) to[bend left] node[auto] {\scriptsize $a,b$} (q2);
    \draw [->] (q2) to[bend left] node[auto] {\scriptsize $a$} (q1);
    \end{tikzpicture}}\\
    (b) FA $\mathcal{B}$
\end{minipage}
\end{minipage}
\begin{minipage}[b]{0.85\linewidth}
    \centering
    \scalebox{0.8}{
    \begin{tikzpicture}[>=latex',join=bevel,]
    \tikzstyle{mstate}=     [draw=green!75,thick,fill=green!30, shape=rectangle, initial text=]%
    \tikzstyle{fstate}=     [draw=green!50,thick, dashed,fill=green!15, shape=rectangle, initial text=]%

    \draw[fill=blue!20, dashed, draw=blue!50] (-1.5,0.5) -- (-0.5,0.5) .. controls (4,0.5) and (5.5,-2) .. (6,-4.5) -- (-5.5,-4.5).. controls (-6,-2.5) and (-3,0.5) ..(-1.5,0.5);

    \draw[fill=red!20, dashed, draw=red!50] (-1.5,0.1) -- (-0.5,0.1) .. controls (4,0.1) and (5.5,-2) .. (5,-3.5)-- (-5,-3.5).. controls (-6,-2.5) and (-3,0.1) ..(-1.5,0.1);

    \draw[fill=orange!30, dashed, draw=orange!200] (-1.5,-0.2) -- (-0.5,-0.2) .. controls (3,-0.6) and (5,-2) .. (4.8,-3.3)-- (-1,-3.3)-- (-1,-2.5)-- (-3,-2.5).. controls (-5,-2.5) and (-3,0) ..(-1.5,-0.2);

    \draw[fill=purple!20, dashed, draw=purple!200] (-2.5,-0.5) .. controls (-1.5,-0.3) and (0.5,-0.3) .. (0.5,-0.6) -- (0.5,-1.1) .. controls (0.5,-1.6) and (-2.5,-1.6) .. (-2.5,-1.1)--(-2.5,-0.5) ;

    \node at (-1,-0.6) {\scriptsize Optimisation~1(b)};
        \node at (1.5,-1.5) {\scriptsize Optimisation~1(a)};
\node at (1.1,-0.2) {\scriptsize Antichain};
\node at (-1,0.3) {\scriptsize Classical};

    \node[mstate, initial] (s1) at (-1, -1) {\scriptsize$p_1, \{q_1\}$};

    \node[mstate] (s2) at (-3, -2) {\scriptsize$p_1, \{q_2\}$};
    \node[mstate] (s3) at (1, -2) {\scriptsize$p_2, \{q_2\}$};

    \node[fstate] (s4) at (-4, -3) {\scriptsize$p_1, \{q_1,q_2\}$};
    \node[fstate] (s5) at (-2, -3) {\scriptsize$p_2, \{q_1,q_2\}$};
    \node[mstate] (s6) at (0, -3) {\scriptsize$p_1, \{q_1,q_2\}$};
    \node[mstate] (s7) at (2, -3) {\scriptsize$p_2, \{q_1,q_2\}$};
    \node[fstate] (s8) at (4, -3) {\scriptsize$p_1, \{q_1\}$};

    \node[fstate] (s9) at (-1.7, -4) {\scriptsize$p_1, \{q_1,q_2\}$};
    \node[fstate] (s10) at (0, -4) {\scriptsize$p_2, \{q_1,q_2\}$};
    \node[fstate] (s11) at (1.8,-4) {\scriptsize$p_1, \{q_1,q_2\}$};
    \node[fstate] (s12) at (3.5,-4) {\scriptsize$p_2, \{q_1,q_2\}$};
    \node[fstate] (s13) at (5,-4) {\scriptsize$p_1, \{q_1\}$};

    \draw [->, thick] (s1) to node[right=1mm] {\scriptsize $a$} (s2);
    \draw [->, thick] (s1) to node[left=1mm] {\scriptsize $a$} (s3);

    \draw [->, thick] (s2) to node[right] {\scriptsize $a$} (s4);
    \draw [->, thick] (s2) to node[left] {\scriptsize $a$} (s5);

    \draw [->, thick] (s3) to node[right] {\scriptsize $a$} (s6);
    \draw [->, thick] (s3) to node[left] {\scriptsize $a$} (s7);
    \draw [->, thick] (s3) to node[left] {\scriptsize $b$} (s8);

    \draw [->, thick] (s6) to node[right] {\scriptsize $a$} (s9);
    \draw [->, thick] (s6) to node[left] {\scriptsize $a$} (s10);

    \draw [->, thick] (s7) to node[right] {\scriptsize $a$} (s11);
    \draw [->, thick] (s7) to node[left] {\scriptsize $a$} (s12);
    \draw [->, thick] (s7) to node[left=1mm] {\scriptsize $b$} (s13);
    \end{tikzpicture}}\\
    (c) A run of the algorithms while checking $\lang{\mathcal{A}}\subseteq \lang{\mathcal{B}}$.
\end{minipage}

  \caption{Language Inclusion Checking Algorithms}
  \label{figure:lan_inclu}
\end{figure}
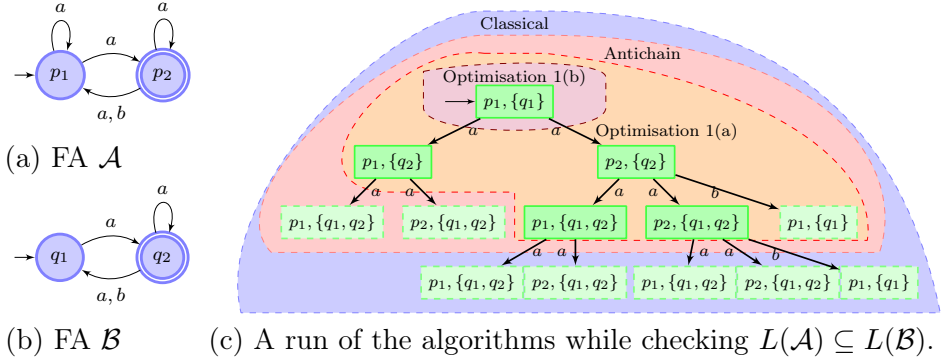

Optimisation~1 that we use for universality checking can be generalised for
language inclusion checking as follows. Let $\mathcal{A} = (\Sigma,
Q_\mathcal{A}, \delta_\mathcal{A}, I_\mathcal{A}, F_\mathcal{A})$ and
$\mathcal{B} = (\Sigma, Q_\mathcal{B},\delta_\mathcal{B}, I_\mathcal{B},
F_\mathcal{B})$ be two FA such that $Q_\mathcal{A} \cap Q_\mathcal{B} =
\emptyset$. We denote by $\mathcal{A}\cup\mathcal{B}$ the FA $(\Sigma,
Q_\mathcal{A} \cup Q_\mathcal{B}, \delta_\mathcal{A} \cup \delta_\mathcal{B},
I_\mathcal{A} \cup I_\mathcal{B}, F_\mathcal{A} \cup F_\mathcal{B})$.  Let
$\preceq$ be a relation in $\rel{(\mathcal{A}\cup\mathcal{B})}$. During the
process of constructing the product automaton and searching for an accepting
product-state, we can stop the search from a product-state $\ps{p}{P}$ if (a)
there exists some visited product-state $\ps{r}{R}$ such that $p\preceq r$ and
$R \preceq^{\forall\exists} P$, or (b) $\exists p'\in P: p \preceq p'$.
Optimisation~1(a) is justified by Lemma~\ref{lemma:language1_p}, which is very
similar to  Lemma~\ref{lemma:language1} for universality checking.

\begin{lemma}\label{lemma:language1_p} Let $\mathcal{A}$, $\mathcal{B}$ be two
FA, $\ps{p}{P}$, $\ps{r}{R}$ be two product-states where $p$, $r$ are states
in $\mathcal{A}$ and $P$, $R$ are macro-states in $\mathcal{B}$, and $\preceq$
be a relation in $\rel{(\mathcal{A}\cup\mathcal{B})}$. Then, $p \preceq r$ and
$R \preceq^{\forall\exists} P$ implies $\lang{\mathcal{A},\mathcal{B}}\ps{p}{P}
\subseteq \lang{\mathcal{A},\mathcal{B}}\ps{r}{R}$. \end{lemma}

By the above lemma, if a word takes the product-state $(p, P)$ to an accepting
product-state, it will also take $(r, R)$ to an accepting product-state.
Therefore, we do not need to continue the search from $(p, P)$.

Let us use Fig.~\ref{figure:lan_inclu}(c) to illustrate Optimisation~1(a). As we
mentioned, the pure antichain-based approach can be viewed as a special case of our simulation enhanced antichain
approach when $\preceq$ is the identity. When $\preceq$ is the identity, we do
not need to continue the search from the product-state $(p_2,\{q_1,q_2\})$
because $\{q_2\}\subseteq \{q_1,q_2\}$. In this case, the algorithm generates 8
product-states (Fig.~\ref{figure:lan_inclu}(c), the area labelled ``Antichain'').
In the case that $\preceq$ is the maximal simulation, we do not need to continue
the search from product-states $(p_1,\{q_2\})$, $(p_1,\{q_1,q_2\})$, and
$(p_2,\{q_1,q_2\})$ because $q_1 \preceq q_2$ and the algorithm already visited
the product-states $(p_1,\{q_1\})$ and $(p_2,\{q_2\})$. Hence, the algorithm
generates only 6 product-states (Fig.~\ref{figure:lan_inclu}(c), the area
labelled ``Optimisation~1(a)'').

If the condition of Optimisation~1(b) holds, we have that the language of $p$
(w.r.t. $\mathcal{A}$) is a subset of the language of $P$ (w.r.t.
$\mathcal{B}$). In this case, for any word that takes $p$ to an accepting state
in $\mathcal{A}$, it also takes $P$ to an accepting macro-state in
$\mathcal{B}$. Hence, we do not need to continue the search from the
product-state $(p, P)$ because all of its successor states are rejecting
product-states.
Consider again the example in Fig.~\ref{figure:lan_inclu}(c). With
Optimisation~1(b), if $\preceq$ is the maximal simulation on the states of
$\mathcal{A}\cup\mathcal{B}$, we do not need to continue the search from the
first product-state $(p_1,\{q_1\})$ because $p_1\preceq q_1$. In this case, the
algorithm can conclude that the language inclusion holds immediately after the
first product-state is generated (Fig.~\ref{figure:lan_inclu}(c), the area
labelled ``Optimisation~1(b)'').

Observe that from Lemma~\ref{lemma:language1_p}, it holds that for any
product-state $(p,P)$ such that $p_1\preceq p_2$ for some $p_1,p_2\in P$,
$\lang{\mathcal{A},\mathcal{B}}(p,P)=\lang{\mathcal{A},\mathcal{B}}(p,
P\setminus\{p_1\})$ (as $P\preceq^{\forall\exists} P\setminus \{p_1\}$).
Optimisation~2 that we used for universality checking can therefore be
generalised for language inclusion checking too.

We give the pseudocode of our optimised inclusion checking in
Algorithm~\ref{algorithm:languageinclusion}, which is a straightforward
extension of Algorithm~\ref{algorithm:universality}. In the algorithm, the
definition of the $\minimize(R)$ function is the same as what we have defined in
Section~\ref{sec:universality}. The function $\initialize(\pstates)$
applies Optimisation~1 on the set of product-states $\pstates$ to avoid
unnecessary searching. More precisely, it returns a maximal subset of $\pstates$
such that (1) for any two elements $(p,P)$, $(q,Q)$ in the subset, $p\not\preceq
q \vee Q \not\preceq^{\forall\exists} P$ and (2) for any element $(p,P)$ in the
subset, $\forall p'\in P:~p\not\preceq p'$. We define the post-image of a
product-state $\mathit{Post}((p,P)):=\{(p',P')\mid \exists
a\in\Sigma:(p,a,p')\in \delta, P'=\{p''\mid \exists p\in P : (p,a,p'')\in
\delta\}\}$.

\begin{algorithm}
    \KwIn{FA $\mathcal{A} = (\Sigma, Q_\mathcal{A}, \delta_\mathcal{A}, I_\mathcal{A}, F_\mathcal{A})$, $\mathcal{B} = (\Sigma, Q_\mathcal{B}, \delta_\mathcal{B}, I_\mathcal{B}, F_\mathcal{B})$. A relation ${\preceq} \in \rel{(\mathcal{A}\cup\mathcal{B})}$.}
    \KwOut{$\true$ if $\lang{\mathcal{A}}\subseteq \lang{\mathcal{B}}$. Otherwise, $\false$.}
  \caption{\textit{Language Inclusion Checking}}
  \label{algorithm:languageinclusion}
    \lIf{there is an accepting product-state in $\{(i, I_\mathcal{B})\mid i\in I_\mathcal{A}\}$}{\KwRet{$\false$}}\;
    $\Processed$:=$\emptyset$\;
    $\Next$:= $\initialize(\{(i, \minimize(I_\mathcal{B}))\mid i\in I_\mathcal{A}\})$\;
    \While{$\Next\neq \emptyset$}{
        Pick and remove a product-state $(r,R)$ from $\Next$ and move it to $\Processed$\;
        \ForEach{$(p,P) \in \{ (r',\minimize(R')) \mid (r',R')\in \mathit{Post}((r,R)) \}$}{
            \lIf{$(p,P)$ is an accepting product-state}{\KwRet{$\false$}}\;
            \ElseIf{$\neg\exists p'\in P$ s.t. $p\preceq p'$}{
                \If{$\neg\exists (s,S) \in \Processed \cup \Next$ s.t. $p\preceq s \wedge S \preceq^{\forall\exists} P$}{
                    Remove all $(s,S)$ from $\Processed \cup \Next$ s.t. $s \preceq p \wedge P \preceq^{\forall\exists} S$\;
                    Add $(p,P)$ to $\Next$\;
                }
            }
        }
    }
    \KwRet{$\true$}
\end{algorithm}

\paragraph{Correctness:} Define $\dist(P)\in \nat \cup \{\infty\}$
as the length of the shortest word in the language of the product-state $P$ or
$\infty$ if the language of $P$ is empty. The value $\dist(\mathit{PStates})\in
\nat \cup \{\infty\}$ is the length of the shortest word in the language of some
product-state in $\mathit{PStates}$ or $\infty$ if $\mathit{PStates}$ is empty.
The predicate $\incl(\mathit{PStates})$ is true iff for all product-states
$(p,P)$ in $\mathit{PStates}$, $\lang{\mathcal{A}}(p) \subseteq
\lang{\mathcal{B}}(P)$. The correctness of
Algorithm~\ref{algorithm:languageinclusion} can now be proved in a very similar
way to Algorithm~\ref{algorithm:universality}, using the
invariants below:\begin{enumerate}

  \item $\neg\incl(\Processed \cup \Next) \implies \neg\incl(\{(i,
  I_\mathcal{B})\mid i\in I_\mathcal{A}\})$.

  \item $\neg\incl(\{(i, I_\mathcal{B})\mid i\in I_\mathcal{A}\}) \implies
  \dist(\Processed) > \dist(\Next)$.

\end{enumerate}

\begin{theorem}\label{theorem:Inclusioncorrectness} Algorithm
\ref{algorithm:languageinclusion} terminates, and returns
$\true$ iff $\L(\A)\subseteq \L(\B)$. \end{theorem}

\section{Universality and Language Inclusion of Tree
Automata}\label{sec:TA-main}

To optimise universality and inclusion checking on word automata, we used
relations that imply language inclusion. For the case of universality and
inclusion checking on tree automata, we now propose to use relations that imply
inclusion of languages of contexts (context is the notion of a tree with ``holes'' instead of (all) leaves defined in Chapter~\ref{chapter:ta_reduction})
that are accepted from tuples of tree automata states. 
As we will see, a
relation that fits here best is upward simulation induced by identity
introduced in Chapter~\ref{chapter:ta_reduction}.  Notice that in contrast to the notion of a language
accepted from a state of a word automaton, which refers to possible ``futures''
of the state, the notion of a language accepted at a state of a TA refers to
possible ``pasts'' of the state. Our notion of languages of contexts accepted
from tuples of tree automata states speaks again about the future of states,
which turns out useful when trying to optimise the (antichain-based) subset
construction for TA. Below, we state formal definitions of the
notions needed within this chapter.

The language of $\A$ accepted from a tuple
$(q_1,\ldots,q_n)$ of states is the set of contexts $\L^\hole(\A)(q_1,\ldots,q_n)=\{t\in T^\hole\mid
t(q_1,\ldots,q_n)\Longrightarrow{q} \text{ for some } q\in F\}$. We
define the language accepted from a tuple of macro-states
$(P_1,\ldots,P_n)\subseteq Q^n$ as the set $\L^\hole(\A)(P_1,\ldots,P_n)=
\bigcup\{\L^\hole(\A)(q_1,\ldots,q_n)\mid (q_1,\ldots,q_n)\in
P_1\times\ldots\times P_n\}$.  We define $\post_a \tup q n := \{q\mid\tup q
n\xtr a q\}$. For a tuple of macro-states, we let $\post_a\tup P n :=
\bigcup\{\post_a\tup q n\mid \tup q n \in P_1\times\cdots\times P_n\}$.

Let us use $t^\hole$ to
denote the context that arises from a tree $t\in T(\Sigma)$ by replacing all
the leaf symbols of $t$ by $\hole$ and let for every leaf symbol $a\in\Sigma$, $I_a =
\{q\mid\ \xtr a q\}$ is the so called $a$-initial macro-state.
Languages accepted \emph{at} final states of $\A$ correspond to the languages accepted \emph{from}
tuples of initial macro-states of $\A$ as stated in Lemma~\ref{lemma:TAlanguages}.

\begin{lemma}\label{lemma:TAlanguages}
Let $t$ be a tree over $\Sigma$ with leaves labelled by $a_1,\ldots,a_n$.
Then $t\in\L(\A)$ if and only if
$t^\hole\in\L^\hole(\A)(I_{a_1},\ldots,I_{a_n})$.
\end{lemma}

\subsection{The Role of Upward Simulation}

We now work towards defining suitable relations on states of TA allowing us to
optimise the universality and inclusion checking. We extend relations
${\preceq}\in Q\times Q$ on states to tuples of states such that $\tup q n
\preceq\tup r n$ iff $q_i\preceq r_i$ for each $1\leq i\leq n$.  We define the
set $\rel \A$ of relations that imply inclusion of languages of tuples of
states such that ${\preceq}\in \rel \A$ iff $\tup q n\preceq\tup r
n$ implies $\L^\hole(\A)\tup q n\subseteq \L^\hole(\A)\tup r n$.

A relation that satisfies the above property is the upward simulation induced by identity defined in Chapter~\ref{chapter:ta_reduction}. For convenience, in this chapter, we will call it simply upward simulation. We note that it can be equivalently  defined in a non-parametric way as follows: 
An \emph{upward simulation} on $\A$ is a
relation ${\preceq}\subseteq Q\times Q$ such that if $q\preceq r$, then (1)
$q\in F\implies r\in F$ and (2) if $\tup q n \xtr a q'$ where $q = q_i$, then
$(q_1,\ldots,q_{i-1},r,q_{i+1},\ldots,q_n)\xtr a r'$ where $q'\preceq r'$.
\footnote{Upward simulations parametrised by a
downward simulation greater than the identity cannot be used in our framework
since they do not generally imply inclusion of languages of tuples of states.}

\begin{lemma}\label{lemma:TAsimproperty}For the maximal upward simulation $\preceq$ on
$\mathcal{A}$, we have ${\preceq}\in\rel \A$.\end{lemma}

\begin{proof}
We first show that the maximal upward simulation $\preceq$
has the following property: If $\tup q n \xtr a q'$ in $\A$, then for every $\tup r n$
with $\tup q n \preceq\tup r n$, there is $r'\in Q$ such that $q'\preceq r'$ and
$\tup r n \xtr a r'$. From $\tup q n \xtr a q'$ and $q_1 \preceq
r_1$, we have that there is some rule $(r_1,q_2,\ldots,q_n)\xtr a s_1$ such that
$q' \preceq s_1$. From the existence of $(r_1,q_2,\ldots,q_n)\xtr a s_1$ and
from $q_2 \preceq r_2$, we then get that there is some rule
$(r_1,r_2,q_3,\ldots,q_n)\xtr a s_2$ such that $s_1 \preceq s_2$, etc. Since the
maximal upward simulation is transitive \cite{abdulla:computing}, we obtain the
property mentioned above. This in turn implies Lemma~\ref{lemma:TAsimproperty}.
\end{proof}

\subsection{Tree Automata Universality Checking}

We now show how upward simulations can be used for optimised universality
checking on tree automata. Let $\A = (\Sigma,Q,\Delta,F)$ be a tree automaton.
We define $T^\hole_n(\Sigma)$ as the set of all contexts over $\Sigma$ with
$n$ leaves. We say that an $n$-tuple $\tup q n$ of states of $\A$ is universal
if $\L^\hole(\A)\tup q n = T^\hole_n(\Sigma)$, this is, all contexts with $n$
leaves constructable over $\Sigma$ can be accepted from $\tup q n$. A set of
macro-states $\mstates$ is universal if all tuples in $\mstates^*$ are
universal. From Lemma~\ref{lemma:TAlanguages}, we can deduce that $\A$ is
universal (i.e., $\L(\A) = T(\Sigma)$) if and only if $\{I_a\mid a\in\Sigma_0\}$
is universal.

The following Lemma allows us to design a new TA universality checking algorithm
in a similar manner to Algorithm~\ref{algorithm:universality} using
Optimisations 1 and 2 from~\mbox{Section}~\ref{sec:universality}.

\begin{lemma}\label{lemma:TAlanguage1}For a given ${\preceq}\in \rel \A$ and two
tuples of macro-states of $\A$, if $\tup R n \preceq^{\forall\exists} \tup
P n$, then $\L^\hole(\A)\tup R n \subseteq \L^\hole(\A)\tup P n$. \end{lemma}

Algorithm~\ref{algorithm:TAuniversality} describes our simulation enhanced antichain approach to checking
universality of tree automata in pseudocode. It resembles closely Algorithm
\ref{algorithm:universality}. There are two main differences: (1) The initial
value of the $\Next$ set is the result of applying the function $\initialize$ to
the set $\{\minimize(I_a)\mid a\in\Sigma_0\}$. $\initialize$ returns the set of
all macro-states in $\{\minimize(I_a)\mid a\in\Sigma_0\}$, which are minimal
w.r.t. $\preceq^{\forall\exists}$ (i.e., those macro states with the best chance
of finding a counterexample to universality). (2) The computation of the
$\post$-image of a set of macro-states is a bit more complicated. More
precisely, for each symbol $a\in\Sigma_n,n\in\nat$, we have to compute the post
image of each $n$-tuple of macro-states from the set. We design the algorithm
such that we avoid computing the $\post$-image of a tuple more than once. We
define the $\post$-image  $\post(\mstates)(R)$ of a set of macro-states
$\mstates$ w.r.t. a macro-states $R\in\mstates$. It is the set of all
macro-states $P = \post_a\tup P n$ where $a\in\Sigma_n,n\in\nat$ and $R$ occurs
at least once in the tuple $\tup P n\in \mstates^*$. Formally,
$\post(\mstates)(R) = \bigcup_{a\in\Sigma}\{\post_a\tup P n \mid n=\#(a),
P_1,\ldots,P_n \in \mstates, R\in\{P_1,\ldots,P_n\}\}$.

\begin{algorithm}[t]
    \KwIn{A tree automaton $\mathcal{A} = (\Sigma, Q, \Delta, F)$ and a
relation ${\preceq}\in\rel \A$.}
    \KwOut{$\true$ if $\mathcal{A}$ is universal. Otherwise, $\false$.}
  \caption{\textit{Tree Automata Universality Checking}}
  \label{algorithm:TAuniversality}
    \lIf{$\exists a\in \Sigma_0$ such that  $I_a$ is rejecting}{\KwRet{$\false$}}\;
    $\Processed$:=$\emptyset$\;
    $\Next$:=  $\initialize\{\minimize(I_a) \mid a\in\Sigma_0\}$\;
    \While{$\Next\neq \emptyset$}{
        Pick and remove a macro-state $R$ from $\Next$ and move it to $\Processed$\;
        \ForEach{$P \in \{ \minimize(R') \mid R' \in \mathit{Post}(\Processed)(R) \}$}{
            \lIf{$P$ is a rejecting macro-state}{\KwRet{$\false$}}\;
            \ElseIf{$\neg\exists Q \in \Processed \cup \Next$ s.t. $Q \preceq^{\forall\exists} P$}{\label{alg:TAtrick2start}
                Remove all $Q$ from $\Processed \cup \Next$ s.t. $P \preceq^{\forall\exists} Q$\;
                Add $P$ to $\Next$;
            }

        }
    }
    \KwRet{$\true$}
\end{algorithm}

\subsection{Correctness of the TA Universality Checking}
\label{sec:TAUnivCorrectness}

In this section, we prove correctness of
Algorithm~\ref{algorithm:TAuniversality} in a very similar way to
Algorithm~\ref{algorithm:universality}, using suitably modified notions of
distances and ranks. Let $\A = (Q, \Sigma,\Delta, F)$ be a TA. For $n\geq 0$ and
an $n$-tuple of macro-states $\vect Q n$ where $Q_i \subseteq Q$ for $1 \leq
i \leq n$, we let $\mathbf{Dist}(Q_1,\ldots,Q_n)=0$ iff $Q_i \cap F = \emptyset$
for some $i\in\{1,\ldots,n\}$. We define $\mathbf{Dist}(Q_1,\ldots,Q_n) = k \in
\nat^+ \cup \{ \infty \}$ iff $Q_i \subseteq F$ for all $i\in\{1,\ldots,n\}$ and
$k = \min(\{ |t| \mid t \in T_n^\hole(\Sigma) \wedge t \not\in
\L^\hole(\A)(Q_1,\ldots,Q_n) \})$. Here, $|t|$ is the number of nodes of $t$ and we assume $\min(\emptyset) = \infty$.
For a set $\mathit{MStates}$ of macro-states over $Q$, we define the measure
$\mathbf{Rank}(\mathit{MStates}) = \min(\{ \mathbf{Dist}(Q_1,\ldots,Q_n) \mid n
\geq 1 \wedge \forall 1 \leq i \leq n: Q_i \in \mathit{MStates} \})$ and the predicate
$\mathbf{Univ}(\mathit{MStates}) ~\Longleftrightarrow~
\mathbf{Rank}(\mathit{MStates})=\infty$.

\begin{lemma}\label{theorem:TAUnivInvariant} The below two loop invariants hold
in Algorithm~\ref{algorithm:TAuniversality}:\begin{enumerate}

  \item $\neg\mathbf{Univ}(\Processed \cup \Next) \implies
  \neg\mathbf{Univ}(\{I_a \mid a \in \Sigma_0 \})$.

  \item $\neg\mathbf{Univ}(\{I_a \mid a \in \Sigma_0 \}) \implies
  \mathbf{Rank}(\Processed) > \mathbf{Rank}(\Processed \cup \Next)$.

\end{enumerate}
\end{lemma}

\begin{proof} It is trivial to see that the invariants hold at the entry of the
loop, taking into account Lemma~\ref{lemma:TAlanguage1}. We show that the
invariants continue to hold when the loop body is executed from a configuration
of the algorithm in which the invariants hold. We use
$\Processed^{\mathit{old}}$ and $\Next^{\mathit{old}}$ to denote the values of
$\Processed$ and $\Next$ when the control is on line 4 before executing the loop
body and we use $\Processed^{\mathit{new}}$ and $\Next^{\mathit{new}}$ to denote
their values when the control gets back to line 4 after executing the loop body
once. We assume that $\Next^{\mathit{old}} \neq \emptyset$.

Let us start with Invariant 1. Assume first that
$\mathbf{Univ}(\Processed^{\mathit{old}} \cup \Next^{\mathit{old}})$ holds.
Then, the macro-state $R$ can appear within tuples constructed over $\Processed^{\mathit{old}}
\cup \Next^{\mathit{old}}$ which are universal only. In such a case, all
macro-states $Q$ reachable from all tuples $T$ built over
$\Processed^{\mathit{old}} \cup \Next^{\mathit{old}}$ are such that when we add
them to $\Processed^{\mathit{old}} \cup \Next^{\mathit{old}}$, the resulting set
will still allow building universal tuples only. Otherwise, one could take a
non-universal tuple containing some of the newly added macro-states $Q$, replace
$Q$ by the tuple $T$ from which it arose, and obtain a non-universal tuple over
$\Processed^{\mathit{old}} \cup \Next^{\mathit{old}}$, which is impossible.
Hence, the possibility of adding the new macro-states to $\Next$ on line 10
cannot cause non-universality of $\Processed^{\mathit{new}} \cup
\Next^{\mathit{new}}$, which due to Lemma~\ref{lemma:TAlanguage1} holds when
adding the minimised macro-states too. Moreover, removing elements from $\Next$
or $\Processed$ cannot cause non-universality either. Hence, Invariant 1 holds
over $\Processed^{\mathit{new}}$ and $\Next^{\mathit{new}}$ in this case. Next,
let us assume that $\neg \mathbf{Univ}(\Processed^{\mathit{old}} \cup
\Next^{\mathit{old}})$ holds. Then, $\neg\mathbf{Univ}(\{I_a \mid a \in \Sigma_0
\})$ holds, and hence Invariant 1 must hold for $\Processed^{\mathit{new}}$ and
$\Next^{\mathit{new}}$ too.

We proceed to Invariant 2 assuming that $\neg\mathbf{Univ}(\{I_a \mid a \in
\Sigma_0 \})$ holds (the other case is trivial). Hence,
$\mathbf{Rank}(\Processed^{\mathit{old}}) >
\mathbf{Rank}(\Processed^{\mathit{old}} \cup \Next^{\mathit{old}})$ holds. We
distinguish two cases:\begin{enumerate}

  \item In order to build a tuple $T$ over $\Processed^{\mathit{old}}$ and
  $\Next^{\mathit{old}}$ that is of $\mathbf{Dist}$ equal to
  $\mathbf{Rank}(\Processed^{\mathit{old}} \cup \Next^{\mathit{old}})$, one
  needs to use a macro-state $Q$ in $\Next^{\mathit{old}} \setminus \{ R \}$.
  The macro-state $Q$ stays in $\Next^{\mathit{new}}$ or is replaced by a
  $\preceq^{\forall\exists}$-smaller macro-state added to $\Next$ on line 10
  that, due to Lemma~\ref{lemma:TAlanguage1}, can only allow to build tuples of
  the same or even smaller $\mathbf{Dist}$. Likewise, the macro-states
  accompanying $Q$ in $T$ stay in $\Next^{\mathit{new}}$ or
  $\Processed^{\mathit{new}}$ or are replaced by
  $\preceq^{\forall\exists}$-smaller macro-states added to $\Next$ on line 10
  allowing to build tuples of the same or smaller $\mathbf{Dist}$, due to
  Lemma~\ref{lemma:TAlanguage1}. Hence, moving $R$ to $Processed$ on line 5
  cannot cause the invariant to break. Moreover, adding some further
  macro-states to $\Next$ on line 10 can only cause $\mathbf{Rank}(\Processed
  \cup \Next)$ to decrease while removing macro-states from $\Processed$ on line
  9 can only cause $\mathbf{Rank}(\Processed)$ to grow. Finally, replacing a
  macro-state in $\Next$ by a $\preceq^{\forall\exists}$-smaller one as a
  combined effect of lines 9 and 10 can again just decrease
  $\mathbf{Rank}(\Processed \cup \Next)$, due to Lemma~\ref{lemma:TAlanguage1}.
  Hence, in this case, Invariant 2 must hold over $\Processed^{\mathit{new}}$
  and $\Next^{\mathit{new}}$.

  \item One can build some tuple $T$ over $\Processed^{\mathit{old}}$ and
  $\Next^{\mathit{old}}$ that is of $\mathbf{Dist}$ equal to
  $\mathbf{Rank}(\Processed^{\mathit{old}} \cup \Next^{\mathit{old}})$ using
  $\Processed^{\mathit{old}} \cup \{ R \}$ only. In this case, there must be
  tuples constructable over $\Processed^{\mathit{old}} \cup \{ R \}$ and
  containing $R$ that are not universal. We can distinguish the following
  subcases:\begin{enumerate}

    \item From some of the tuples built over $\Processed^{\mathit{old}} \cup \{
    R \}$ and containing $R$, a non-accepting macro-state is reached via a
    single transition of $\A$, and the algorithm stops without getting back to
    line 4.

    \item Otherwise, some macro-states that appear in
    $\mathit{Post}(\Processed,R)$ and that will be added in the minimised form
    to $\Next$ must allow one to construct tuples which are of $\mathbf{Dist}$
    smaller than those based on $R$. This holds since if a macro-state $Q$ is
    reached from some tuple $T$ containing $R$ by a single transition, we can
    replace $T$ in larger tuples leading to non-acceptation by $Q$, and hence
    decrease the size of the context needed to reach non-acceptation. Taking
    into account Lemma~\ref{lemma:TAlanguage1} to cover the effect of the
    minimisation and using a similar reasoning as above for covering the effect
    of lines 9 and 10, it is then clear that Invariant~2 will remain to hold in
    this case.

  \end{enumerate}

\end{enumerate}
\end{proof}

We can now prove Lemma~\ref{lemma:TAUnivTermination} and Theorem~\ref{theorem:TAcorrectness} below in a very similar way as Lemma~\ref{lemma:termination} and Theorem~\ref{theorem:correctness}, respectively. 

\begin{lemma}\label{lemma:TAUnivTermination}
Algorithm~\ref{algorithm:TAuniversality} eventually terminates.\end{lemma}

\begin{theorem}\label{theorem:TAcorrectness}
Algorithm
\ref{algorithm:TAuniversality} always terminates, and returns $\true$ if and only if
the input tree automaton $\mathcal{A}$ is universal.
\end{theorem}

\subsection{Downward Universality Checking with Antichains}

The \emph{upward universality} introduced above tree automata
automata conceptually corres\-ponds to the \emph{forward} universality
checking of finite word automata of
\cite{wulf:antichains,doyen:antichain} where also a dual \emph{backward}
universality checking is introduced. 
The backward universality algorithm from \cite{wulf:antichains,doyen:antichain} is  based on
computing the \emph{controllable predecessors} of the set of non-final states.
Controllable predecessors are the predecessors that can be forced by an input
symbol to continue into a given set of states. Then, the automaton is
non-universal iff the controllable predecessors of the non-final states cover
the set of initial states.

\emph{Downward universality checking} for tree automata as a dual approach to
upward universality checking is problematic since the controllable predecessors
of a set of states $s \subseteq Q$ of an TA $\mathcal{A}=(Q,\Sigma, F,\Delta)$
do not form a set of states, but a set of \emph{tuples} of states, i.e., for
$a\in\Sigma$,  $\mathit{CPre}_a(s) = \{ (q_1,\dotsc,q_n) \mid n \in \mathbb{N}
\wedge \forall q \in Q: \trans q n a q \in s \}$. Note that if we flatten the
set $\mathit{CPre}_a(s)$ to the set $\mathit{FCPre}_a(s)$ of states that appear
in some of the tuples of $\mathit{CPre}_a(s)$ and check that starting from leaf
rules the computation can be forced into some subset of $\mathit{FCPre_a}(s)$,
then this does not imply that the computation can be forced into some state of
$s$. That is because for any rule $\trans q n a q$, $q \in s$, not all of the
states $q_1, \dotsc, q_n$ may be reached.  Moreover, it is too strong to
require that starting from leaf rules, it must be possible to force the
computation into all states of $\mathit{FCPre}_a(s)$. Clearly, it is enough if
the computation starting from leaf rules can be forced into $s$ via some of the
vectors in $\mathit{CPre}_a(s)$, not necessarily all of them. Also, if we keep
$\mathit{CPre}_a(s)$ for $s \subseteq Q$ as a set of vectors, we also have to
define the notion of controllable predecessors for sets of vectors of states,
which is a set of vectors of vectors of states, etc. Clearly, such an approach
is not practical and does not even terminate. Yet, we feel that some further research on ways possibly circumventing this problems can be interesting as we discuss in Section~\ref{sec:conclusion}.


\subsection{Tree Automata Language Inclusion Checking}

We are interested in testing language inclusion of two tree automata $\A =
(\Sigma,Q_\A,\Delta_\A,F_\A)$ and $\B = (\Sigma,Q_\B,\Delta_\B,F_\B)$. From
Lemma~\ref{lemma:TAlanguages}, we have that $\L(\A)\subseteq\L(\B)$ if and only if for
every tuple  $a_1, \ldots, a_n$ of leaf symbols from $\Sigma_0$,
$\L^\hole(\A)(I_{a_1}^\A,\ldots,I^\A_{a_n}) \subseteq
\L^\hole(\B)(I_{a_1}^\B,\ldots,I^\B_{a_n})$. In other words, for any
$a_1,\ldots, a_n\in\Sigma_0$, every context that can be accepted from a tuple
of states from $I_{a_1}^\A\times\ldots\times I^\A_{a_n}$ can also be accepted
from a tuple of states from $I_{a_1}^\B\times\ldots\times I^\B_{a_n}$. This
justifies a similar use of the notion of product-states as in Section
\ref{sec:language}. We define the language of a tuple of product-states as
$\L^\hole(\A,\B)((q_1,P_1),\ldots,(q_n,P_n)) := \L^\hole(\A)\tup q n \setminus
\L^\hole(\B)\tup P n$. Observe that we obtain that $\L(\A)\subseteq\L(\B)$ iff
the language of every $n$-tuple (for any $n\in\nat$) of product-states from the
set $\{(i,I_a^\B)\mid a\in\Sigma_0,i\in I_a^\A\}$ is empty.

Our algorithm for testing language inclusion of tree automata will check whether
it is possible to reach a product-state of the form $(q,P)$ with $q\in F_\A$ and
$P\cap F_\B = \emptyset$ (that we call accepting) from a tuple of product-states
from $\{(i,I_a^\B)\mid a\in\Sigma_0,i\in I_a^\A\}$. The following lemma allows
us to use Optimisation 1(a) and Optimisation 2 from Section~\ref{sec:language}.

\begin{lemma}\label{lemma:TAlanguage_p3}
Given ${\preceq}\in\rel {(\A\cup\B)}$, two tuples of states and two tuples
of pro\-duct-sta\-tes with $\tup p n\preceq\tup r n$ and
$\tup R n\preceq^{\forall\exists} \tup P n$, it holds that
$\L^\hole(\A,\B)((p_1,P_1),\ldots,(p_n,P_n)) \subseteq
\L^\hole(\A,\B)((r_1,R_1),\ldots,(r_n,R_n))$.
\end{lemma}

It is also possible to use Optimisation 1(b) where we stop searching from
product-states of the form $(q,P)$ such that $q\preceq r$ for some $r\in P$.
However, note that this optimisation is of limited use for tree automata. Under
the assumption that the automata $\A$ and $\B$ do not contain useless states,
the reason is that for any $q\in Q_\A$ and $r\in Q_\B$, if $q$ appears at
a left-hand side of some rule of arity more than 1, then no reflexive relation
from ${\preceq} \in \rel {(\A\cup\B)}$ allows $q\preceq
r$.\footnote{To see this, assume that a context
tree $t$ is accepted from $\tup q n\in Q_\A^n,q = q_i,1\leq i\leq n$. If
$q\preceq r$, then by the definition of $\preceq$,
$t\in\L^\hole(\A\cup\B)(q_1,\ldots,q_{i-1},r,q_{i+1},\ldots,q_n)$. However, that
cannot happen, as $\A\cup\B$ does not contain any rules with left hand sides
containing both states from $\A$ and states from $\B$. }

Algorithm~\ref{algorithm:TAlanguageinclusion} describes our method for checking
language inclusion of TA in pseudocode. It closely follows Algorithm
\ref{algorithm:languageinclusion}. It differs in two main points. First, the
initial value of the $\Next$ set is the result of applying the function
$\initialize$ on the set $\{(i,\minimize(I_a^\B))\mid a\in\Sigma_0, i\in
I_a^\A\}$ where $\initialize$ is the same function as in Algorithm
\ref{algorithm:languageinclusion}. Second, the computation of the $\post$ image
of a set of product-states means that for each symbol $a\in\Sigma_n,n\in\nat$,
we construct the $\post_a$-image of each $n$-tuple of product-states from the
set. Like in Algorithm~\ref{algorithm:TAuniversality}, we design the algorithm
such that we avoid computing the $\post_a$-image of a tuple more than once. We
define the post image $\post(\pstates)(r,R)$ of a set of product-states
$\pstates$ w.r.t. a product-state $(r,R)\in\pstates$. It is the set of all
product-states $(q,P)$ such that there is some $a\in\Sigma,\#(a) = n$ and some
$n$-tuple $((q_1,P_1),\ldots,(q_n,P_n))$ of product-states from $\pstates$ that
contains at least one occurrence of $(r,R)$ where $q\in\post_a\tup q n$ and
$P=\post_a\tup P n$.

\begin{algorithm}[t]
    \KwIn{TA $\mathcal{A}$ and $\mathcal{B}$ over an alphabet $\Sigma$. A relation ${\preceq} \in \rel{(\A\cup\B)}$.}
    \KwOut{$\true$ if $\lang{\mathcal{A}}\subseteq \lang{\mathcal{B}}$. Otherwise, $\false$.}
  \caption{\textit{Tree Automata Language Inclusion Checking}}
  \label{algorithm:TAlanguageinclusion}
    \lIf{there exists an accepting product-state in $\bigcup_{a\in\Sigma_0}\{(i,I_a^\B)\mid i\in I_a^\A\}$}{\KwRet{$\false$}}\;
    $\Processed$:=$\emptyset$\;
    $\Next$:=$\initialize(\bigcup_{a\in\Sigma_0}\{(i,\minimize(I_a^\B))\mid i\in I_a^\A\})$\;
    \While{$\Next\neq \emptyset$}{
        Pick and remove a product-state $(r,R)$ from $\Next$ and move it to $\Processed$\;
        \ForEach{$(p,P) \in \{ (r',\minimize(R')) \mid (r',R') \in \post(\Processed)(r,R)\}$}{
            \lIf{$(p,P)$ is an accepting product-state}{\KwRet{$\false$}}\;
            \ElseIf{$\neg\exists p'\in P$ s.t. $p\preceq p'$}{
                \If{$\neg\exists (q,Q) \in \Processed \cup \Next$ s.t. $p\preceq q \wedge Q \preceq^{\forall\exists} P$}{
                    Remove all $(q,Q)$ from $\Processed \cup \Next$ s.t. $q \preceq p \wedge P \preceq^{\forall\exists} Q$\;
                    Add $(p,P)$ to $\Next$\;
                }
            }
        }
    }
    \KwRet{$\true$}
\end{algorithm}

\paragraph{Correctness of the TA Language Inclusion Checking.}

We prove correctness of
Algorithm~\ref{algorithm:TAlanguageinclusion} in a very similar way to
Algorithm~\ref{algorithm:languageinclusion}, using suitably modified notions of
distances and ranks.

Let $\A = (\Sigma,Q_\A,\Delta_\A,F_\A)$
and $\B = (\Sigma,Q_\B,\Delta_\B,F_\B)$ be two tree automata.
Given $n\geq 0$ and an $n$-tuple of macro-states $((q_1,P_1),\ldots,(q_n,P_n))$,
we define
$
\mathbf{Dist}((q_1,P_1),\ldots,(q_n,P_n))=0
\ \mbox{iff} \ \epsilon \in \L^\hole(\A,\B)((q_1,P_1),\ldots,(q_n,P_n))
$.
Otherwise we define $\mathbf{Dist}((q_1,P_1),\ldots,(q_n,P_n))=k \in
\nat^+ \cup \{ \infty \}$
iff
$k = \min(\{
|t| \mid t \in T_n^\hole(\Sigma) \wedge
t \in \L^\hole(\A,\B)((q_1,P_1),\ldots,(q_n,P_n))
\})$. Here, we assume $\min(\emptyset) = \infty$.
For a set $\mathit{PStates}$ of product-states, we let
$\mathbf{Rank}(\mathit{PStates}) =
\min(\{ \mathbf{Dist}((q_1,P_1),\ldots,(q_n,P_n)) \mid
n \ge 1 \wedge \forall 1 \le i \le n: (q_i, P_i) \in \mathit{PStates}\})$.
The predicate $\incl(\mathit{PStates})$ is defined to be true iff
$\mathbf{Rank}(\mathit{PStates}) = \infty$.

\begin{lemma}\label{theorem:TAInclusionInvariant} The following two loop invariants hold
in Algorithm~\ref{algorithm:TAlanguageinclusion}:\begin{enumerate}

  \item $\neg\incl(\Processed \cup \Next) \implies
  \neg\incl(\bigcup_{a\in\Sigma_0}\{(i,I_a^\B)\mid i\in I_a^\A\})$.

  \item $\neg\incl(\bigcup_{a\in\Sigma_0}\{(i,I_a^\B)\mid i\in I_a^\A\}) \implies
  \mathbf{Rank}(\Processed) >\mathbf{Rank}( \Next \cup \Processed)$.
\end{enumerate}
\end{lemma}
The proof is similar to that of Lemma~\ref{theorem:TAUnivInvariant}.
With the invariants in hand, we can now prove Lemma~\ref{lemma:TAInclusionTermination} and Theorem~\ref{theorem:TAInclusioncorrectness} below in a very similar way as 
Lemma~\ref{lemma:termination} and Theorem~\ref{theorem:correctness}, respectively.

\begin{lemma}\label{lemma:TAInclusionTermination}
Algorithm~\ref{algorithm:TAlanguageinclusion} eventually terminates.\end{lemma}

\begin{theorem}\label{theorem:TAInclusioncorrectness} Algorithm
\ref{algorithm:TAlanguageinclusion} terminates, and returns
$\true$ iff $\L(\A)\subseteq \L(\B)$. \end{theorem}

\section{Experiments with Classical versus Pure Antichain Algorithms for Tree Automata}
\label{SecUniExprm}

In this section, we describe the experimental results obtained in
\cite{bouajjani:antichain} where we compare classical subset construction
based algorithms for tree automata with pure antichain based ones.  
The pure antichain algorithms
may be seen as special cases of Algorithms~\ref{algorithm:languageinclusion} and \ref{algorithm:TAlanguageinclusion},
where the role of simulation relation is played by the identity relation. 

We have implemented the above pure antichain approach for testing universality
and inclusion of tree automata in a prototype based on the Timbuk tree automata
library \cite{timbuk}. We give the results of our experiments  run on an Intel
Xeon processor at with 2.7GHz and 16GB of memory in
Fig.~\ref{FigUniExperiments}.  We ran our tests on randomly generated automata
and on automata obtained from abstract regular tree model checking applied in
verification of several pointer-manipulating programs. 

In the random tests, we use an approach for systematic generating random automata
with different parameters inspired by the approach proposed by Tabakov and
Vardi in \cite{tabakov:experimental} (which was also used in
\cite{wulf:antichains}).  The parameters of the generated automata are
\emph{number of states}, \emph{density of their
transitions} (the average number of different right-hand side states for a
given left-hand side of a transition rule, i.e., $|\Delta|/|\{a\vect q n
\mid \in\Sigma,q \in Q: \trans q n a q \}|$) and the \emph{density
of their final states} (i.e., $|F|/|Q|$).

\subsection{Experiments with Antichain-based Universality Checking}

For experiments with the pure antichain tree automata universality algorithm,
we used automata with 20 states and varied transition density and density of final states. 
Fig.~\ref{FigUniProbability} shows the probability of such tree automata
being universal, and Fig.~\ref{FigUniOverview}
the average times needed for checking their universality using our
antichain-based approach. The difficult instances are those where the
probability of being universal is about one half.  In
Fig.~\ref{FigUniSelected}, we show how the running times change for some
selected instances of the problem (in terms of some chosen densities of
transitions and final states, including those for which the problem is the most
difficult) when the number of states of the automata grows. We also show the
time needed when  universality is checked using determinisation, complement,
and emptiness checking. We see that the antichain-based approach behaves in a
significantly better way. The same conclusion can also be drawn from the
results of Fig.~\ref{FigUniARTMC} obtained on automata from experimenting with
abstract regular tree model checking applied for verifying various procedures
manipulating trees presented in Section~\ref{SecARTMCExprm}.

\begin{figure}[t!]
\begin{center}
\setlength{\subfigcapskip}{-1mm}
  \mbox{

    \subfigure[Probability that a tree automaton (TA) with 20 states and some
    density of transitions and final states is
    universal]{\label{FigUniProbability}
    \includegraphics[width=56mm]{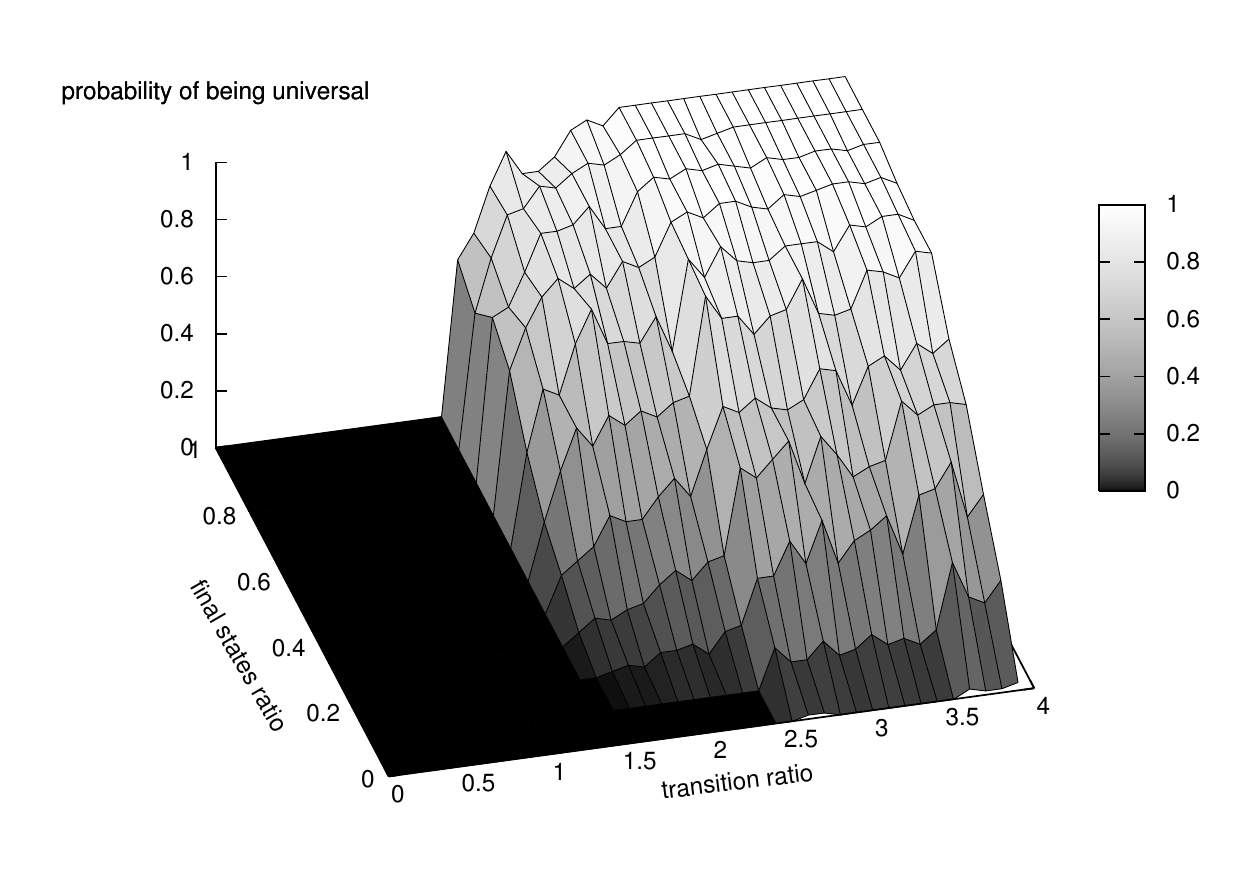}}

    \hspace*{1mm}

    \subfigure[Average times of antichain-based universality checking on TA with
    20 states and some density of transitions and final
    states]{\label{FigUniOverview}
    \includegraphics[width=61mm]{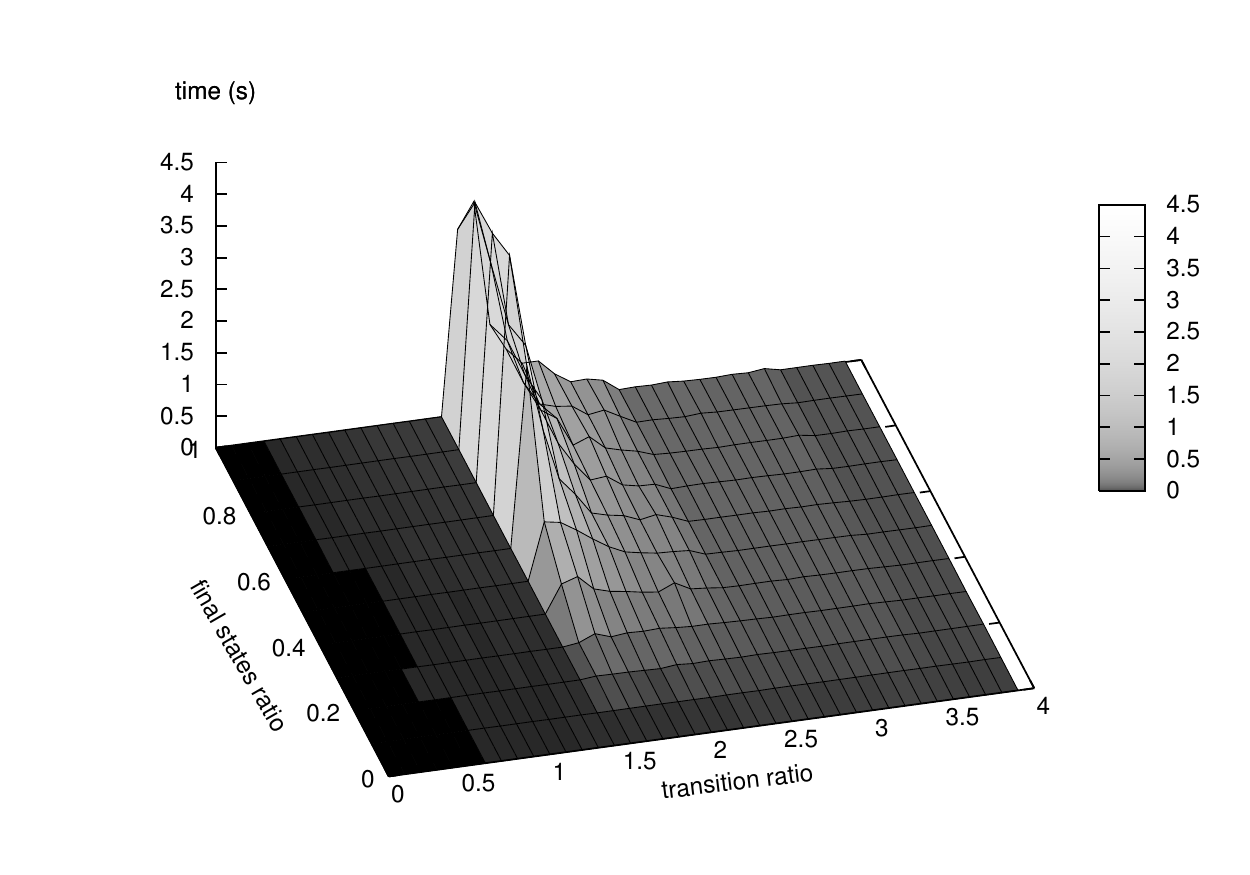}}}

  \mbox{\subfigure[Universality checking via determinisation and antichains on TA
    with selected densities of transitions and final states]
    {\label{FigUniSelected}
    \includegraphics[width=85mm]{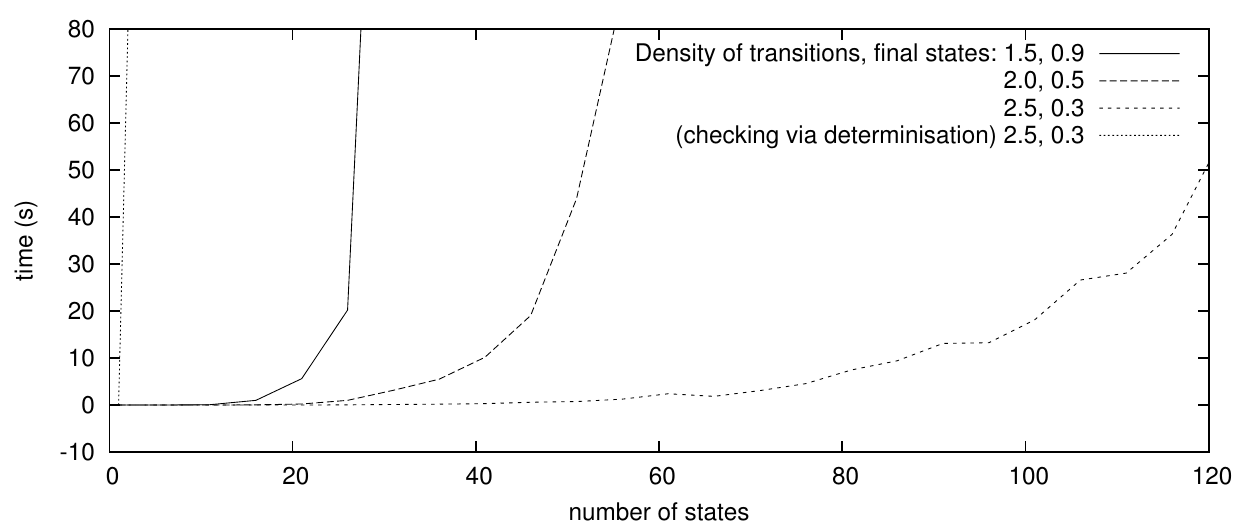}}}

  \mbox{\subfigure[Determinisation-based and antichain-based universality
    checking on TA from abstract regular tree model checking]
    {\label{FigUniARTMC}
    \includegraphics[width=85mm]{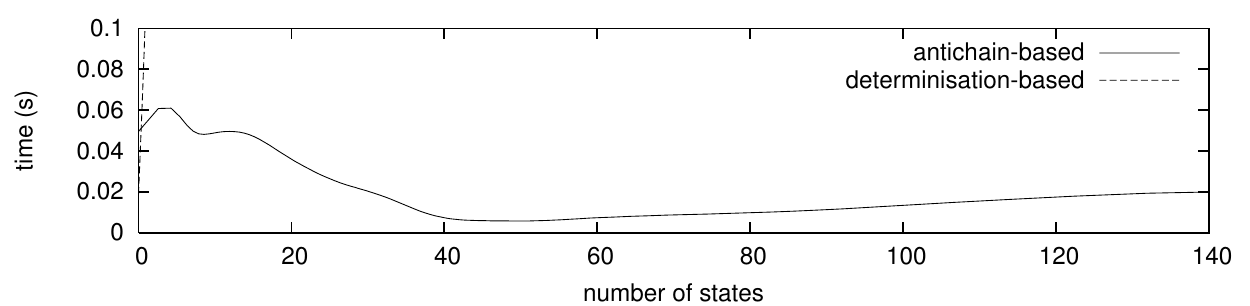}}}

  \caption{Experiments with universality checking on tree automata}
  \label{FigUniExperiments}
    
\end{center}
\end{figure}

\subsection{Experiments with  Antichain-based Inclusion Checking}

Below, in Fig.~\ref{FigInclExperiments1} and Fig.~\ref{FigInclExperiments2}, we
present the results that we have obtained from experimenting with pure antichain-based inclusion checking for tree automata.
We first ran our tests on pairs of randomly generated automata having 10 states
and different possible densities of transitions and final states. The probability
that $\mathcal{L}(\A) \subseteq \mathcal{L}(\B)$ holds for
randomly generated tree automata $\A$ and $\B$ (both having
the same densities of transitions and final states) is shown in
Fig.~\ref{FigInclProbability}. Fig.~\ref{FigInclOverview} then shows how the
antichain-based inclusion checking behaves on such automata. We see that its time
consumption is naturally growing for automata where the probability of whether
$\mathcal{L}(\A) \subseteq \mathcal{L}(\B)$ holds is
neither too low nor too high.

\begin{figure}[ht]
\begin{center}
\setlength{\subfigcapskip}{0mm}
  \mbox{

    \subfigure[Probability of $\mathcal{L}(\A) \subseteq
    \mathcal{L}(\B)$ for tree automata (TA) with 10 states and some
    density of transitions and final states]{\label{FigInclProbability}
    \includegraphics[width=58mm]{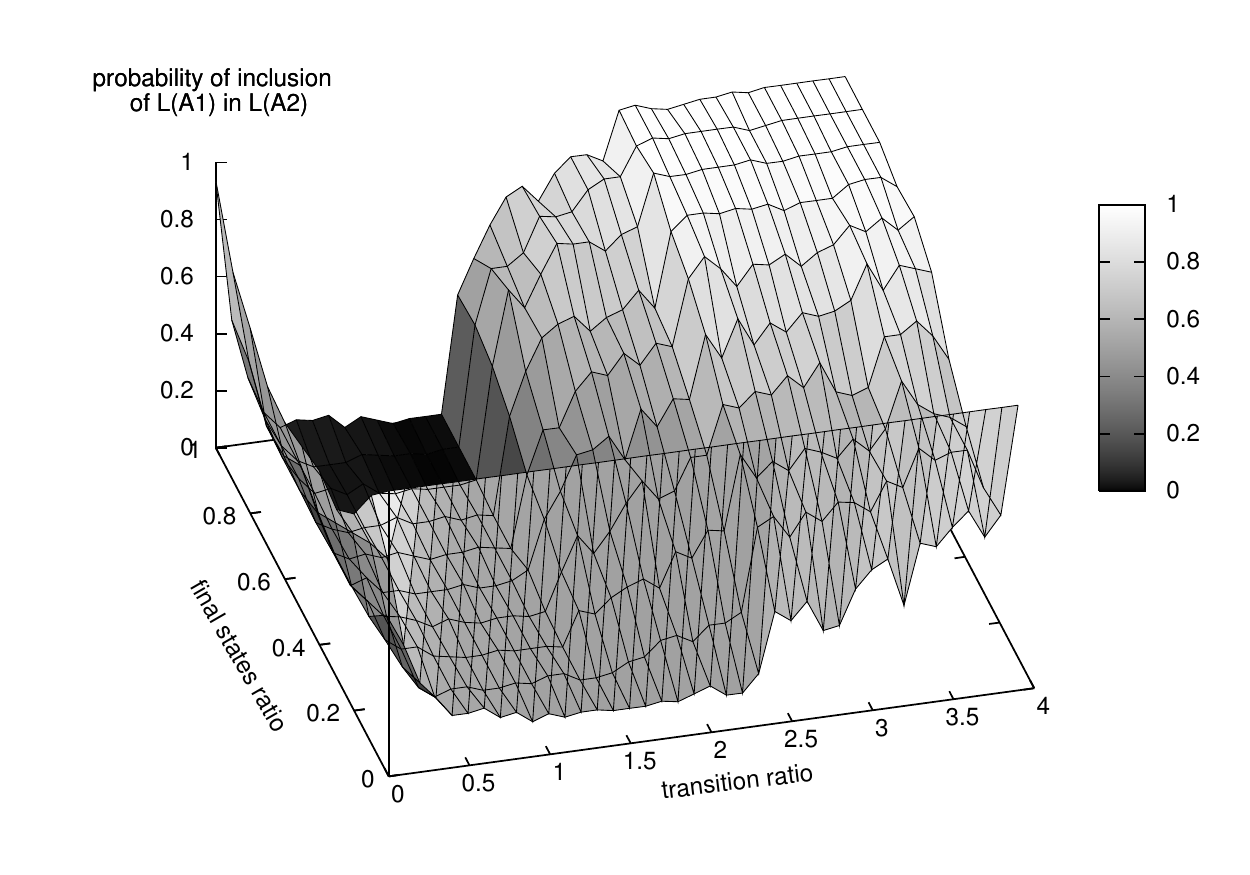}}

    \hspace*{-0.2mm}

    \subfigure[Average times of antichain-based inclusion checking on TA with
    some density of transitions and final states]{\label{FigInclOverview}
    \includegraphics[width=63mm]{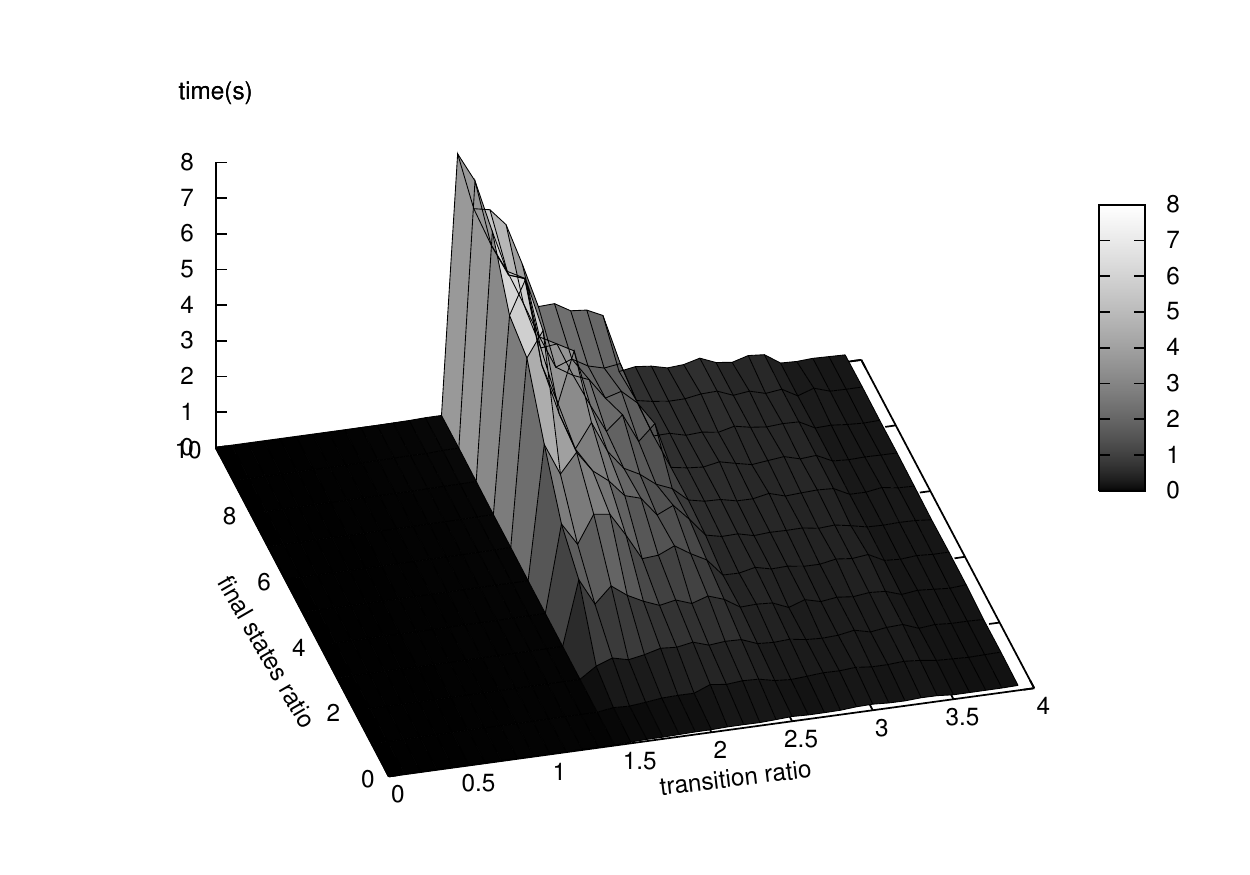}}}

  \mbox{

    \subfigure[Antichain-based inclusion checking on TA, $\A$ random,
    $\B$ with some density of transitions and final
    states]{\label{FigInclA1Random}
    \includegraphics[width=63mm]{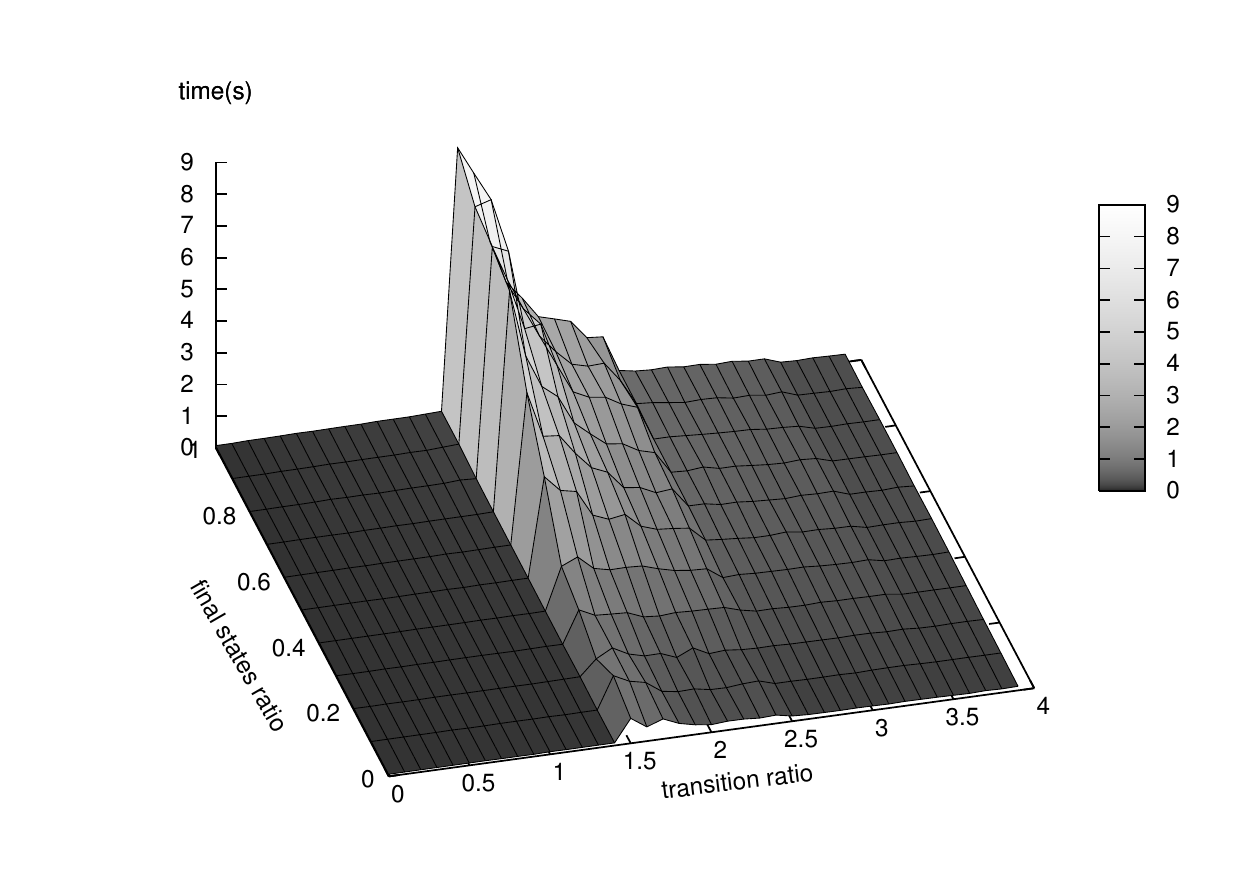}}

    \hspace*{-0.2mm}

    \subfigure[Antichain-based inclusion checking on TA, $\B$ random,
    $\A$ with some density of transitions and final
    states]{\label{FigInclA2Random}
    \includegraphics[width=63mm]{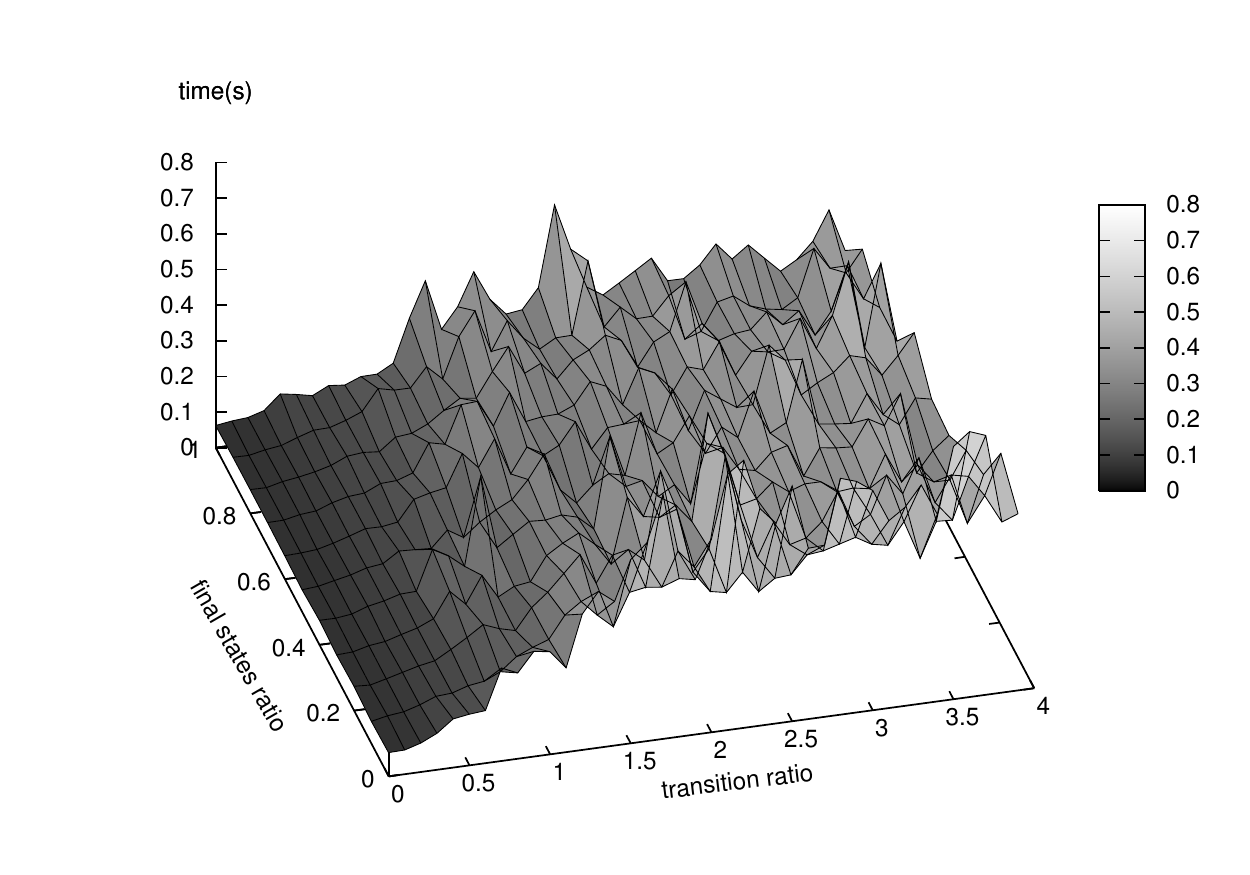}}}

    \caption{Experiments with inclusion checking on tree automata}
  \label{FigInclExperiments1}

\end{center}
\end{figure}

Fig.~\ref{FigInclA1Random} and Fig.~\ref{FigInclA2Random} show what happens if
either $\A$ or $\B$ is left completely random, and only
$\B$ or $\A$, respectively, follows a given density of
transitions and final states. The fact that the results in
Fig.~\ref{FigInclA1Random} follow Fig.~\ref{FigInclOverview}, whereas the time
consumption in Fig.~\ref{FigInclA2Random} is roughly implied by the size of
$\A$ (in terms of transitions), implies that the time consumption of
the antichain-based inclusion checking is---as expected---influenced much more by
the automaton $\B$.

Finally, in Fig.~\ref{FigInclSelected}, we show how the running times change for
some selected instances of the problem (in terms of some selected densities of
transitions and final states, including those for which the problem is the most
difficult) when the number of states of the automata starts growing. The figure
also shows the time needed when the inclusion checking is based on determinising
and complementing $\B$ and checking emptiness of the language
$\mathcal{L}(\A) \cap \overline{\mathcal{L}(\B)}$. We see
that the antichain-based approach really  behaves in a very significantly better
way. The same conclusion can then be drawn also from the results shown in
Fig.~\ref{FigInclARTMC} that we obtained on automata saved from experimenting
with abstract regular tree model checking applied for verifying various real-life
procedures manipulating trees (cf. Section~\ref{SecARTMCExprm}). In fact, the
antichain-based inclusion checking allowed us to implement an abstract regular
tree model checking framework entirely based on nondeterministic tree automata
which is significantly more efficient than the framework based on deterministic
automata.

\begin{figure}[ht]
\begin{center}
\setlength{\subfigcapskip}{0mm}
    
  \mbox{\subfigure[Determinisation-based and antichain-based inclusion checking
    on TA with selected densities of transitions and final states]
    {\label{FigInclSelected}
    \includegraphics[width=100mm]{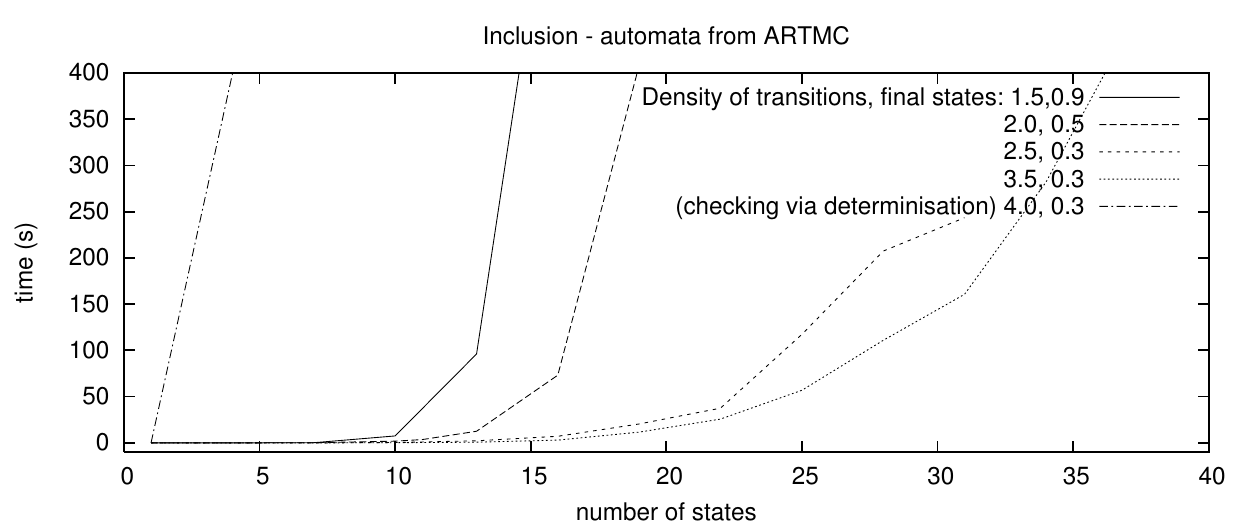}}}

  \mbox{\subfigure[Determinisation-based and antichain-based inclusion checking
    on TA from abstract regular tree model checking] {\label{FigInclARTMC}
    \includegraphics[width=100mm]{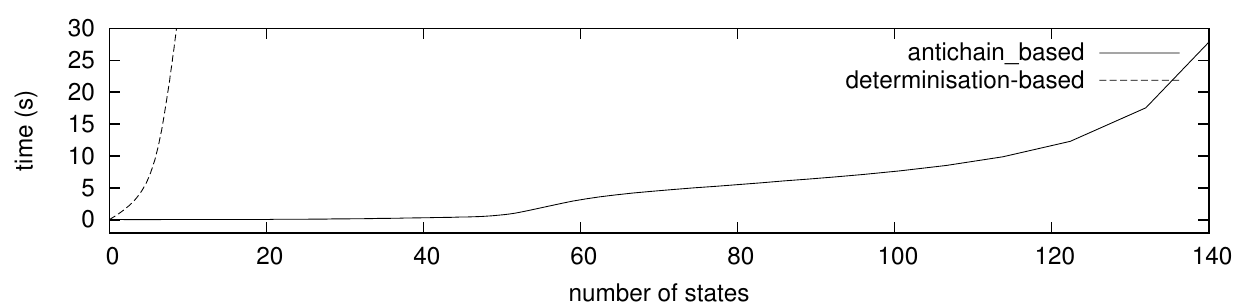}}}

  \caption{Further experiments with inclusion checking on tree
  automata}
  \label{FigInclExperiments2}
  
\end{center}
\end{figure}

\subsection{Experiments with Regular Tree Model Checking}\label{section:rtmc}

We now present our experiments with regular tree model checking that
illustrate practical applicability of the language inclusion testing algorithms
and the tree automata reduction algorithms from
Chapter~\ref{chapter:ta_reduction}. We will show how the two techniques
allow us to build the (abstract) regular tree model checking on nondeterministic tree automata
instead of on deterministic ones which greatly improves efficiency of the
method. 

\paragraph{Nondeterministic Abstract Regular Tree Model Checking.}

As is clear from the definition of $\hat{\tau}$ in Section~\ref{section:artmc}, ARTMC was
originally defined for and tested on \emph{minimal deterministic} tree automata
(DTA). However, the various experiments done showed that the determinisation step
is a significant bottleneck. To avoid it and to implement ARTMC using
nondeterministic tree automata (TA), we need the following operations over TA:
(1)~application of the transition relation $\tau$, (2)~union, (3)~abstraction and
its refinement, (4)~intersection with the set of bad configurations,
(5)~emptiness, and (6)~inclusion checking (needed for testing if the abstract
reachability computation has reached a fixpoint). 
Finally, (7)~a~method to reduce
the size of the computed TA is also desirable---$\hat{\tau}(\mathcal{A})$ is
then redefined to be the reduced version of the TA obtained from an application
of $\tau$ on an TA $\mathcal{A}$. We note that the method would in theory work without reduction methods too. However, often hundreds of the steps (1) to (6) are performed within a single verification run, and most of them
increases the size of automata\footnote{Some abstraction methods reduce the size of automata too, however, not sufficiently enough to outweigh the increase of size caused by the other steps.}. Therefore, good reduction techniques are in
fact crucial since the size of automata tends to explode which
reduces scalability of the method.

%
%

An implementation of Points (1), (2), (4), and (5) is easy. Moreover, concerning
Point (3), the abstraction mechanisms of \cite{bouajjani:abstractTree} can be lifted to work on
TA in a straightforward way while preserving their guarantees to be finitary,
overapproximating, and the ability to exclude spurious counterexamples.
Furthermore, Chapter~\ref{chapter:ta_reduction} gives efficient algorithms for
reducing TA based on computing suitable simulation equivalences on their states,
which covers Point (7). Hence, the last obstacle for implementing
nondeterministic ARTMC was Point (6), i.e., the need to efficiently check
inclusion on TA. We have solved this problem by Algorithm~\ref{algorithm:TAlanguageinclusion}, which allowed us to implement a nondeterministic ARTMC
framework in a prototype tool and test it on suitable examples. Below, we present
the first very encouraging results that we have achieved.
We note that we were so far considering only the pure antichains where the role of simulation within Algorithm~\ref{algorithm:TAlanguageinclusion} is played only by the identity relation%
\footnote{
We have not yet managed to incorporate simulation enhanced antichain  algorithms into the framework of ARTMC. We plan to use them in the further prototype tools that we mention in Section~\ref{sec:conclusion}.  We believe that the overall impact of the simulation subsumption technique will be positive, judging from the experience that we have gathered and that is presented in Section~\ref{sec:experiments}.
}.


\paragraph{Experiments with Nondeterministic ARTMC.}\label{SecARTMCExprm}

We have implemented the version of ARTMC framework based on nondeterministic tree automata using the Timbuk tree
automata library \cite{timbuk} and compared it with an ARTMC implementation based on the
same library, but using DTA. In particular, the deterministic ARTMC framework
uses determinisation and minimisation after computing the effect of each
forward or backward step to try to keep the automata as small as possible and
to allow for easy fixpoint checking: The fixpoint checking on DTA is not based
on inclusion, but identity checking on the obtained automata (due to the fact
that the computed sets are only growing and minimal DTA are canonical). For
TA, the tree automata reduction from Chapter~\ref{chapter:ta_reduction} that we use does
not yield canonical automata, and so the antichain-based inclusion checking is
really needed.

We have applied the framework to verify several procedures manipulating dynamic
tree-shaped data structures linked by pointers. The trees being manipulated are
encoded directly as the trees handled in ARTMC, each node is labelled by the data
stored in it and the pointer variables currently pointing to it. All program
statements are encoded as (possibly non-structure preserving) tree transducers.
The encoding is~fully~automated. The only allowed destructive pointer updates
(i.e., pointer manipulating statements changing the shape of the tree) are tree
rotations \cite{cormen:introduction} and addition of new~leaf~nodes.

We have in particular considered verification of the depth-first tree traversal
and the standard procedures for rebalancing red-black trees after insertion or
deletion of a leaf node \cite{cormen:introduction}. We have verified that the programs do
not manipulate undefined and null pointers in a faulty way. For the procedures on
red-black trees, we have also verified that their result is a red-black tree
(without taking into account the non-regular balancedness condition). In general,
the set of possible input trees for the verified procedures as well as the set of
correct output trees were given as tree automata. In the case of the procedure
for rebalancing red-black trees after an insertion, we have also used a generator
program preceding the tested procedure which generates random red-black trees and
a tester program which tests the output trees being correct. Here, the set of
input trees contained just an empty tree, and the verification was reduced to
checking that a predefined error location is unreachable. The size of the
programs ranges from 10 to about 100 lines of pure pointer manipulations.

The results of our experiments on an Intel Xeon processor at 2.7GHz with 16GB of
available memory (as in Section~\ref{SecUniExprm}) are summarised in
Table~\ref{TabARTMC}. The predicate abstraction proved to give much better
results (therefore we do not consider the finite-height abstraction here). The
abstraction was either applied after firing each statement of the program (``full
abstraction'') or just when reaching a loop point in the program (``restricted
abstraction''). The results we have obtained are very encouraging and show a
significant improvement in the efficiency of ARTMC based on nondeterministic tree
automata. Indeed, the ARTMC framework based on deterministic tree automata has
either been significantly slower in the experiments (up to $25$-times) or has
completely failed (a too long running time or a lack of memory)---the latter case
being quite frequent.

\newcommand{\failed}{$\times$}
\newcommand{\detc}{\makebox[13mm]{det.}}
\begin{table}[t]
\renewcommand{\tabcolsep}{5pt}
\begin{center}

  \caption{Running times (in sec.) of det. and nondet. ARTMC applied for
  verification of various tree manipulating programs ($\times$ denotes a too
  long run or a failure due to a lack of memory)}
  \label{TabARTMC}

\renewcommand{\tabcolsep}{1.5mm}
  \begin{tabular}{|r||c|c||c|c||c|c|}

    \hline &\multicolumn{2}{|c||}{\rule{0pt}{1.5em}DFT} &
    \multicolumn{2}{|c||}{\parbox{20mm}{\centering{RB-delete\\(null,undef)}}} &
    \multicolumn{2}{|c|}{{\parbox{20mm}{\centering{RB-insert\\(null,undef)}}}}\\

    \hline

    \hline

    & \detc & nondet. & \detc & nondet. & \detc & nondet. \\

    \hline

    full abstr.       & 5.2 & 2.7 & \failed & \failed & 33  & 15  \\

    \hline

    restricted abstr. & 40  & 3.5 & \failed & 60 & 145 & 5.4 \\

    \hline \hline

    &
    \multicolumn{2}{c||}{\rule{0pt}{1.5em}\parbox{30mm}
    {\centering{RB-delete\\(RB preservation)}}} & \multicolumn{2}{|c||}
    {\parbox{30mm}{\centering{RB-insert\\(RB preservation)}}} 
    & \multicolumn{2}{|c|}{\parbox{20mm}{\centering{RB-insert\\(gen.,
    test.)}}}\\

    \hline

    \hline

    & det. & nondet. & det. & nondet. & det. & nondet. \\

    \hline

    full abstr. & \failed & \failed & \failed & \failed & \failed &\failed \\

    \hline

    restricted abstr.&\failed  & 57 &\failed  & 89 &\failed  & 978 \\

    \hline

  \end{tabular}

\end{center}
\end{table}

\section{Experiments with Pure versus Simulation Enhanced Antichain Algorithms.}
\label{sec:experiments}

In this section, we describe the experimental result obtained in
\cite{abdulla:when} where we compare pure antichain algorithms for FA and TA
with simulation enhanced antichain algorithms.  
Recall that by pure antichain algorithms we mean algorithms published in
\cite{wulf:antichains} for FA and in \cite{bouajjani:antichain} for TA that
may be seen as special cases of Algorithms~\ref{algorithm:universality},
\ref{algorithm:languageinclusion}, \ref{algorithm:TAuniversality}, and
\ref{algorithm:TAlanguageinclusion} where the role of simulation relation is
played by the identity relation. Notice that in this case, only Optimisation 1
comes to play within  Algorithms \ref{algorithm:universality} and
\ref{algorithm:TAuniversality} for checking universality, and only Optimisation
1(a) applies within Algorithms \ref{algorithm:languageinclusion} and
\ref{algorithm:TAlanguageinclusion} for checking language inclusion. Since $\preceq$ is the identity relation, Checking
the relation $\preceq^{\forall\exists}$ on sets of states is then replaced be
checking subset inclusion.

We concentrated on experiments with inclusion checking,
since it is more common than universality checking in various symbolic
verification procedures, decision procedures, etc. We compared our approach,
parametrised by maximal simulation (or, for tree automata, maximal upward
simulation), with the previous pure antichain-based approach of
\cite{wulf:antichains,bouajjani:antichain}, and with classical
subset-construction-based approach.  We implemented all the above in OCaml. We
used the algorithm in~\cite{holik:optimizing} for computing maximal
simulations.  In order to make the figures easier to read, we often do not show
the results of the classical algorithm. The reason is that in all of the
experiments, the classical algorithm performed much worse than the other two
approaches that these experiments are primarily directed to compare.

We note that we have also done some preliminary experiments with random
automata generated according to the framework by Vardi and Tabakov in the same
way as in the previous section. Sadly, for this type of automata, the simulation
optimisation give almost no speedup. It seems that for the hard areas of the
space of settings of parameters of the generator, simulation is very sparse and
the speedup that it gives hardly compensates the time needed for computing the
simulation itself.  On the other hand, for the easy settings, pure antichain
algorithms finish too fast and the time needed for computing simulation
dominates. Therefore, we decided to perform more experiments with automata that
have more structure such as those from the sources described above and which
are also closer too real life applications than the random ones.  As we will
see, for these automata the simulation optimisations really help.

\subsection{Experiments on FA}\label{sec:experiments-FA}

\begin{figure}[t]
\begin{center}
\subfigure[Detailed results]{
    \includegraphics[width=10cm]{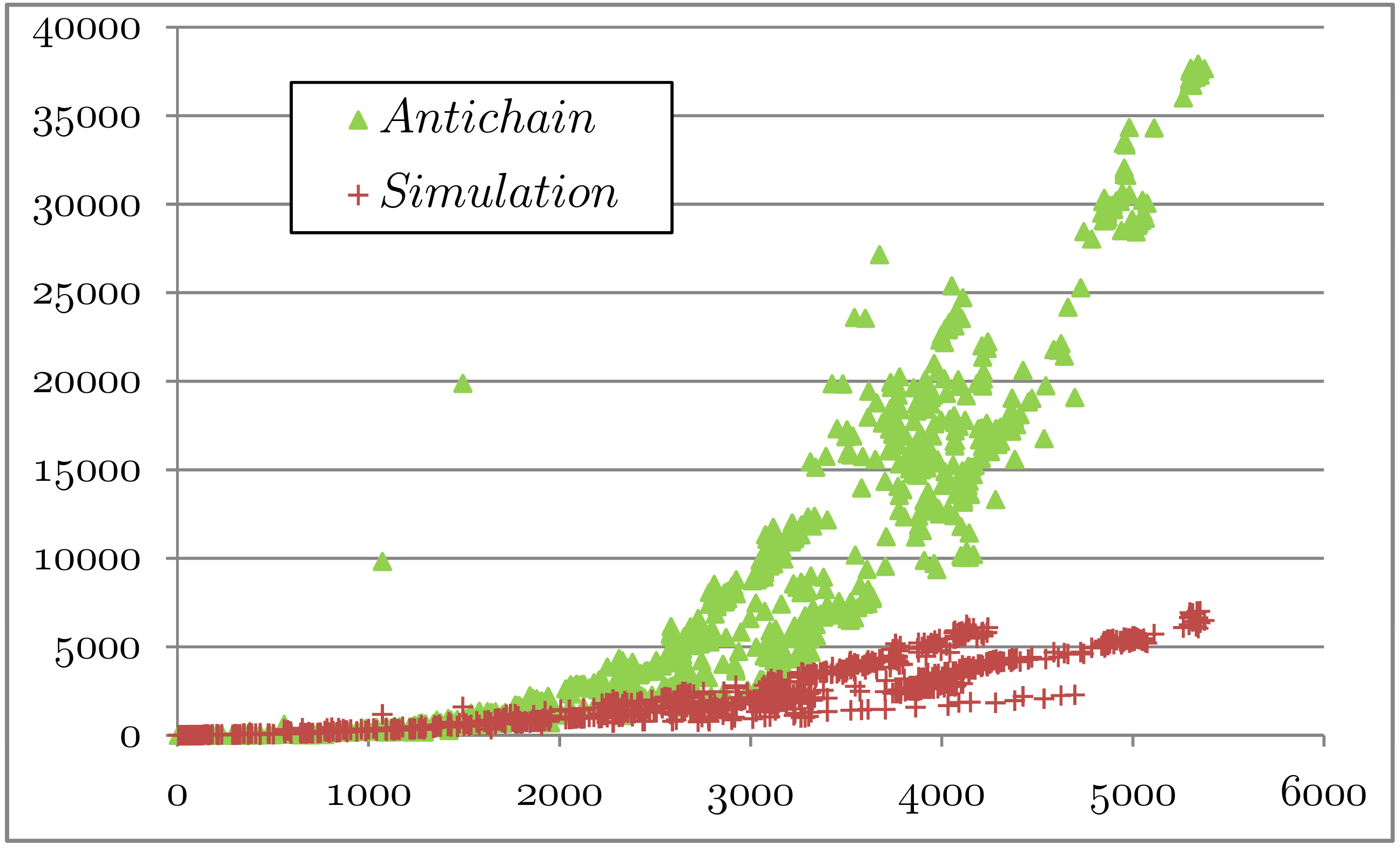}
}

\subfigure[Average execution time for different FA pair sizes (in seconds)]{
\centering
    \begin{tabular}{|lcl|D{.}{.}{3}|D{.}{.}{3}|}
      \hline
    Size & &  & \multicolumn{1}{c|}{Antichain} & \multicolumn{1}{c|}{Simulation}\\
      \hline
    0 & - & 1000 &  0.059	&0.099 \\
    1000 & - & 2000& 1.0	&0.7\\
    2000 & - & 3000 & 3.6	&1.69\\
    3000 & - & 4000 &11.2	&3.2\\
    4000 & - & 5000 &20.1	&4.79\\
    5000 & - &    &33.7	&6.3\\
      \hline
    \end{tabular}
}
  \caption{Language inclusion checking on FA generated from a regular model checker}
  \label{figure:lc_ARMCexp}
\end{center}
\end{figure}

For language inclusion checking of FA, we compared the simulation enhanced
approach that corresponds to Algorithm~\ref{algorithm:languageinclusion}
against the former pure antichain approach that corresponds to the same
algorithm but with the simulation relation being identity.  We tested the two
on examples generated from the intermediate steps of a tool for abstract
regular model checking~\cite{bouajjani:abstract}. In total, we have 1069 pairs
of FA generated from different verification tasks, which included verifying a
version of the bakery algorithm, a system with a parametrised number of
producers and consumers communicating through a double-ended queue, the bubble
sort algorithm, an algorithm that reverses a circular list, and a Petri net
model of the readers/writers protocol (cf.
\cite{bouajjani:abstract,bouajjani:verifying} for a detailed description of the
verification problems). In Fig.~\ref{figure:lc_ARMCexp}~(a), the horizontal
axis is the sum of the sizes of the pairs of automata\footnote{We measure the
size of the automata as the number of their states.} whose language inclusion
we check, and the vertical axis is the execution time (the time for computing
the maximal simulation is included).  Each point denotes a result from
inclusion testing for a pair of FA.  Fig.~\ref{figure:lc_ARMCexp}~(b) shows
the average results for different FA sizes. From the figure, one can see that
our approach has a much better performance than the antichain-based one. Also,
the difference between our approach and the antichain-based approach becomes
larger when the size of the FA pairs increases. If we compare the average
results on the smallest 1000 FA pairs, our approach is 60\% slower than the
the antichain-based approach. For the largest FA pairs (those with size larger
than 5000), our approach is 5.32 times faster than the the antichain-based
approach. We note that the time needed for computing simulation is always
included in the overall running time of the simulation enhanced algorithm.

\begin{figure}
\begin{center}
  \subfigure[Language inclusion does not always hold]
  {\label{figure:lc_REexp_nothold}
    \includegraphics[width=10cm]{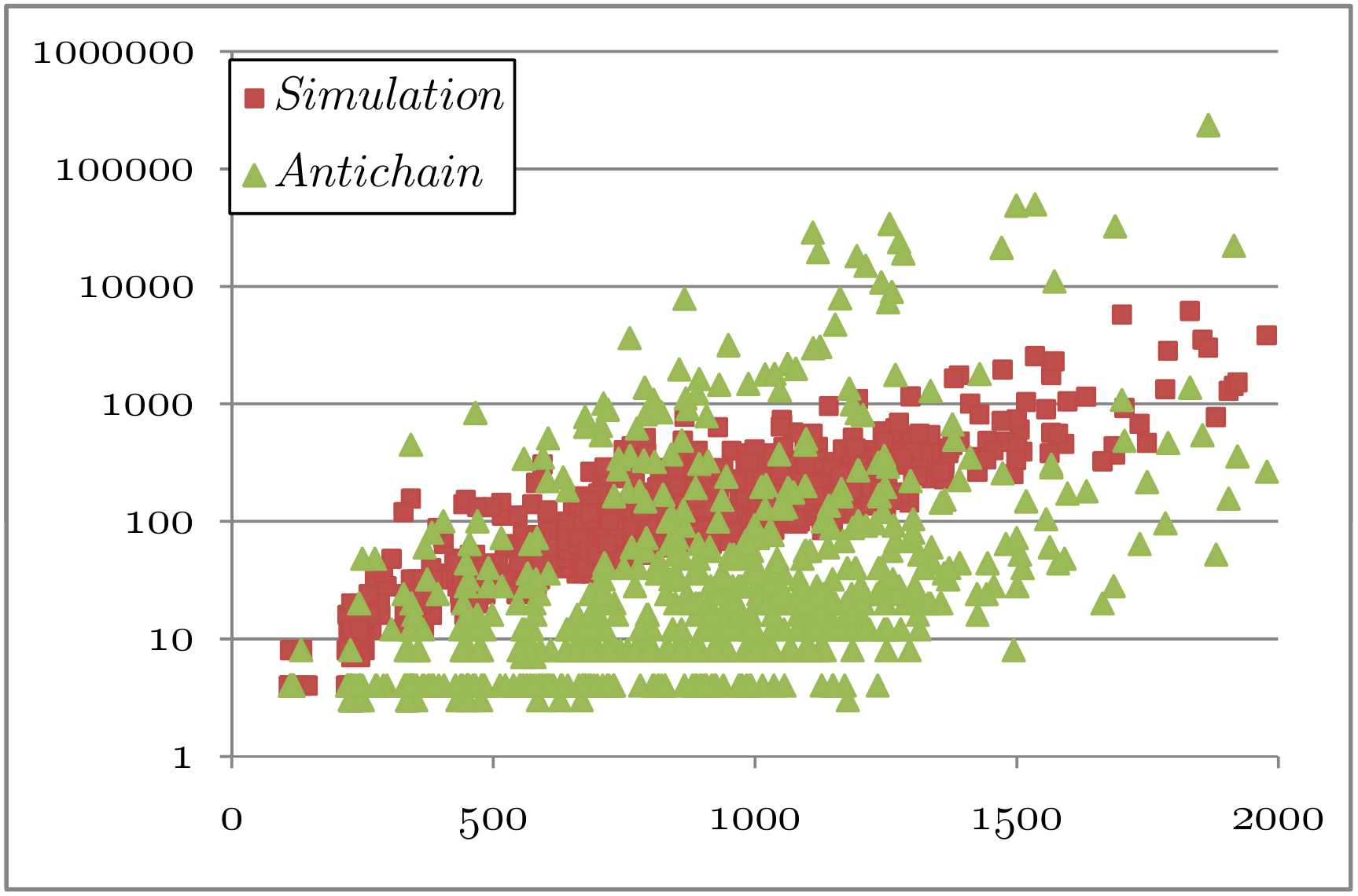}
  }
  \subfigure[Language inclusion always holds]
  {\label{figure:lc_REexp_hold}
    \includegraphics[width=10cm]{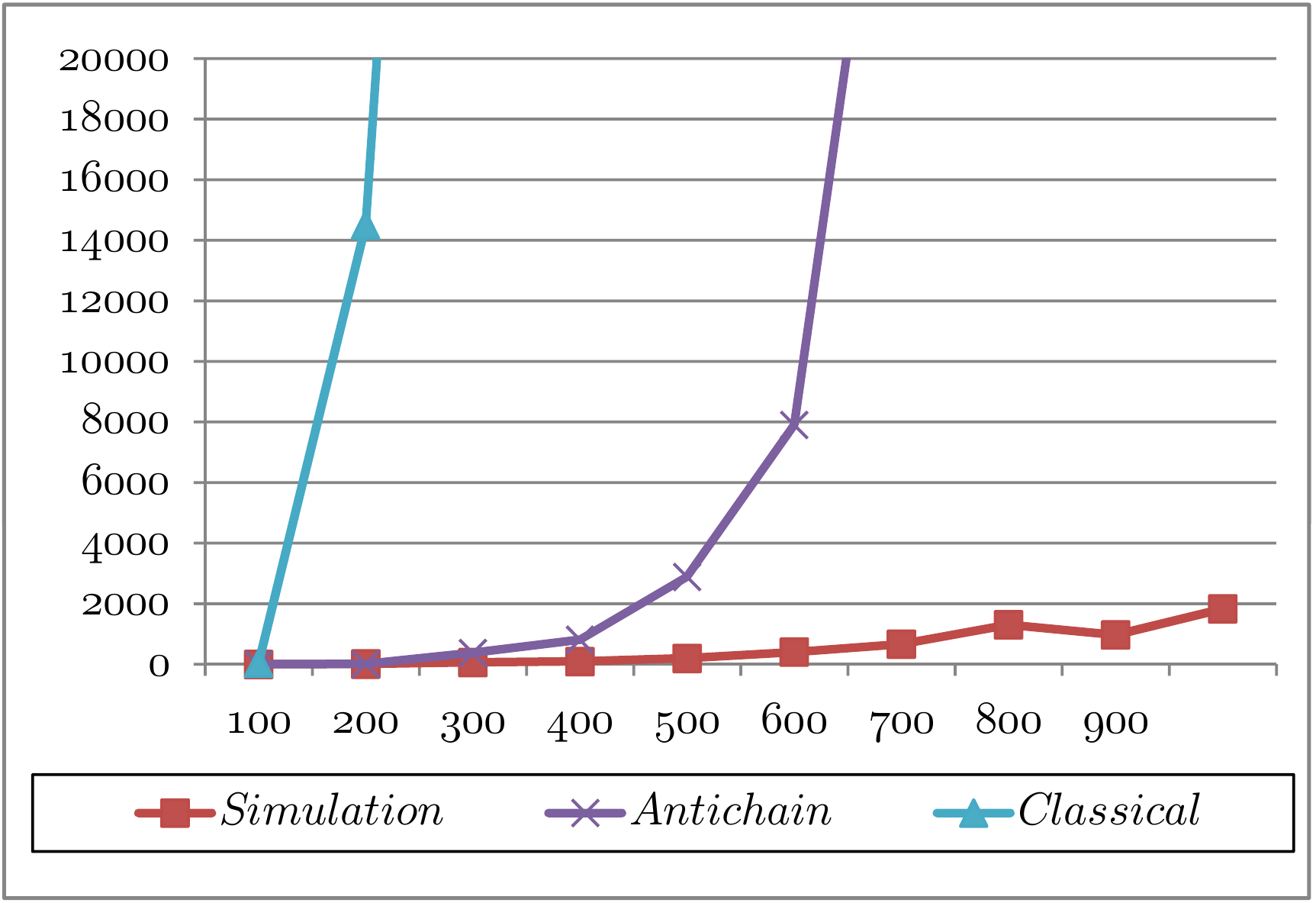}
  }
  \caption{Language inclusion checking on FA generated from regular expressions}
  \label{figure:lc_REexp}
\end{center}
\end{figure}

We also tested our approach using FA generated from random regular expressions.
We have two different tests: (1) language inclusion does not always hold and (2)
language inclusion always holds\footnote{To get a sufficient number of tests for
the second case, we generate two FA $\mathcal{A}$ and $\mathcal{B}$ from random
regular expressions, build their union automaton $\mathcal{C} = \mathcal{A} \cup
\mathcal{B}$, and test $\lang{\mathcal{A}} \subseteq \lang{\mathcal{C}}$.}. The
result of the first test is in Fig.~\ref{figure:lc_REexp_nothold}. In the
figure, the horizontal axis is the sum of the sizes of the pairs of automata
whose language inclusion we check, and the vertical axis is the execution time
(the time for computing the maximal simulation is included). From
Fig.~\ref{figure:lc_REexp_nothold}, we can see that the performance of our
approach is much more stable. It seldom produces extreme results. In all of the
cases we tested, it always terminates within 10 seconds. In contrast, the
antichain-based approach needs more than 100 seconds in the worst case. The
result of the second test is in Fig.~\ref{figure:lc_REexp_hold} where the
horizontal axis is the length of the regular expression and the vertical axis is
the average execution time of 30 cases in milliseconds. From
Fig.~\ref{figure:lc_REexp_hold}, we observe that our approach has a much better
performance than the antichain-based approach if the language inclusion holds.
When the length of the regular expression is 900, our approach is almost 20
times faster than the antichain-based approach.

\begin{figure}
\begin{center}
    \includegraphics[width=10cm]{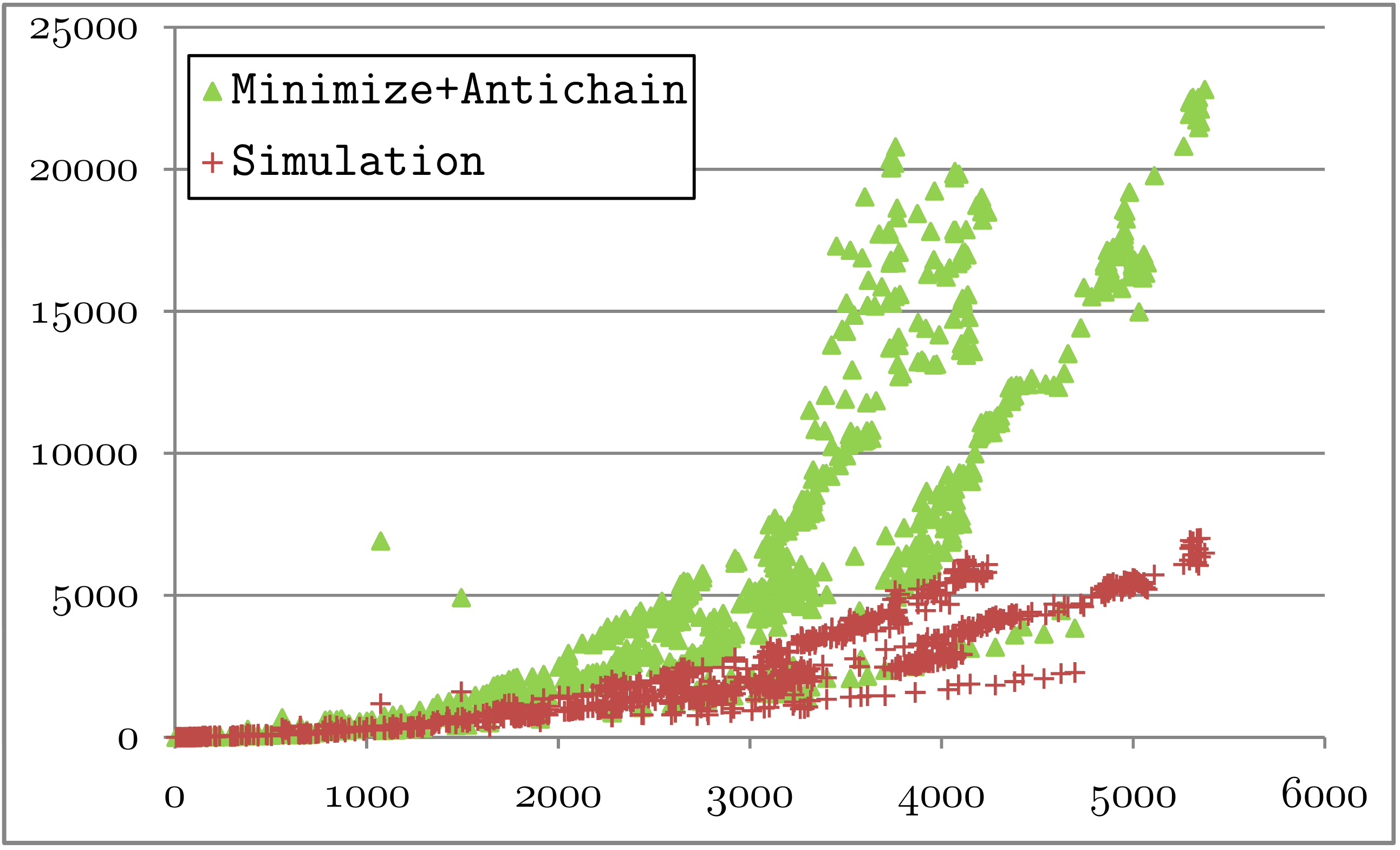}
    \caption{Compare the performance of our approach with minimise + antichain}
  \label{figure:lc_ARMCreduce}
\end{center}
\end{figure}

When the maximal simulation relation $\preceq$ is given, a natural way to
accelerate the language inclusion checking is to use $\preceq$ to minimise the
size of the two input automata by merging $\preceq$-equivalent states. In this
case, the simulation relation becomes sparser. A~question arises whether our
approach has still a better performance than the antichain-based approach in
this case. Therefore, we also evaluated our approach under this setting. Here
again, we used the FA pairs generated from abstract regular model
checking~\cite{bouajjani:abstract}. The results presented at Figure~\ref{figure:lc_ARMCreduce} show that although the antichain-based approach
gains some speed-up (compare with Figure~\ref{figure:lc_ARMCexp}) when combined with minimisation, it is still slower than
our approach. The main reason is that in many cases, simulation holds only in
one direction, but not in the other. Our approach can also utilise this type of
relation. In contrast, the minimisation algorithm merges only simulation
equivalent states.

\subsection{Experiments on TA}

For language inclusion checking of TA, we tested our approach on 86 tree
automata pairs generated from the intermediate steps of a regular tree model
checker from Section~\ref{SecARTMCExprm} while verifying the algorithm of rebalancing
red-black trees after insertion or deletion of a~leaf node. We were again
comparing simulation enhanced antichain approach that corresponds to Algorithm
\ref{algorithm:TAlanguageinclusion} with the pure antichain approach that corresponds to the same
algorithm but with the simulation relation being the identity.  The results
are given in Table~\ref{figure:lc_Tree}. Our approach has a~much better
performance when the size of a TA pair is large. For TA pairs of size smaller
than 200, our approach is on average 1.39 times faster than the antichain-based
approach. However, for those of size above 1000, our approach is on average 6.8
times faster than the antichain-based approach.

\begin{table}
\begin{center}
    \begin{tabular}{|rcr|D{.}{.}{2}|D{.}{.}{4}|c|r|}
      \hline
    \multirow{2}{*}{Size} & &  & \multicolumn{1}{c|}{Antichain} &
    \multicolumn{1}{c|}{Simulation} &
    \multirow{2}{*}{Diff.}&\multirow{2}{*}{\# of Pairs}\\
           & & & (sec.) & (sec.) & &\\
      \hline
    0 & -- & 200 &1.05&0.75&140\%&29 \\
    200 & -- & 400&11.7&4.7&246\%&15\\
    400 & -- & 600&65.2&19.9&328\%&14\\
    600 & -- & 800&3019.3&568.7&531\%&13\\
    800 & -- & 1000&4481.9&840.4&533\%&5\\
    1000 & -- & 1200&11761.7&1720.9&683\%&10\\
      \hline
    \end{tabular}\\
    \caption{Language inclusion checking on TA}
  \label{figure:lc_Tree}
\end{center}
\end{table}

\section{Conclusions and Future Work}\label{sec:conclusion}

We presented algorithms for finite word and tree automata universality and
language inclusion checking that combine the antichain principle from
\cite{wulf:antichains} with a use of simulation relations (forward simulation in the case of FA and upward simulation parametrised by identity in the case of TA).  
The algorithms have
been thoroughly tested both on randomly generated automata and on automata
obtained from various verification runs of the ARTMC framework. The new
algorithms are significantly more efficient than the pure antichain algorithms
from \cite{wulf:antichains} and \cite{bouajjani:antichain}.

In the case of TA, we also presented experimental results from our previous
work \cite{bouajjani:antichain} on pure antichain tree automata versions of the
algorithms from \cite{wulf:antichains} which preceded the work on their
versions improved with simulation presented here. We compare these algorithms with the
classical subset construction-based algorithm and we conclude that similarly as
shown in \cite{wulf:antichains} for FA, the antichain technique fundamentally
improves performance of universality and language inclusion checking over tree
automata.  Moreover, using the proposed pure antichain-based inclusion checking
algorithm together with our simulation based reduction methods from
Chapter~\ref{chapter:ta_reduction}, we have implemented a complete ARTMC
framework based on nondeterministic tree automata and tested it on verification
of several real-life pointer-intensive procedures. The results show a very
encouraging improvement in the capabilities of the framework.

We are considering several directions of future work. First, our simulation based
improvements of antichain algorithms is based on relatively simple and natural
principles and we believe that these techniques can be developed for
other classes of automata. We have already done the first attempt  in
\cite{abdulla:simulationsubsumption} where we have successfully combined the Ramsey based approach
to universality and inclusion checking for B\"uchi automata with simulations.  

Next, we have already proven first results showing that it is possible to
design downward tree automata antichain algorithms. These could be then
combined with downward tree automata simulation. We believe that in practice,
downward algorithms could outperform the upward ones. The upward algorithms
suffer from a need of exploring relatively high nondeterministic choice of an
upward tree automata run. One dimension of this nondeterminism could be
eliminated by a downward algorithm.  Moreover, downward simulation is cheaper
and often richer than upward simulation parametrised by identity, which could
be another advantage of downward algorithms.

Another interesting idea is to try to combine relations in the spirit of our
mediated preorder from Chapter~\ref{chapter:ta_reduction} with the antichain
methods. Mediated preorders are richer than simulations, but imply different
yet still interesting properties of runs of automata.

We would like to perform even more experiments,
including, e.g., experiments where our most recent techniques will be incorporated into the
entire framework of abstract regular (tree) model checking or into some
automata-based decision procedures.
A work on a BDD based tree automata library (in the style of MONA tree automata library \cite{klarlund:MONA}) 
that could make the recent tree automata techniques widely available even for more practical purposes has already started.  
We hope that this will
yield another significant improvement in the tree automata technology allowing
for a new generation of tools using tree automata. 
Finally, we are working on an ARTMC-based tool for verifying pointer manipulating programs that will also use all the
recent tree automata techniques.
We also expect that the tools will generate meaningful experimental data that will be helpful for further research on finite automata.

\chapter{Simulation-based Reduction of Alternating B\"uchi Automata}
\label{chapter:aba_reduction}

\renewcommand{\root}{\mathit{root}}
\newcommand{\leaf}{\mathit{leaf}}
\newcommand{\branch}{\mathit{branches}}
\renewcommand{\succ}[1]{\mathit{succ_{#1}}}
\renewcommand{\L}{\mathcal{L}}
\newcommand{\bs}[2]{\langle #1,#2 \rangle}
\newcommand{\xtrlow}[1]{\mathrel{\raisebox{0pt}[0pt][-2pt] {\ensuremath {\xtr {\raisebox{0pt}[0pt][2pt]{$\scriptstyle{#1}$}}}}} }
\newcommand{\mc}[2]{\multicolumn{#1}{|c|}{#2}}
\newcommand{\mr}[2]{\multirow{#1}{*}{#2}}

\newcommand{\ext}{\mathit{ext}}
\newcommand{\height}{\mathit{height}}
\newcommand{\mcov}{\preceq_{\mathsf{w}\text{-}\ext}}
\newcommand{\sw}{\mathsf{sw}_T}
\newcommand{\morestrong}[2]{\prec^{T}_{#1,#2}}
\newcommand{\moreacc}{\preceq_{\alpha^+\!\Rightarrow\alpha}}
\def\centerframe#1#2#3{%
\framebox{%
\parbox{0pt}{\rule{0pt}{#2}}\parbox{#1}{\makebox[#1]{#3}}}}

\newcommand{\fcov}{\preceq_\ext}
\newcommand{\fcovx}[1]{\preceq_{\ext_{#1}}}

\newcommand{\wma}{\preceq_{\alpha F}}
\newcommand{\moreaccm}{\moreacc^M}
\newcommand{\U}{\mathbb{U}^\infty}
\newcommand{\V}{\mathbb{V}^\infty}

\newcommand{\Bad}{\mathit{Bad}}
\newcommand{\Good}{\mathit{Good}}
\newcommand{\New}{\mathit{New}}

In this chapter, we present the results from our first attempt to adapt our
techniques beyond the scope of finite word/tree automata, which was first
published in \cite{abdulla:mediating}. Namely, we focus on simulation-based
reduction of alternating B\"uchi automata inspired by the technique
described in Chapter~\ref{chapter:ta_reduction}. 

Alternating B\"uchi automata (ABA) are succinct state-machine representations
of $\omega$-regular languages (regular sets of infinite sequences).  They are
widely used in the area of formal specification and verification of
non-terminating systems.  One of the most prominent examples of the use of ABA
is the complementation of nondeterministic B\"uchi
automata~\cite{kupferman:weak}. It is an essential step of the
automata-theoretic approach to model checking when the specification is given
as a positive B\"uchi automaton~\cite{vardi:automata2} and also learning based
model checking for liveness properties~\cite{farzen:extending}.  The other
important usage of ABA is as the intermediate data structure for translating a
linear temporal logic (LTL) specification to an automaton~\cite{gastin:fast}.

However, because of the compactness of ABA\footnote{ABA's are exponentially
more succinct than nondeterministic B\"uchi automata.}, the algorithms that work on
them are usually of high complexity. For example, both the complementation and
the LTL translation algorithms transform an intermediate ABA to an equivalent
NBA. The transformation is exponential in the size of the input ABA.  Hence,
one may prefer to reduce the size of the ABA (with some relatively cheaper
algorithm) before giving it to the exponential procedure.

In the study of Fritz and Wilke, simulation-based minimisation is proven as a
very effective tool for reducing the size of ABA~\cite{fritz:simulation}.
However, they considered only \emph{forward} simulation relations. Inspired by
our work on tree automata reduction methods, we introduce also a notion of
\emph{backward} simulation (parametrised by forward simulation) that can be
used for reducing the size of ABA as well. As will be explained in
Section~\ref{section:simulation}, similarly as for tree automata upward simulation,
quotienting wrt.  \emph{backward} simulation (i.e., simplifying the automaton by
collapsing backward simulation equivalent states) does not preserve the
language, however, backward simulation can be used for quotienting in
combination with forward simulation. In fact, we will arrive to an alternating
automata equivalent of the tree automata notion of \emph{mediated equivalence}
from Chapter~\ref{chapter:ta_reduction}.  

We evaluate the performance of minimising ABA with mediated equivalence is
evaluated on a large set of experiments. In the experiments, we apply different
simulation-based minimisation approaches to improve the complementation
algorithm of nondeterministic B\"uchi automata. The experimental results show
that the minimisation using mediated preorder significantly outperforms the
minimisation using forward simulation. To be more specific, on average,
mediated minimisation results in a 30\% better reduction in the number of
states and 50\% better reduction in the number of transitions than forward
minimisation on the intermediate ABA.  Moreover, in the complemented
nondeterministic B\"uchi automata, mediated minimisation results in a 100\%
better reduction in the number of states and 300\% better reduction in the
number of transitions than forward minimisation.

\section{Basic Definitions}\label{section:definition}

Given a finite set $X$, we use $X^*$ to denote the set of all finite words over
$X$ and $X^\omega$ for the set of all infinite words over $X$. The empty word
is denoted $\epsilon$ and $X^+ = X^*\setminus \{\epsilon\}$.  The concatenation
of a finite word $u\in X^*$ and a finite or infinite word $v\in X^*\cup
X^\omega$ is denoted by $uv$.  For a word $w\in X^*\cup X^\omega$, $|w|$ is
the length of $w$ ($|w|=\infty$ if $w\in X^\omega$), $w_i$ is the $i$th letter
of $w$ and $w^i$ the $i$th prefix of $w$ (the word $u$ with $w = uv$ and
$|u|=i$). $w^0=\epsilon$.  The concatenation of a finite word $u$ and a set
$S\subseteq X^*\cup X^\omega$ is defined as $uS = \{uv\mid v\in S\}$.

An \emph{alternating B\"uchi automaton} is a tuple $\A =
(\Sigma,Q,\iota,\delta,\alpha)$ where $\Sigma$ is a finite alphabet, $Q$ is a
finite set of states, $\iota\in Q$ is an initial state, $\alpha\subseteq Q$ is
a set of accepting states, and $\delta:Q\times \Sigma \rightarrow 2^{2^Q}$ is a
total transition function.  A \emph{transition} of $\A$ is of the form $p
\xtr{a} P$ where $P\in \delta(q,a)$.

A \emph{tree $T$ over $Q$} is a subset of $Q^+$ that contains all nonempty
prefixes of each one of its elements (i.e., $T\cup\{\epsilon\}$ is
prefix-closed).  Furthermore, we require that $T$ contains exactly one $r\in
Q$, the \emph{root of $T$}, denoted $\root(T)$.  We call the elements of $Q^+$
\emph{paths}. For a path $\pi q$, we use $\leaf(\pi q)$ to denote its last
element $q$.  Define the set $\branch(T)\subseteq Q^+\cup Q^\omega$ such that
$\pi\in \branch(T)$ iff $T$ contains all prefixes of $\pi$ and $\pi$ is not a
proper prefix of any path in $T$.  In other words, a \emph{branch of $T$} is
either a maximal path of $T$, or it is a word from $Q^\omega$ such that $T$
contains all its nonempty prefixes.  
We use $\succ T(\pi) = \{r\mid\pi r\in
T\}$ to denote the set of successors of a path $\pi$ in $T$, and $\height(T)$
to denote the length of the longest branch of $T$.
A tree $U$ over $Q$ is a
\emph{prefix of $T$} iff $U\subseteq T$ and for every $\pi\in U$, $\succ U(\pi)
= \succ T(\pi)$ or $\succ U(\pi) = \emptyset$.  The \emph{suffix of $T$}
defined by a path $\pi q$ is the tree $T(\pi q) = \{q \psi\mid \pi q \psi \in
T\}$.

Given a word $w\in\Sigma^\omega$, a tree $T$ over $Q$ is a \emph{run of $\A$ on
$w$}, if for every $\pi \in T$, $\leaf(\pi)\xtrlow{w_{|\pi|}} \succ T(\pi)$ is
a transition of $\A$.  Finite prefixes of $T$ are called \emph{partial runs on
$w$}.  A run $T$ of $\A$ over $w$ is \emph{accepting} iff every infinite branch
of $T$ contains infinitely many accepting states. A word $w$ is \emph{accepted}
by $\A$ from a state $q\in Q$ iff there exists an accepting run $T$ of $\A$
over $w$ with $\root(T) = q$. The \emph{language of a state $q\in Q$ in $\A$},
denoted $\L_\A(q)$, is the set of all words accepted by $\A$ from $q$. Then
$\L(\A) = \L_\A(\iota)$ is the \emph{language of $\A$}.  For simplicity of
presentation, we assume in the rest of the paper that $\delta$ never allows a
transition of the form $p \xtr{a} \emptyset$. This means that no run can
contain a finite branch. Any automaton can be easily transformed into one
without such transitions by adding a new accepting state $q$ with $\delta(q,a)
= \{\{q\}\}$ for every $a\in\Sigma$ and replacing every transition $p \xtr{a}
\emptyset$ by $p \xtr{a} \{q\}$.

We note that for technical reasons, we use a simpler definition of a tree and a
run of an alternating automaton than the usual one (e.g., \cite{kupferman:weak}
or Chapter~\ref{chapter:ta_reduction}). A tree is usually defined as a prefix
closed subset of $\nat^*$ and a run is then a map $r$ that assigns a state to
every element (node) of a tree. This definition allows existence of nodes with
more than one immediate successor labelled by the same state and successors of
a node are ordered. However, order as well as number of occurrences of a state
in the role of a successor of a parent state has no relevance for semantics of
an ABA.  From this point of view, it is more convenient to define runs simply
as unordered trees. 

\section{Simulation Relations}\label{section:simulation}

In this section, we give the definitions of forward and backward simulation over
ABA and discuss some of their properties.  The notion of backward simulation is
inspired by a similar tree automata notion studied 
 in Chapter~\ref{chapter:ta_reduction}---namely, the upward simulation
parametrised by a downward simulation (the connection between tree automata and
ABA follows from the fact that the runs of ABA are in fact trees).

For the rest of the section, we fix an ABA $\A =
(\Sigma,Q,\iota,\delta,\alpha)$. We define relations $\preceq_\alpha$ and
$\preceq_\iota$ on $Q$ s.t. $q\preceq_\alpha r$ iff $q\in \alpha\implies r\in
\alpha$ and $q\preceq_\iota r$ iff $q = \iota \implies r = \iota$. For a binary
relation $\preceq$ on a set $X$, the relation $\preceq^{\forall\exists}$ on
subsets of $X$ is defined as $Y\preceq^{\forall\exists} Z$ iff $\forall z\in
Z.\  \exists y\in Y.\ y\preceq z$, i.e., iff the upward closure of $Z$ wrt.
$\preceq$ is a subset of the upward closure of $Y$ wrt. $\preceq$.

\paragraph{Forward Simulation.}
A \emph{forward simulation} on $\A$ is a relation $ {\preceq_F} \subseteq Q
\times Q$ such that $p \preceq_F r$ implies that (i) $p\preceq_\alpha r$ and
(ii)~for all $p \xrightarrow{a} P$, there exists a $r
\xrightarrow{a} R$ such that $P\preceq_F^{\forall\exists} R$.

For the basic properties of forward simulation, we rely on the work
\cite{gurumurthy:complementing} by Gurumurthy et al. In particular,
(i) there exists a unique maximal forward simulation $\preceq_F$ 
on $\A$ called \emph{forward simulation preorder} which is reflexive and transitive, (ii) for any $q,r\in Q$ such that
$q\preceq_F r$, it holds that $\L_\A(q)\subseteq \L_\A(r)$, and (iii) quotienting
wrt.  $\preceq_F\cap\preceq_F^{-1}$ preserves the language of $\A$.

\paragraph{Backward Simulation. }
Let ${\preceq_F}$ be a forward simulation on $\A$.  A \emph{backward
simulation} on~$\A$ parametrised by $\preceq_F$ is a relation $ {\preceq_B}
\subseteq Q \times Q$ such that $p \preceq_B r$ implies that (i)~$p
\preceq_\iota r $, (ii)~$p \preceq_\alpha r $, and (iii) for all $q
\xtr{a} P\cup \{p\},p\not\in P$, there exists a $s \xtr{a} R\cup
\{r\},r\not\in R$ such that $q \preceq_B s$ and $P\preceq_F^{\forall\exists} R$.
The lemma below describes basic properties of backward simulation.
\begin{lemma}\label{lemma:backward-basic}
For any reflexive and transitive forward simulation $\preceq_F$ on $\A$, there
exists a~unique maximal backward simulation $\preceq_B$ on $\A$ parametrised by
$\preceq_F$ that is reflexive and transitive.
\end{lemma}
\begin{proof} 
The proof is an analogy of the proof of Lemma~\ref{lemma:upsim-basic}.

\emph{Union:} Given two backward
simulations $ \preceq_B^1$ and $ \preceq_B^2$ induced by $ \preceq_F$, we want to prove that ${\preceq_B} =
 {\preceq_B^1}\cup{\preceq_B^2}$ is also a backward simulation induced by $ \preceq_F$.  Let $p \preceq_B r$
for some $p,r\in Q$, then either $p\preceq_B^1r$ or $p\preceq_B^2r$.  Assume without loss of
generality that $p\preceq_B^1r$. Then, from the definition of backward simulation, whenever
$p' \ltr a P\cup\{p\},p\not\in P$, then there is a rule $r'\ltr a R\cup\{r\}, r\not\in R$, $p'\preceq_B^1 r'$, and $P\preceq_F^{\forall\exists} R$.  As ${\preceq_B^1}\subseteq {\preceq_B}$ gives $p'\preceq_B r'$, $\preceq_B$ fulfils the definition of
backward simulation induced by $\preceq_F$.

\emph{Reflexive closure:} It can be seen from the definition that the identity is
a backward simulation induced by $\preceq_F$ for any forward simulation $\preceq_F$. 
Therefore, from the closure under union, the union of the identity and any backward
simulation induced by $\preceq_F$ is a backward simulation induced by $\preceq_F$. 

\emph{Transitive closure:} Let $\preceq_B$ be a backward simulation induced by
$ \preceq_F$ and let $\preceq_B^T$ be its transitive closure.  Let
$p^1\preceq_B^Tp^m$ and ${r^1} \ltr a P^1\cup\{p^1\},p^1\not\in P^1$. Apparently, $p^1\preceq_\alpha p^m$ since $\preceq_\alpha$ is a transitive subset of $\preceq_B$. From
$p^1\preceq_B^Tp^m$, we have that there are states $p^1,\ldots,p^m$ such that
$p^1\preceq_Bp^2\preceq_B\cdots \preceq_Bp^m$.  Therefore, there are also rules
${r^2} \ltr a P^2\cup\{p^2\},\ldots,{r^m} \ltr a P^m\cup\{p^m\}$ with
$p^2\not\in P^2,\ldots,p^m\not\in P^m$, $r^1\preceq_B\cdots \preceq_Br^m$,  
and $P^1 \preceq_F^{\forall\exists} P^2 \preceq_F^{\forall\exists} \cdots \preceq_F^{\forall\exists} P^m$.
Thus, by definition of $\preceq_B^T$, we have $r^1\preceq_B^Tr^m$, and by 
transitivity of $\preceq_F^{\forall\exists}$, $P^1 \preceq_F^{\forall\exists} P^m$. Therefore, $\preceq_B^T$ fulfils the
definition of a backward simulation induced by $ \preceq_F$.  \end{proof}

By Lemma~\ref{lemma:backward-basic}, for a reflexive and transitive forward
simulation $\preceq_F$, there is a unique maximal upward simulation
parametrised by $\preceq_F$ and it is a preorder. We call it the \emph{backward
simulation preorder} on $\A$ parametrised by $\preceq_F$. Our backward
simulation is a close analogy of tree automata upward simulation.  Similarly as
upward simulation, backward simulation cannot be directly used for quotienting
(below we give an example of an automaton where quotienting using backward
simulation does not preserve language).  However, in Section
\ref{section:mediated}, we show that backward simulation can be combined with
forward simulation into a mediated equivalence (in the same way as tree
automata upward simulation can be combined with downward simulation) that can
be used for quotienting.

\begin{example}[backward simulation cannot be used for quotienting]
Consider the ABA $\A = (\{a,b\},\{s_0, s_1, s_2, s_3, s_4, s_5, s_6\},s_0,\delta,\{s_0, s_1, s_2, s_3, s_4, s_5, s_6\})$ where 
$$\begin{array}{llll}
s_0 \xrightarrow{a}\{s_4\},& s_1 \xrightarrow{b}\{s_2,s_5\},& s_2 \xrightarrow{b}\{s_2,s_3\}, & s_5 \xrightarrow{b}\{s_0\},\\
s_0 \xrightarrow{a}\{s_1\},& s_1 \xrightarrow{b}\{s_1,s_3\},& s_3 \xrightarrow{a}\{s_0\},     & s_6 \xrightarrow{a}\{s_0\}\\
s_0 \xrightarrow{b}\{s_0\},&                                & s_4 \xrightarrow{b}\{s_4, s_6\}, &
\end{array}
$$
are transitions of $\A$.
The maximal forward simulation relation $\preceq_F$ in $\A$ is
$$
\begin{array}{l}
\{
(s_0, s_0), (s_1, s_0), (s_1, s_1), (s_1, s_5), (s_2, s_0), (s_2, s_1), (s_2, s_2), (s_2, s_4),\\ 
(s_2, s_5), (s_3, s_3), (s_3, s_6), (s_4, s_0), (s_4, s_1), (s_4, s_2), (s_4, s_4), (s_4, s_5), \\
(s_5, s_0), (s_5, s_5), (s_6, s_3), (s_6, s_6)\}.
\end{array}
$$
The maximal backward simulation relation $\preceq_B$ parametrised with $\preceq_F$ is
$$
\begin{array}{l}
\{
(s_0, s_0), (s_1, s_1), (s_1, s_4), (s_2, s_2), (s_3, s_3), (s_4, s_1), (s_4, s_4), (s_5, s_2),\\ 
(s_5, s_3), (s_5, s_5), (s_5, s_6), (s_6, s_2), (s_6, s_3), (s_6, s_5), (s_6, s_6)\}.
\end{array}
$$
If we collapse states of $\A$ wrt. $\preceq_M$ (i.e., the two sets of states $\{s_1 ,s_4\}, \{s_5,s_6\}$ are collapsed), we will get the following alternating B\"uchi automaton $\A' = (\{a,b\}, \{s_0,s_1,s_2,s_3,s_4\}, s_0, \delta, \{s_0,s_1,s_2,s_3,s_4\})$ where
$$\begin{array}{llll}                                                 
                                                                 
s_0 \xrightarrow{a}\{s_1\},    & s_1 \xrightarrow{b}\{s_2,s_4\}, &  s_2 \xrightarrow{b}\{s_2,s_3\}, & s_4 \xrightarrow{a}\{s_0\},\\
s_0 \xrightarrow{b}\{s_0\},    & s_1 \xrightarrow{b}\{s_1,s_4\}, &  s_3 \xrightarrow{a}\{s_0\},     & s_4 \xrightarrow{b}\{s_0\}\\ 
                               & s_1 \xrightarrow{b}\{s_1, s_3\},&                                  & \\
            	               &                                 &                                  &
\end{array}$$
are transitions of $\A'$.
Note that $\A'$ accepts the word $ab^\omega$, but $\A$ does not.
\qed
\end{example}

\subsection{Runs and Simulations}
We now formulate connections between
simulations and runs of ABA that are fundamental for our further reasoning.  Let
$\preceq_F$ and $\preceq_B$ be forward and backward simulations on $\A$, which
are both reflexive and transitive.  For every $x\in\{B,F,\alpha\}$, we extend
the relation $\preceq_x$ to $Q^+\times Q^+$ such that for $\pi,\psi\in Q^+$,
$\pi\preceq_x\psi$ iff $|\pi| = |\psi|$ and for all $1\leq i\leq |\pi|$,
$\pi_i\preceq_x \psi_i$.  We say that $\psi$ forward simulates $\pi$, $\psi$
backward simulates $\pi$, or $\psi$ is more accepting than $\pi$ when
$\pi\preceq_F\psi$, $\pi\preceq_B\psi$, or $\pi\preceq_\alpha\psi$,
respectively.
This notation is further extended to trees.
For trees $T,U$ over $Q$ and for $x\in\{\alpha,F\}$, we write,
$T\preceq_x U$ if $\branch(T) \preceq_x^{\forall\exists} \branch(U)$. Similarly, we say that $U$
forward simulates $T$,
or $U$ is more accepting than $T$ when $T\preceq_F U$,
or $T\preceq_\alpha U$, respectively.
Note that $\preceq_x$ is reflexive and transitive
for all the variants of $x\in\{F,B,\alpha\}$ defined over states, paths, or
trees (this follows from the assumption that the original relations $\preceq_F$ and $\preceq_B$ on states are reflexive and transitive).
Moreover,
${\preceq_B}\subseteq{\preceq_\alpha}$,
${\preceq_B}\subseteq{\preceq_\iota}$, and
${\preceq_F}\subseteq{\preceq_\alpha}$.

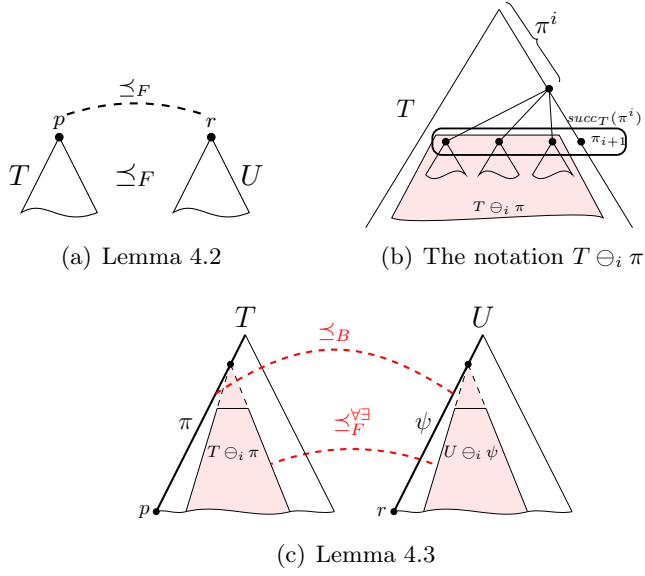
\begin{figure}[t]
\centering
  \subfigure[Lemma~\ref{forward:lemma}]
  {
    \begin{tikzpicture}
\draw[dashed,thick] (0.1,0.3) .. controls (0.7,0.5) and (1.3,0.5)  ..
node[midway,above=-1mm,yshift=0.5mm] {\scriptsize $\preceq_F$} (1.9,0.3);

\draw (0,0) -- node [left] {$T$} (-0.5,-1)
.. controls (-0.17,-0.8) and (0,-1.2) ..
 (0.5,-1)-- (0,0);
\draw[fill=black] (0,0) circle (0.5mm);
\node[yshift=-1mm] at (0,0.3) {\scriptsize $p$};

\node at (1,-0.5) {$\preceq_F$};

\draw[xshift=2cm]  (0,0) --  (-0.5,-1)
.. controls (-0.17,-0.8) and (0,-1.2) ..
 (0.5,-1)-- node [right] {$U$}(0,0);
\draw[xshift=2cm,fill=black] (0,0) circle (0.5mm);
\node[yshift=-1mm,xshift=2cm] at (0,0.3) {\scriptsize $r$};

\end{tikzpicture}
    \label{definition:forward}
  }
  \subfigure[The notation $T\ominus_i \pi$]
  {
    \resizebox{!}{3cm}{\makebox[6cm]{\begin{tikzpicture}

\draw (0,0) -- (-2,-3.3);
\draw (0,0) -- (2,-3.3);

\draw[fill=red!10] (-0.95,-1.9) -- (0.9,-1.9) --
(1.56,-3.15)
.. controls (1.26,-3.0) and (-1.31,-3.3) ..
(-1.61,-3.15) --  (-0.95,-1.9) ;

\draw[decorate,decoration=brace] (0.1,0.1) --  node[xshift=1mm,right,near start] {$\pi^i$} (0.9,-1.1);
\draw[xshift=-0.8cm,fill=red!10] (0,-2.0) -- (-0.3,-2.5) .. controls (-0.1,-2.7) and (0.1,-2.3)  .. (0.3,-2.5) -- (0,-2.0);
\draw[xshift=-0.8cm,fill=black] (0,-2.0) circle (0.5mm);

\draw[fill=red!10] (0,-2.0) -- (-0.3,-2.5) .. controls (-0.1,-2.3) and (0.1,-2.7)  .. (0.3,-2.5) -- (0,-2.0);
\draw[fill=black] (0,-2.0) circle (0.5mm);

\draw[xshift=0.8cm,fill=red!10] (0,-2.0) -- (-0.3,-2.5) .. controls (-0.1,-2.7) and (0.1,-2.3)  .. (0.3,-2.5) -- (0,-2.0);
\draw[xshift=0.8cm,fill=black] (0,-2.0) circle (0.5mm);

\draw[fill=black] (0.75,-1.2) circle (0.5mm);

\draw[fill=black] (1.23,-2.0) circle (0.5mm);

\draw[thick,rounded corners](-1.0 cm,-1.8) rectangle (1.9,-2.2);

\draw (0.75,-1.2) -- (0.8,-2.0);
\draw (0.75,-1.2) -- (0,-2.0);
\draw (0.75,-1.2) -- (-0.8,-2.0);

\node at (1.65,-2.0) {\tiny$\pi_{i+1}$};
\node[yshift=0.2mm,xshift=0.8mm] at (1.5,-1.65) {\tiny${\it succ}_T(\pi^i)$};

\node at (0,-3) {\tiny$T\ominus_i\pi$};

\node at (-1.4,-1.5) {$T$};
\end{tikzpicture}}}
    \label{definition:notation}
  }
  \subfigure[Lemma~\ref{backward:lemma}]
  {
    \resizebox{!}{3cm}{\makebox[6cm]{\begin{tikzpicture}

\draw[dashed,red,line width=1pt] (-0.2,-2.5) .. controls (1.05,-1.75) and (2.55,-1.75) ..
node[above=1mm]{$\preceq_F^{\forall\exists}$} (3.8,-2.5);

\draw[white,fill=red!10]  (-0.25,-0.5) -- (-0.475,-1.25) --
 (-0.475,-1.25) -- (-1,-3) --
(-1,-3) .. controls (-0.4,-2.85) and (0.2,-3.15) ..  (0.75,-3) --
(0.75,-3) --  (0.05,-1.25)  --
(0.05,-1.25) --  (-0.25,-0.5);

\node at (0,0.3) {\Large $T$};
\draw[line width=1pt] (0,0) -- (-1.5,-3);
\draw[fill=black] (-1.5,-3) circle (0.5mm);
\node[font=\small] at (-1.7,-3) {$p$};

\draw[fill=black] (-0.25,-0.5) circle (0.5mm);
\node at (-1,-1.5) {\large $\pi$};

\draw[dashed] (-0.25,-0.5) -- (-0.475,-1.25);
\draw (-0.475,-1.25) -- (-1,-3);

\draw[dashed] (-0.25,-0.5) -- (0.05,-1.25);
\draw (0.05,-1.25) -- (0.75,-3);

\draw[] (-0.475,-1.25) --  (0.05,-1.25);

\draw (0,0) -- (1.5,-3);

\draw (-1.5,-3) .. controls (-1.35,-2.95) and (-1.2,-3.05) ..  (-1,-3);
\draw (-1,-3) .. controls (-0.4,-2.85) and (0.2,-3.15) ..  (0.75,-3);
\draw (0.75,-3) .. controls (1.0,-2.95) and (1.25,-3.05) ..  (1.5,-3);

\node[font=\scriptsize] at (-0.2,-2) {$T\ominus_i\pi$};

\draw[xshift=4cm,white,fill=red!10]  (-0.25,-0.5) -- (-0.475,-1.25) --
 (-0.475,-1.25) -- (-1,-3) --
(-1,-3) .. controls (-0.4,-2.85) and (0.2,-3.15) ..  (0.75,-3) --
(0.75,-3) --  (0.05,-1.25)  --
(0.05,-1.25) --  (-0.25,-0.5);

\node[xshift=4cm] at (0,0.3) {\Large $U$};
\draw[xshift=4cm,line width=1pt] (0,0) -- (-1.5,-3);
\draw[xshift=4cm,fill=black] (-1.5,-3) circle (0.5mm);
\node[xshift=4cm,font=\small] at (-1.7,-3) {$r$};

\draw[xshift=4cm,fill=black] (-0.25,-0.5) circle (0.5mm);
\node[xshift=4cm] at (-1,-1.5) {\large $\psi$};

\draw[xshift=4cm,dashed] (-0.25,-0.5) -- (-0.475,-1.25);
\draw[xshift=4cm] (-0.475,-1.25) -- (-1,-3);

\draw[xshift=4cm,dashed] (-0.25,-0.5) -- (0.05,-1.25);
\draw[xshift=4cm] (0.05,-1.25) -- (0.75,-3);

\draw[xshift=4cm] (-0.475,-1.25) --  (0.05,-1.25);

\draw[xshift=4cm] (0,0) -- (1.5,-3);

\draw[xshift=4cm] (-1.5,-3) .. controls (-1.35,-2.95) and (-1.2,-3.05) ..  (-1,-3);
\draw[xshift=4cm] (-1,-3) .. controls (-0.4,-2.85) and (0.2,-3.15) ..  (0.75,-3);
\draw[xshift=4cm] (0.75,-3) .. controls (1.0,-2.95) and (1.25,-3.05) ..  (1.5,-3);

\node[xshift=4cm,font=\scriptsize] at (-0.2,-2) {$U\ominus_i\psi$};

\draw[dashed,red,line width=1pt] (-0.5,-1.0) .. controls (1,0) and (2.0,0) ..
node[above=0mm]{$\preceq_B$}
(3.5,-1.0);

\end{tikzpicture}}}
    \label{definition:backward}
  }
  \caption{Illustration of the lemmas}
\end{figure}

\begin{lemma}\label{forward:lemma}
For any $p,r\in Q$ with $p\preceq_F r$ and a
partial run $T$ of $\A$ on $w\in\Sigma^\omega$ with the root $p$, there is a
partial run $U$ of $\A$ on $w$ with the root $r$ such that $T\preceq_F U$.
\end{lemma}
\begin{proof} 
We
prove the lemma by induction on $\height(T)$. In the base case when $T =
\{p\}$, it is sufficient to take $U = \{r\}$.  Suppose now that the lemma holds
for every word $u$ and for every partial run $V$ of $\A$ on $u$ such that
$\height(V) < \height(T)$. From $p\preceq_F r$, there is a
transition $r\xtr {w_1} R$ of $\A$ where $ \succ T(p)
\preceq_F^{\forall\exists} R$.
Observe that $T = \{p\}\cup\bigcup_{p' \in \succ T(p)} pT(p')$ where for each
$p'\in \succ T(p)$, $T(p')$ is a partial run of $\A$ with the root $p'$ on the word
$v$ such that $w = w_1v$. Notice that $\height(T(p'))<\height(T)$.
The induction hypothesis now can be applied to every triple $p'\in
\succ T(p),r'\in R, T(p')$ with $p'\preceq_F r'$. It gives us a partial run $U_{r'}$ of $\A$ on $v$ with $\root(U_{r'}) = r'$ such that $T(p')\preceq_F U_{r'}$.  The run $U$ with the required properties is then
constructed by plugging the runs $U_{r'},r'\in R$, to $r$, i.e., $U =
\{r\}\cup\bigcup_{r'\in R} r U_{r'}$.   \end{proof}

We will need to inspect the connection between runs and backward simulation in a
relatively detailed way. For this, we introduce to following notation. Given a
tree $T$ over $Q$, $\pi\in T$, and $1\leq i\leq |\pi|$, the set $T\ominus_i
\pi$ is the union of branches of suffix trees $T(\pi^i q),q\in \succ T(\pi^i)$,
with the branches of the suffix tree $T(\pi^{i+1})$ excluded.  Formally, let
$Q^i = \succ T(\pi^i)\setminus\{\pi_{i+1}\}$ be the set of all successors of
$\pi^i$ in $T$ without the successor continuing in $\pi$.  Then $T\ominus_i \pi
= \bigcup_{q\in Q^i}\branch(T(\pi^i q))$ (notice that if $i = 0$, then
$T\ominus_i\pi = \emptyset$).

\begin{lemma}\label{backward:lemma}
For any $p,r\in Q$ with $p\preceq_B r$, a partial run $T$ of $\A$ on $w\in\Sigma^\omega$ and $\pi \in
\branch(T)$ with $\leaf(\pi) = p$, there is a partial run $U$ of $\A$ on $w$ and $\psi
\in \branch(U)$ with $\leaf(\psi) = r$ such that
$\pi\preceq_B\psi$, and for all $1\leq i \leq |\pi|$,
${T\ominus_i \pi} \preceq_F^{\forall\exists} {U\ominus_i \psi}$.
\end{lemma}

\begin{proof}
By induction on the length of $\pi$. In the base case, when $\pi = p$
and $T = \{p\}$, it is sufficient to take $U = \{r\}$ and $\psi = r$.
Suppose now that $\pi \neq p$ and that the lemma holds for every partial run
$T'$ of $\A$ on $w$, states $p',r'\in Q$ such that $p'\preceq_B r'$, and every
$\pi'\in \branch(T')$ with $\leaf(\pi') = p'$ and $|\pi'|<|\pi|$.

For the induction step, let $\pi = \pi' p$ and let $\succ T(\pi') = P\cup \{p\},
p\not\in P$. By the definition of $\preceq_B$, there is a transition
$s\xtrlow{w_{|\pi|}} R \cup\{r\},r\not\in R$ of $\A$ such that $\leaf(\pi')\preceq_B s$
and $P\preceq_F^{\forall\exists}R$.
Let $T' = T \setminus \{\pi\} \setminus  \bigcup_{p'\in P}\pi' T(\pi' p')$. 
Then $T'$ is a partial run of $\A$ on $w$ and $\pi' \in
\branch(T')$, $|\pi'|<|\pi|$, and therefore we can apply induction hypothesis
to $T'$, $\leaf(\pi')$, $s$, and $\pi'$.  This gives us a partial run $U'$ of $\A$ on $w$
with $\psi' \in \branch(U')$ such that $\leaf(\psi') = s$, $\pi'\preceq_B\psi'$ and
for each $1\leq j \leq |\pi'|$,
${T'\ominus_j \pi'} \preceq_F^{\forall\exists} {U'\ominus_j \psi'}$.
For every $p'\in\succ T(\pi')$, $T(\pi' p')$ is a
partial run of $\A$ with the root $p'$ on the suffix $v$ of $w$ such that $w =
uv,|u|= |\pi|-1$.  We can apply Lemma~\ref{forward:lemma} to the triples $r'\in
R, p'\in P, T(\pi' p')$ with $p'\preceq_F r'$. This gives us for each $r'\in R$ a
run $U_{r'}$ of $\A$ on $v$ with $\root(U_{r'})=r'$ such that there is some $p'\in P$ with $T(\pi' p')\preceq_F U_{r'}$.
Now we construct a run $U$ and a path $\psi$ with the required properties
by plugging $r$ and runs $U_{r'},r'\in R$ to the path $\psi'$ in $U'$,
i.e., $\psi = \psi' r$ and  $U = U'\cup\{\psi\}\cup\bigcup_{r'\in R} \psi'U_{r'}$.
(To see that $U$ really satisfies the required properties, observe the following: (i) As
$U\ominus_{|\pi'|} \psi = \bigcup_{r'\in R}\branch(U_{r'})$ and $T\ominus_{|\pi'|} \pi = \bigcup_{p'\in P}\branch(T(\pi'p'))$, and because for each $r'\in R$, there is $p'\in P$ with $T(\pi' p')\preceq_F U_{r'}$, we have that
$T\ominus_{|\pi'|} \pi\preceq^{\forall\exists}_F U\ominus_{|\pi'|} \psi$. (ii) For all $1\leq j < |\pi'|$,
${T\ominus_j \pi =  T' \ominus_j \pi'}\preceq_{F}^{\forall\exists}{{U'\ominus_j \psi'} = {U\ominus_j \psi}}$.).
\end{proof}

\section{Mediated Equivalence and Quotienting}

Here, we discuss the possibility of an indirect use of backward simulation for
simplifying ABA via quotienting. We will introduce an alternating
B\"uchi automata variant of the mediated preorder from
Chapter~\ref{chapter:ta_reduction} that is a combination of forward and
backward simulation suitable for quotienting.
\subsection{The Notion and Intuition of Mediated Equivalence}

We again use the concept of ``jumping runs'' based on the observation that
quotienting an automaton wrt. some equivalence allows a run that arrives to
some state to \emph{jump} to equivalent state and continue from there.
Alternatively, this can be viewed as \emph{extending} the source state of the
jump by the outgoing transitions of the target state\footnote{The first view is
better when explaining the intuition whereas the other is easier to be used in
proofs.}.  The equivalence must have the property that the language is not
increased even when the jumps (or, alternatively, transition extensions) are
allowed. 
It turns out that forward and backward simulation can be combined into a
suitable relation in the same way as downward and upward simulation in
Chapter~\ref{chapter:ta_reduction}. This is, we will define the mediated
preorder $\preceq_M$ as a suitable transitive fragment of $\preceq_F \circ
\preceq_B^{-1}$ and show that allowing jumping to mediated smaller states does
not affect the language. It follows that quotienting wrt. mediated equivalence
(the largest symmetric fragment of $\preceq_M$) preserves language too.

The intuition behind allowing a run to jump from a state $r$ to a state $q$
that are related by a mediated preorder is very similar to the one given in
Chapter~\ref{chapter:ta_reduction}. The relation $q \preceq_F \circ
\preceq_B^{-1} r$ guarantees the existence of the so called \emph{mediator},
which is a state $s$ such that $q \preceq_F s \preceq_B^{-1} r$
(see~Figure~\ref{mediated:basic:aba}). The state $s$ can be reached in the same way
and in the same context\footnote{If a state $s$ is a leaf of a partial run,
then by a \emph{context} of $s$ we mean all the other leaves of the partial
run.} as $r$, and, at the same time, the automaton can continue from $s$ in the
same way as from $q$. Hence, intuitively, the newly allowed run based on the
jump from $r$ to $q$ does not add anything to the language because it can
anyway be realised through $s$ without jumps.

Similarly as in the case of tree automata, jumping cannot be allowed between
all pairs of states from $\preceq_F \circ \preceq_B^{-1}$. We will have to
restrict ourselves only to its fragments $\preceq_M$ that are \emph{preorders}
and are also \emph{forward extensible}, which means that if $q_1 \preceq_M q_2
\preceq_F q_3$, then $q_1 \preceq_M q_3$.  

The reason for this is that we were so far taking into account only one
isolated jump, however, nothing prevents another jumps from occurring in the
context or below the marked occurrence of $r$. This is problematic since the
relations $q \preceq_F s \preceq_B^{-1} r$ are guaranteed only when no further
jumps are allowed. The forward extensibility is required to ensure the
mechanism to work with arbitrary many jumps.  We describe the potential
problems when $\preceq_M$ is not forward extensible (see
Figure~\ref{mediated:problems_appendix} for the illustration).

\label{section:mediated}
\begin{figure}[t]
  \centering
  \subfigure[The Mediator]
  {
    {\begin{tikzpicture}[scale=1.3]

\draw(0,0) -- (-0.5,-1) node [left,midway] {$V$}   -- (0.5,-1) -- (0,0);
\draw[yshift=-1cm] (0,0) -- (-0.5,-1)  -- (0.5,-1) -- (0,0);
\draw[thick,dotted,xshift=5mm,yshift=-1cm] (0,0) -- (0.5,-1);
\draw[thick,dotted,xshift=2mm,yshift=-1cm] (0,0) -- (0.5,-1);

\draw[thick,dotted,xshift=-5mm,yshift=-1cm](0,0) -- (-0.5,-1);
\draw[thick,dotted,xshift=-2mm,yshift=-1cm](0,0) -- (-0.5,-1);


\draw[fill=white,color=white] (0,-1) circle (1mm);
\node[] at (0,-1) { $s$};


\node[yshift=-1.3cm] at (0,-1.3) { $W$};

\draw[fill=white,color=white] (0,0) circle (1mm);
\node[] at (0,0) { $\iota$};

\draw[xshift=-2.0cm,yshift=-1cm](0,0) -- (-0.5,-1) node [left,midway] {$U$}   -- (0.5,-1) -- (0,0);
\draw[xshift=-2.0cm,yshift=-1cm,fill=white,color=white] (0,0) circle (1.1mm);
\node[xshift=-2.6cm,yshift=-1.3cm] at (0,0) { $q$};



\draw[xshift=-4.0cm](0,0) -- (-0.5,-1) node [left,midway] {$T$}   -- (0.5,-1) -- (0,0);


\draw[xshift=-4.0cm,fill=white,color=white] (0,-1) circle (1mm);
\node[xshift=-5.2cm] at (0,-1) { $r$};
\draw[xshift=-4.0cm,fill=white,color=white] (0,0) circle (1mm);
\node[xshift=-5.2cm] at (0,0) { $\iota$};

\draw[xshift=-5cm] (0,0) -- node [midway,left=-0.7mm] {$u$}(0,-1);
\draw[xshift=-5cm] (0,-1) -- node [midway,left=-0.7mm] {$v$}(0,-2);
\draw[xshift=-5cm] (0,-2) -- (0,-2.1);
\draw[xshift=-5cm,dashed] (0,-2.1) -- node [midway,left=-0.7mm] {$w$}(0,-2.7);

\draw[xshift=-5cm]  (-0.1,0) -- (0.1,0);
\draw[xshift=-5cm]  (-0.1,-1) -- (0.1,-1);
\draw[xshift=-5cm]  (-0.1,-2) -- (0.1,-2);

\end{tikzpicture}}
    \label{mediated:basic:aba}
  }
  \hspace{.2in}
  \subfigure[Potential Problems]
  {
     {
\begin{tikzpicture}[scale=1.3]

\draw(0,0) -- (-0.5,-1.0) node [left,midway] {$V$}   -- (0.5,-1.0) -- (0,0);
\draw[yshift=-1.0cm] (0,0) -- (-0.5,-1.0)  -- (0.5,-1.0) -- (0,0);
\draw[thick,dotted,xshift=5mm,yshift=-1.0cm] (0,0) -- (0.5,-1.0);
\draw[thick,dotted,xshift=2mm,yshift=-1.0cm] (0,0) -- (0.5,-1.0);

\draw[thick,dotted,xshift=-5mm,yshift=-1.0cm](0,0) -- (-0.5,-1.0);
\draw[thick,dotted,xshift=-2mm,yshift=-1.0cm](0,0) -- (-0.5,-1.0);


\draw (0,0) -- (0,-1.0); 
\draw[color=white,fill=white] (0,-0.6) circle (1.3mm);
\node at (0,-0.6) {$\psi$};

\draw[yshift=-1.0cm] (0,0) -- (0,-1.0); 
\draw[yshift=-1.0cm,color=white,fill=white] (0,-0.6) circle (1.3mm);
\node[yshift=-1.3cm] at (0,-0.6) {$\rho$};

\draw[fill=white,color=white] (0,-1.0) circle (1mm);
\node[] at (0,-1.0) { $s$};
\draw[fill=white,color=white] (0.35,-1.0) circle (1mm);
\node[] at (0.35,-1.0) { $x$};

\draw[yshift=-1.0cm,fill=white,color=white] (0,-1.0) circle (1mm);
\node[yshift=-1.3cm] at (0,-1.0) { $y$};

\draw[fill=white,color=white] (0,0) circle (1mm);
\node[] at (0,0) { $\iota$};
\node[yshift=-1.0cm] at (0,-1.6) { $W$};

\draw[xshift=-2.0cm,yshift=-1.0cm] (0,0) -- (0,-1.0); 
\draw[xshift=-2.0cm,yshift=-1.0cm,color=white,fill=white] (0,-0.6) circle (1.3mm);
\node[xshift=-2.6cm,yshift=-1.3cm] at (0,-0.6) {$\phi$};
\draw[xshift=-2.0cm,yshift=-1.0cm](0,0) -- (-0.5,-1.0) node [left,midway] {$U$}   -- (0.5,-1.0) -- (0,0);
\draw[xshift=-2.0cm,yshift=-1.0cm,fill=white,color=white] (0,0) circle (1.1mm);
\node[xshift=-2.6cm,yshift=-1.3cm] at (0,0) { $q$};

\draw[xshift=-2.0cm,yshift=-1.0cm,fill=white,color=white] (0,-1.0) circle (1mm);
\node[xshift=-2.6cm,yshift=-1.3cm] at (0,-1.0) { $r$};

\draw[xshift=-4.0cm](0,0) -- (-0.5,-1.0) node [left,midway] {$T$}   -- (0.5,-1.0) -- (0,0);


\draw[xshift=-4.0cm] (0,0) -- (-0.17,-1.0); 
\draw[xshift=-4.0cm,color=white,fill=white] (-0.1,-0.5) circle (1.2mm);
\node[xshift=-5.2cm] at (-0.1,-0.5) {$\pi$};

\draw[xshift=-4.0cm] (0,0) -- (0.16,-1.0); 
\draw[xshift=-4.0cm,color=white,fill=white] (0.15,-0.7) circle (1.2mm);
\node[xshift=-5.2cm] at (0.15,-0.65) {$\pi'$};

\draw[xshift=-4.0cm,fill=white,color=white] (-0.17,-1.0) circle (1mm);
\node[xshift=-5.2cm] at (-0.17,-1.0) { $r$};

\draw[xshift=-4.0cm,fill=white,color=white] (0,0) circle (1mm);
\node[xshift=-5.2cm] at (0,0) { $\iota$};

\draw[xshift=-4.0cm,fill=white,color=white] (0.16,-1.0) circle (1mm);
\node[xshift=-5.2cm] at (0.16,-1.0) { $r$};

\draw[xshift=-5cm] (0,0) -- node [midway,left=-0.7mm] {$u$}(0,-1);
\draw[xshift=-5cm] (0,-1) -- node [midway,left=-0.7mm] {$v$}(0,-2);
\draw[xshift=-5cm] (0,-2) -- (0,-2.1);
\draw[xshift=-5cm,dashed] (0,-2.1) -- node [midway,left=-0.7mm] {$w$}(0,-2.7);

\draw[xshift=-5cm]  (-0.1,0) -- (0.1,0);
\draw[xshift=-5cm]  (-0.1,-1.0) -- (0.1,-1.0);
\draw[xshift=-5cm]  (-0.1,-2) -- (0.1,-2);

\end{tikzpicture}}
     \label{mediated:problems_appendix}
  }
  \caption{Basic Intuition Behind Mediated Equivalence}
\end{figure}
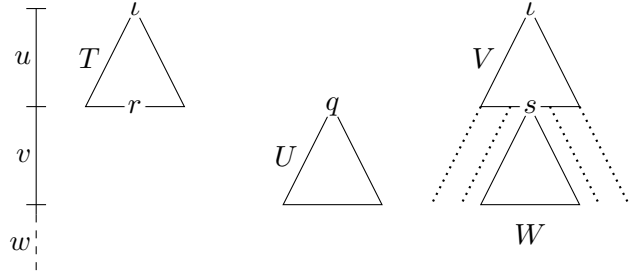
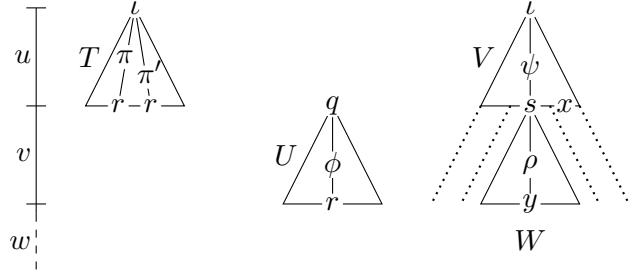

\emph{Problem (i):} The first problem will arise if there is a branch $\phi$ of
$U$ with $\leaf(\phi) = r$.  Here, apart from interconnecting $T$ and $U$, $r$
can use its new transitions also at the end of $\pi\phi$ and connect another
copy of $U$ to the end of $\pi\phi$. Suppose that all leaves of $T$ except $r$
accept $vvw$ and that all leaves of $U$ except $r$ accept $vw$. Then this
enables a new accepting run on the word $uvvw$. In this case, the existence of
the mediator $s$ is not a guarantee that some accepting run on $uvvw$ was
possible before adding transitions to $r$.

\emph{Problem (ii):} Another problem may arise if there are two (or more)
branches in $T$ ending by $r$. Here we use the two branches $\pi$ and $\pi'$ in
Figure~\ref{mediated:problems_appendix} as an example. To construct an
accepting run on $uvw$ from $T$, $r$ has to use the transitions of $q$ at the
end of $\pi$ as well as at the end of $\pi'$ to connect $U$ to $T$ in the both
places.  But partial run $V$ ``covers'' only one of the two occurrences of $r$.
There may be a leaf $x$ of $V$ different from $s$ for which $r$ is the only
leaf in $T$ with $r\preceq_F x$.  Therefore, $x$ needs not accept $vw$ as
there is no guaranteed relation between $q$ and $x$. In this case $V$ is not a
prefix of an accepting run on $uvw$ and $uvw$ need not be in $\L(\A)$. 

We will show how the two problems can be solved by requiring $\preceq_M$ to be a forward extensible preorder. 

In the case of Problem (i), if $y$ uses transitions of $q$ to accept $vw$, then
$W$ becomes a prefix of an accepting run on $vvw$ and thus $V$ becomes a prefix
of a new accepting run on $uvvw$.  We know that $r\preceq_F y$. Thus, by forwards extensibility, $q\preceq r \preceq_F y$ gives $q\preceq y$,
which implies that there is a mediator for $q$ and $y$.  Observe that $y$ used
transitions of $q$ just once.  Therefore, by an analogical argument by which we
derived that $\A$ accepts $uvw$ in the first case when $r$ used the new
transitions only once, we can here derive that there is an accepting run of
$\A$ on $uvvw$ which does not involve new transitions.

In the case of Problem (ii),  if $x$ uses the transitions of $q$ to accept
$vw$, $V$ becomes a prefix of a new accepting run on $uvw$.  We know that
$r\preceq_F x$ and thus by forward extensibility $q\preceq r \preceq_F x$ gives
$q\preceq x$, which means that there is a mediator for $q$ and $x$. Similarly
as in the previous case, since $x$ used the transitions of $q$ only once, we
can derive that there exists an accepting run of $\A$ on $uvw$ that does not
involve new transitions.

The argumentation from the two above paragraphs can be used inductively for a
run where $r$ uses transitions of $q$ arbitrarily many times.


\paragraph{Mediated Preorder and Equivalence.} We formally define mediated preorder
for ABA analogically as we have defined it in
Chapter~\ref{chapter:ta_reduction} for tree automata.  Consider a reflexive and
transitive forward simulation $\preceq_F$ on $\A$, and a reflexive and
transitive backward simulation $\preceq_B$ induced by $\preceq_F$. Recall the
relation combination operator $\oplus$ defined in
Chapter~\ref{chapter:ta_reduction}. We call the relation ${\preceq_M} =
{{\preceq_F}\oplus {\preceq_B^{-1}}}$ a \emph{mediated preorder induced by
$\preceq_F$ and $\preceq_B$} and ${\eq{M}} =
{{\preceq_M}\cap{\preceq_M^{-1}}}$ a \emph{mediated equivalence induced by
$\preceq_F$ and $\preceq_B$}. By Lemma~\ref{lemma:composition-unique},
$\preceq_M$ is a unique maximal preorder satisfying ${\preceq_F} \subseteq
{{\preceq_F}\oplus{\preceq_B^{-1}}}\subseteq{{\preceq_F}\circ{\preceq_B^{-1}}}$.

\paragraph{Ambiguity.} To make the mediated equivalence applicable, we must
pose one more requirement. Namely, we require that the transitions of the given
ABA are not \emph{$\preceq_F$-ambiguous}, meaning that no two states on the
right hand side of a transition are forward equivalent. Intuitively, allowing
such transitions goes against the spirit of the backward simulation. For a
mediator $p$ to backward simulate a state $r$ wrt. rules $\rho_1: p' \xtr{a}
P\cup \{p\},p\not\in P$, and $\rho_2: r' \xtr{a} R\cup \{r\},r\not\in R$, it
must be the case that each state $x$ in the context $P$ of $p$ within $\rho_1$
is less restrictive (i.e., forward bigger) than some state $y$ in the context
$R$ of $r$ within $\rho_2$. The state $r$ itself is not taken into account when
looking for $y$ because we aim at extending its behaviour by collapsing (and it
could then become less restrictive than the appropriate $x$). In the case of
$\preceq_F$-ambiguity, the spirit of this restriction is in a sense broken
since the forward behaviour of $r$ may still be taken into account when
checking that the context of $p$ is less restrictive than that of $r$. This is
because the behaviour of $r$ appears in $R$ as the behaviour of some other
forward equivalent state $r''$ too. Consequently, $r$ and $r''$ may back up
each other in a circular way when checking the restrictiveness of the contexts
within the construction of the backward simulation. Both of them can then seem
extensible, but once their behaviour gets extended, the restriction of their
context based on their own original behaviour is lost, which may then increase
the language (an example of such a scenario is given below). However, in
Section~\ref{section:algorithm}, we show that $\preceq_F$-ambiguity can be
efficiently removed.

\begin{example}[mediated minimization cannot be used on an ambiguous ABA]
Consider the following ABA $\A = (\{a,b\},\{s_0,s_1,s_2,s_3,s_4\},s_0,\delta,\{s_4\})$ where
$$\begin{array}{ll}
s_0 \xrightarrow{a}\{s_1,s_2,s_3\},& s_3 \xrightarrow{b}\{s_4\},\\
s_1 \xrightarrow{b}\{s_4\},        & s_3 \xrightarrow{a}\{s_1,s_2,s_3\},\\
s_2 \xrightarrow{b}\{s_4\},        & s_4 \xrightarrow{a}\{s_4\} 
\end{array}$$
are transitions of $\A$. The maximal forward simulation relation $\preceq_F$ in $\A$ is
$$\begin{array}{ll}
\{
(s_0, s_0), (s_0, s_3), (s_1, s_1), (s_1, s_2), (s_1, s_3),\\ (s_2, s_1), (s_2, s_2), (s_2, s_3), (s_3, s_3), (s_4, s_4)\}. 
\end{array}$$
From $s_1\equiv_F s_2$ and the transition $s_0 \xrightarrow{a}\{s_1,s_2,s_3\}$ we can find that $\A$ is $\preceq_F$-ambiguous. The maximal backward simulation relation $\preceq_B$ parametrised with $\preceq_F$ is
$$\begin{array}{ll}
\{
(s_0, s_0), (s_1, s_1), (s_1, s_2), (s_1, s_3), (s_2, s_1),\\ (s_2, s_2), (s_2, s_3), (s_3, s_1), (s_3, s_2), (s_3, s_3), (s_4, s_4)\} 
\end{array}$$
and the mediated preorder $\preceq_M$ is
$$\begin{array}{ll}
\{
(s_0, s_0), (s_0, s_1), (s_0, s_2), (s_0, s_3), (s_1, s_1), (s_1, s_2), (s_1, s_3),\\ (s_2, s_1), (s_2, s_2), (s_2, s_3), (s_3, s_1), (s_3, s_2), (s_3, s_3), (s_4, s_4)
\}.
\end{array}$$

If we collapse states wrt. $\preceq_M$ (i.e., merge the three states $s_1$, $s_2$, and $s_3$), we will get the following ABA $\A' = (\{a,b\},\{s_0,s_1,s_2\},s_0,\delta,s_2)$ where
$$
s_0 \xrightarrow{a}\{s_1\},
s_1 \xrightarrow{a}\{s_1\},
s_1 \xrightarrow{b}\{s_2\},
s_2 \xrightarrow{a}\{s_2\} 
$$
are transitions of $\A'$.
Note that $\A'$ accepts the word $aaba^\omega$, but $\A$ does not.
\qed
\end{example}

\subsection{Quotienting Automata According to Mediated Equivalence Preserves
Language}
In this section, we give a formal proof that under the assumption that $\A$ is
$\preceq_F$-unambiguous, quotienting with respect to mediated equivalence
preserves the language.  The proof roughly follows the pattern of the proof in
Chapter~\ref{chapter:ta_reduction} that quotienting tree automata according to
the mediated equivalence preserves language. However, the fact that we are
dealing with infinite tree runs with the B\"uchi accepting condition and that
two accepting runs on the same word need not be isomorphic makes the argument
significantly more complicated. 

\paragraph{Quotient Automata versus Extended Automata.}
As already mentioned, quotienting can be seen as a simpler operation of adding
transitions and accepting states which simplifies the
forthcoming reasoning. Let $\A = (\Sigma,Q,\iota,\delta,\alpha)$ be an ABA and let
$\equiv$ be an equivalence on $Q$ such that ${\equiv} =
{\preceq\cap\preceq^{-1}}$ for some preorder $\preceq$.  We will use
$\A/{\equiv}$ to denote the quotient of $\A$ wrt$.$ $\equiv$ that arises by merging
$\equiv$-equivalent states of $\A$, and $\A^+_\preceq$ will stand for the automaton \emph{extended}
according to $\preceq$, that is created as follows: for every two states $q,r$
of $\A$ with $q\preceq r$, (i) add all outgoing transitions of $q$ to $r$, (ii)
if $q\equiv r$ and $q$ is final, make $r$ final.

Formally, the automata $\A/{\equiv}$ and $\A^+_\preceq$ are defined as follows. Let $Q/{\equiv}$ denote
the partitioning of $Q$ w.r.t$.$ $\equiv$, and let $[q]$ denote the equivalence
class of $\equiv$ containing $q$. Then
$\A/{\equiv} = (\Sigma, Q/{\equiv},[\iota],\delta/{\equiv},\{[q]\mid q\in\alpha\})$ and
$\A^+_\preceq = (\Sigma,Q,\delta^+_\preceq,\iota,\alpha^+_\preceq)$ where
$\alpha^+_\preceq = \{p \mid \exists q\in\alpha.\ q\equiv p \}$ and,
for each $a\in\Sigma$, $q\in Q$,
$\delta/{\equiv}([q],a) = \bigcup_{p\in[q]}\{\{[p']\mid p'\in P\}\mid {P\in\delta(p,a)}\}$ and
$\delta^{+}_\preceq(q,a) = \bigcup_{p\in Q \wedge p\preceq q}
\delta(p,a)$.

The following lemma implies that if adding
transitions and accepting states according to $\preceq$ preserves the language,
then quotienting according to $\equiv$ preserves the language too.

\begin{lemma}\label{quotient=extended:lemma}
$\L(\A/{\equiv}) \subseteq \L(\A^+_\preceq)$.
\end{lemma}

\begin{proof}
Let $\A^+_\equiv = (\Sigma,Q,\iota,\delta^+_\equiv,\alpha^+_\equiv)$ be the automaton extended according to $\equiv$.
Observe that states $q$ and $r$ with $q\equiv r$ are forward simulation equivalent in $\A^+_\equiv$.
($q$ and $r$ are in $\A^+_\equiv$ either both accepting or both nonaccepting, and for all $a\in\Sigma$, $\delta^+_\equiv(q,a) = \delta_\equiv^+(r,a)$).
Gurumurthy et al. in \cite{gurumurthy:complementing} prove that quotienting with respect to forward simulation preserves language.
Therefore, $\L(\A/{\equiv}) = \L(\A^+_\equiv)$. It is also easy to see that $\L(\A^+_\equiv)\subseteq\L(\A^+_\preceq)$, as $\A^+_\preceq$ has a richer transition function than $\A^+_\equiv$ and $\alpha^+_\preceq = \alpha^+_\equiv$. Thus, $\L(\A/{\equiv})=\L(\A^+_\equiv)\subseteq\L(\A^+_\preceq)$.
\end{proof}


We now give the proof that extending automata according to the
mediated preorder preserves the language. 
For the rest of the section, we fix an ABA $\A =
(\Sigma,Q,\iota,\delta,\alpha)$, a reflexive and transitive forward simulation
$\preceq_F$ on $\A$ such that $\A$ is $\preceq_F$-unambiguous, and a reflexive
and transitive backward simulation $\preceq_B$ on $\A$ parametrised by
$\preceq_F$. Let ${\preceq_M}$ be the mediated preorder induced by $\preceq_F$
and $\preceq_B$, and let $\A^+ = (\Sigma,Q,\iota,\delta^+,\alpha^+)$ be the automaton extended according to
$\preceq_M$ (we omit the subscript $\preceq_M$ for the ease of notation). Let ${\equiv_M} = {\preceq_M} \cap{\preceq_M^{-1}}$.

We want to prove that $\L(\A^+) = \L(\A)$. The nontrivial part is showing that
$\L(\A^+) \subseteq \L(\A)$---the converse is obvious. To prove $\L(\A^+) \subseteq \L(\A)$, we need to show that, for every accepting run of $\A^+$ on a word $w$, there is an accepting run of $\A$ on $w$. 
We first prove Lemma~\ref{missing:lemma}, which shows how partial runs of $\A$ with an increased power of their leaves (wrt. $\preceq_M$) can be built incrementally from other runs of $\A$, bridging the gap between $\A$ and $\A^+$. Then we prove Lemma~\ref{adding:lemma} saying that for every partial run on a word $w$ of $\A^+$, there is a partial run of $\A$ on $w$ that is more accepting (recall that partial runs are finite).
By carry this result over to infinite runs we get the proof that extending
automata according to $\preceq_M$, and thus also quotienting wrt. $\equiv_M$, preserves language.

\paragraph{Extension Function and Covering.}
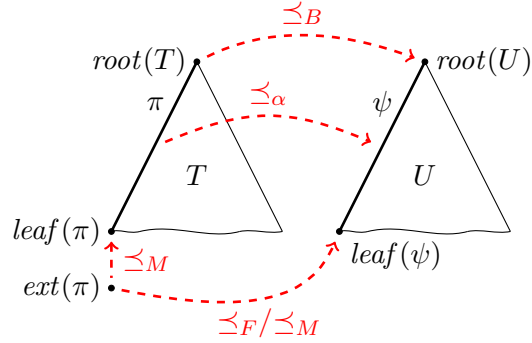
\begin{figure}
\centerline{
\begin{tikzpicture}[scale=0.75]

\node (roott) at (0,0) {};
\node (rootu) at (4,0) {};
\draw[fill=black] (roott) circle (0.5mm);
\draw[fill=black] (rootu) circle (0.5mm);
\node[left] at (roott) {$\root(T)$};
\node[right] at (rootu) {$\root(U)$};
\node at (0,-2) {$T$};
\node at (4,-2) {$U$};
\draw[line width=1pt] (0,0) -- (-1.5,-3) node [near start,left]{$\pi$} node (pi) [midway] {};
\draw[fill=black] (-1.5,-3) circle (0.5mm);
\node (p) at (-1.5,-3) {};
\node (q) at (-1.5,-4) {};
\node[left] at (q) {$\ext(\pi)$};
\draw[fill=black] (q) circle (0.5mm);
\node[left] at (p) {$\leaf(\pi)$};

\draw (0,0) -- (1.5,-3);

\draw (-1.5,-3) .. controls (-1.35,-2.95) and (-1.2,-3.05) ..  (-1,-3);
\draw (-1,-3) .. controls (-0.4,-2.85) and (0.2,-3.15) ..  (0.75,-3);
\draw (0.75,-3) .. controls (1.0,-2.95) and (1.25,-3.05) ..  (1.5,-3);

\draw[line width=1pt] (4,0) -- (2.5,-3)  node  [near start,left]{$\psi$} node (psi) [midway]{};
\draw[fill=black](2.5,-3) circle (0.5mm);
\node (r) at (2.5,-3) {};
\node[anchor=north west] at (r) {$\leaf(\psi)$};

\draw (4,0) -- (5.5,-3);

\draw (2.5,-3) .. controls (2.65,-2.95) and (2.8,-3.05) ..  (3,-3);
\draw (3,-3) .. controls (3.6,-2.85) and (4.2,-3.15) ..  (4.75,-3);
\draw (4.75,-3) .. controls (5.0,-2.95) and (5.25,-3.05) ..  (5.5,-3);

\draw[dashed,red,line width=1pt,->] (pi) .. controls (0.9,-0.8) and (1.6,-0.8) ..
node[above=0mm]{$\preceq_\alpha$}
(psi);

\draw[dashed,red,line width=1pt,->] (q) .. controls (1,-4.5) and (2,-4.2) ..
node[below]{${\preceq_F}{/}{\preceq_M}$}
(r);

\draw[dashed,red,line width=1pt,->] (roott) .. controls (1.2,0.6) and (2.4,0.6) ..
node[above]{$\preceq_B$}
(rootu);

\draw[dashed,red,line width=1pt,<-] (p) -- (q) node[midway,right]{$\preceq_M$};
\end{tikzpicture}}
\caption{$U$ strongly/weakly covers $T$ w.r.t$.$ $\ext$}
\label{figure:covering}
\end{figure}

Consider a partial run $T$ of $\A$ on a word $w$, we choose for each leaf $p$
of $T$ an $\preceq_M$-smaller state $p'$.  Suppose that we allow $p$ to make
one step using the transitions of $p'$ or to become accepting if $p'$ is
accepting and $p'\equiv_M p$. (Thus, we give the leaves of $T$ a part of the
power they would have in $\A^+$). We will show that there exists a partial run
$U$ of $\A$ on $w$ such that (1) it is more accepting than $T$, and (2) the
leaves of $U$ can mimic the next step of the leaves of $T$ even if the leaves
of $T$ use their extended power.

The above is formalised in Lemma~\ref{missing:lemma} using the following
notation. For a partial run $T$ of $\A$ on $w$, we define $\ext$ as an
\emph{extension function} that assigns to every branch $\pi$ of $T$ a state
$\ext(\pi)$ such that $\ext(\pi)\preceq_M \leaf(\pi)$.

Let $U$ be a partial run of $\A$ on $w$. For two branches $\pi \in \branch(T)$
and $\psi \in \branch(U)$, we say that $\psi$ \emph{strongly covers} $\pi$ wrt$.$
$\ext$, denoted $\pi\fcov\psi$, iff $\pi\preceq_\alpha\psi$ and
$\ext(\pi)\preceq_F\leaf(\psi)$. Similarly, we say that $\psi$ \emph{weakly
covers} $\pi$ wrt$.$ $\ext$, denoted $\pi\mcov\psi$, iff $\pi\preceq_\alpha\psi$
and $\ext(\pi)\preceq_M\leaf(\psi)$. We extend the concept of covering to
partial runs as follows. We write $T\fcov U$ ($U$ \emph{strongly covers} $T$
wrt$.$ $\ext$) iff $\branch(T)\fcov^{\forall\exists}\branch(U)$ and
$\root(T)\preceq_B\root(U)$. Likewise, we write $T\mcov U$ ($U$ \emph{weakly
covers} $T$  wrt$.$ $\ext$) iff $\branch(T)\mcov^{\forall\exists}\branch(U)$ and
$\root(T)\preceq_B\root(U)$. See Figure~\ref{figure:covering} for an illustration.
Note that we have ${\fcov} \subseteq {\mcov}$ for
branches as well for partial runs because
${\preceq_F}\subseteq{\preceq_M}$---the strong covering implies the weak one.

\begin{lemma}\label{missing:lemma}
For any partial run $T$ of $\A$ on a word $w$ with an extension function $\ext$,
there is a~partial run $U$ of $\A$ on $w$ with 
$T\fcov U$.
\end{lemma}

Proving Lemma~\ref{missing:lemma} is the most intricate part of the proof of Theorem~\ref{adding:theorem}.
We now introduce the concepts used within the proof, 
prove auxiliary Lemma~\ref{progress:lemma}, and subsequently present the proof of Lemma~\ref{missing:lemma} itself. 

Observe that $\root(T)\preceq_B\root(T)$, and every branch of $T$ weakly covers
itself, which means that $T\mcov T$.  Within the proof of
Lemma~\ref{missing:lemma}, we will show how to reach $U$ by a chain of partial
runs derived from~$T$. The partial runs within the chain will all weakly cover
$T$. Runs further from $T$ will in some sense cover $T$ more strongly than the
runs closer to $T$ and the last partial run of the chain will cover $T$
strongly. In the following paragraph, we formulate what it means that a
partial run weakly covering $T$ covers $T$ more strongly than another partial
run. 

\paragraph{The Relation of Covering $T$ More Strongly.} To define the relation
of covering $T$ more strongly on partial runs that weakly cover $T$, we
concentrate on those branches of partial runs that cause that they do not cover
$T$ strongly. Let $V$ be a partial run of $\A$ on $w$ with $T\mcov V$.  We call
a branch $\psi\in\branch(V)$ \emph{strict weakly covering} if there is no
$\pi\in\branch(T)$ with $\pi\preceq_\ext\psi$ (there are only some
$\pi\in\branch(T)$ with $\pi\mcov\psi$). Let  $\sw(V)$ denote the tree which is
the subset of $V$ containing prefixes of strict weakly covering branches of $V$
wrt. $T$.  Note that $T\fcov V$ iff $V$ contains no strict weakly covering
branches, which is equivalent to $\sw(V)=\emptyset$. Given a partial run $W$ of
$\A$ on $w$, we will define which of $V$ and $W$ cover $T$ more strongly by
comparing $\sw(V)$ and $\sw(W)$. For this, we need the following definitions.

Given a finite tree $X$ over $Q$ and $\tau\in Q^+$, we define the \emph{tree
decomposition} of $X$ according to $\tau$ as the sequence of (finite) sets of
paths $\bs \tau X =
X\ominus_1\tau,X\ominus_2\tau,\ldots,X\ominus_{|\tau|}\tau$. We also let $\bs
\epsilon X = \branch(X)$ (it is a sequence of length 1). A substantial property
of tree decompositions is that under the condition that
$\tau\not\in\branch(X)$, $\bs \tau X = \emptyset\ldots\emptyset$ implies that
$X= \emptyset$. Notice that if $\tau\in\branch(X)$, $\bs \tau X =
\emptyset\ldots\emptyset$ does not imply $X= \emptyset$ as $\tau$ could be the
only branch of $X$. This is important as for a partial run $Y$ and $\tau'\in
Y$, if $\tau'\not\in\branch(Y)$, the implications $\bs {\tau'} {\sw(Y)} =
\emptyset \ldots \emptyset \implies \sw(Y) = \emptyset \implies T\fcov Y$
hold.  However, the first implication does not hold if $\tau'\in\branch(Y)$. 

Let $\tau_V\in V\cup\{\epsilon\}$ and $\tau_W\in W\cup\{\epsilon\}$ be such
that $\tau_V\not\in\branch(\sw(V))$ and $\tau_W\not\in\branch(\sw(W))$. We say
that \emph{$W$ covers $T$ more strongly} than $V$ wrt. $\tau_V$ and $\tau_W$,
denoted $V\morestrong{\tau_V}{\tau_W} W$, iff $\root(V)\preceq_B\root(W)$ and
$\bs {\tau_V} {\sw(V)} \sqsubset \bs {\tau_W} {\sw(W)}$ where $\sqsubset$ is a
binary relation on finite sequences of sets of paths defined as follows:

For two sets of paths $P$ and $P'$, we use $P \prec_F^{\forall\exists} P'$ to denote that
$P \preceq_F^{\forall\exists} P'$ but not $P' \preceq_F^{\forall\exists} P$.
In other words, the upward closure of $P'$ wrt. $\preceq_F$ is a proper subset of the upward closure of $P$ wrt. $\preceq_F$.
Then, for two finite sequences $S,S'\in (2^{Q^+})^+$ of sets of paths, $S\sqsubset S'$ iff
there is some
$k\in\nat, k\leq\min\{|S|,|S'|\}$,
such that
$S_k\prec_F^{\forall\exists} S'_k$ and for all $1\leq j < k$,
$S_j\preceq_F^{\forall\exists} S'_j$.  

Given $c\in\nat$, we say that a sequence $S$ of sets of paths is \emph{$c$-bounded} if
$|S|\leq c$ and also the length of every path in every $S_i,1\leq i\leq |S|$ is at most $c$.
Lemma~\ref{lemma:chainto0} below shows that every maximal increasing chain of
$c$-bounded sequences related by $\sqsubset$ eventually arrives to
$\emptyset\ldots\emptyset$.  This will allow us to show that every maximal
sequence of partial runs that cover $T$ more and more strongly must terminate
by a partial run that covers $T$ strongly. 
\begin{lemma}\label{lemma:chainto0}
Given a constant $c$, every maximal increasing chain of $c$-bounded sequences related by
$\sqsubset$ eventually terminates by
$\emptyset\ldots\emptyset$. 
\end{lemma}
\begin{proof}
First, observe that for every sequence $S$ of sets of paths with $S \neq
\emptyset\ldots\emptyset$, it holds that $S \sqsubset
\emptyset\ldots\emptyset$.  This is easy to see since
$\emptyset\preceq_F^{\forall\exists}\emptyset$ and
$\emptyset\prec_F^{\forall\exists} X$ for any nonempty $X\in 2^{Q^+}$.
Therefore, to prove the lemma, it is sufficient to show that $\sqsubset$ does
not allow infinite increasing chains of $c$-bounded sequences. 

Let ${S}= S(1)\sqsubset S(2) \sqsubset S(3) \sqsubset \cdots$ be such a chain
of $c$-bounded sequences. We will show that $S$ must be finite. Observe that the domain of possible $c$-bounded $S(i)$s is
finite since there is only finitely many of paths with the length bounded by
$c$ ($Q$ is finite).  Therefore, if $S$ is an infinite chain, there has to be $i$ and $j$ with
$i<j$ such that $S(i) = S(j)$. We will argue that this is not possible by
showing that $\sqsubset$ is irreflexive and transitive, which means that it
does not allow loops (if there was a loop $X\sqsubset \cdots \sqsubset X$, then
by transitivity, $X\sqsubset X$ which contradicts irreflexifity). 

Irreflexivity of $\sqsubset$ may be shown as follows.  Let $S\sqsubset S$ for some $c$-bounded sequence $S$. By the
definition of $\sqsubset$, there is $k\in\nat$ such that
$S_i\preceq_F^{\forall\exists}S_i$ for all $i\in\nat$ smaller than $k$, and
$S_k\prec_F^{\forall\exists}S_k$. However, this is clearly not possible since since the upward closure of $S_k$ wrt. $\preceq_F$ would have to be a proper subset of itself. 

Transitivity of $\sqsubset$ can be shown as follows.  Let $S,S',S''$ be three
$c$-bounded sequences with $S\sqsubset S'\sqsubset S''$.
By the definition of $\sqsubset$, there is $k\in \nat$ such that 
$S_i\preceq_F^{\forall\exists}S'_i$ for all $i\in\nat$ smaller than $k$, and
$S_k\prec_F^{\forall\exists}S'_k$; and there is $k'\in \nat$ such that
$S'_i\preceq_F^{\forall\exists}S''_i$ for all $i\in\nat$ smaller than $k'$, and
$S'_{k'}\prec_F^{\forall\exists}S''_{k'}$.
Let $l = \min\{k,k'\}$. By transitivity of $\preceq_F^{\forall\exists}$, we
have that $S_i\preceq_F^{\forall\exists}S''_i$ for all $i\in\nat$ smaller than
$l$. Then, for the $l$th position, we have that
$S_l\prec_F^{\forall\exists}S'_l\prec_F^{\forall\exists}S''_l$ or
$S_l\preceq_F^{\forall\exists}S'_l\prec_F^{\forall\exists}S''_l$  or
$S_l\prec_F^{\forall\exists}S'_l\preceq_F^{\forall\exists}S''_l$. All these three possibilities give $S_l\prec_F^{\forall\exists}S''_l$, and thus $S\sqsubset S''$.

\end{proof}

The last ingredient we need for the proof of Lemma~\ref{missing:lemma} is to
show that for every maximal sequence of partial runs that cover $T$ more and
more strongly, the underlying $\sqsubset$-related sequence is also maximal.
Particularly, we need to show that for any partial run weakly (but not
strongly) covering $T$, we are always able to construct a partial run covering
$T$ more strongly. This is stated by the following lemma.

\begin{lemma}\label{progress:lemma}
Given a partial run $V$ of $\A$ on $w$ s.t. $T\mcov V$, $T\not\fcov V$, and $\tau_V\in V\cup\{\epsilon\}$ with $\tau_V\not\in\branch(\sw(V))$,
we can construct a partial run $W$  of $\A$ on $w$ with $T\mcov W$ and a~path $\tau_W\in W$ with $\tau_W\not\in\branch(\sw(W))$ such that
$V\morestrong{\tau_V}{\tau_W} W$.
\end{lemma}

\begin{proof}
The proof 
relies on Lemma~\ref{backward:lemma}
and the definition of $\preceq_M$.  We first choose a suitable branch $\pi$ of
$\sw(V)$ as follows.  Let $1\leq k\leq |\tau_V|$ be some index such that
$\sw(V)\ominus_k\tau_V$ is nonempty. If $\tau_V = \epsilon$, then $k = 1$.
We choose some $\pi'\in \sw(V)\ominus_k\tau_V$ which is minimal wrt.
$\preceq_F$, meaning that there is no $\pi''\in \sw(V)\ominus_k\tau_V$ different
from $\pi'$ such that $\pi''\preceq_F\pi'$.  We put $\pi = \tau_V^k\pi'$.  We
note that this is the place where we use the $\preceq_F$-unambiguity
assumption.  If $\A$ was $\preceq_F$-ambiguous, there need not be a $k$ such
that $\sw(V)\ominus_k\tau_V$ contains a minimal element wrt. $\preceq_F$.

As $T\mcov V$, there is $\sigma\in\branch(T)$ with $\sigma\mcov\pi$.
From $\ext(\sigma)\preceq_M\leaf(\pi)$, there is a mediator $s$ with
$\ext(\sigma)\preceq_F s\succeq_B \leaf(\pi)$. We can apply Lemma~\ref{backward:lemma}
to $V$, $\pi$, $\leaf(\pi)$ and $s$, which give us a partial run $W$ and $\psi
\in \branch(W)$ with $\leaf(\psi) = s$ such that $\pi\preceq_B\psi$, and for
all $1\leq i \leq |\pi|$, ${V\ominus_i \pi} \preceq_F^{\forall\exists} {W\ominus_i
\psi}$.  
Let $\tau_W = \psi$.
The proof will be concluded by showing that (i) $T\mcov W$, (ii) $\tau_W\not\in\branch(\sw(W))$,
and (iii) $\bs {\tau_V} {\sw(V)}\sqsubset \bs {\tau_W} {\sw(W)}$, which implies $V \morestrong{\tau_V}{\tau_W}W$.

\emph{(i)} To show that $T\mcov W$, we proceed as follows.
Observe that for every $\phi\in \branch(W)\setminus\{\psi\}$ there is a branch
$\phi'\in \branch(V)\setminus\{\pi\}$ such that
$\leaf(\phi')\preceq_F\leaf(\phi)$ and $\phi'\preceq_\alpha\phi$.  This holds
because for all $1\leq i \leq |\pi|$, ${V\ominus_i \pi} \preceq_F^{\forall\exists}
{W\ominus_i \psi}$ and because $\pi\preceq_B\psi$
(To be more detailed, for every $\phi\in \branch(W)\setminus\{\psi\}$, $\phi = \psi^i\rho$ for some
$i$ and $\rho\in {W\ominus_i \psi}$. There must be $\rho'\in V\ominus_i
\pi$ with $\rho'\preceq_F\rho$. As $\pi\preceq_B\phi$, $\pi^i\preceq_B\phi^i$ which implies $\pi^i\preceq_\alpha\phi^i$. Similarly, $\rho'\preceq_F\rho$ implies $\rho'\preceq_\alpha\rho$ and also $\leaf(\rho') \preceq_F \leaf(\rho)$. Therefore, we can construct the branch
$\phi'=\pi^i\rho'\in\branch(V)\setminus\{\pi\}$ with
$\pi^i\rho'\preceq_\alpha\psi^i\rho = \phi$ and $\leaf(\pi^i\rho') \preceq_F \leaf(\psi^i\rho)$).
We also know that since $T\mcov V$,
$\branch(T)\mcov^{\forall\exists}\branch(V)$.
Thus, by the definition of $\mcov$,
we have that for every
$\phi\in\branch(W)\setminus\{\psi\}$, there are $\phi'\in\branch(V)$ and
$\phi''\in\branch(T)$ with $\phi''\preceq_\alpha\phi'\preceq_\alpha\phi$ and
$\ext(\phi'')\preceq_M\leaf(\phi')\preceq_F\leaf(\phi)$. This by
transitivity of $\alpha$ and the definition of $\preceq_M$ gives
$\phi''\preceq_\alpha\phi$ and $\ext(\phi'')\preceq_M\leaf(\phi)$, which means
$\phi''\mcov\phi$. 
To see that also $\psi$ is weakly covering, observe
that since $\sigma\mcov\pi$, we have
$\sigma\preceq_\alpha\pi\preceq_B\psi$ and $\ext(\sigma)\preceq_F s = \leaf(\psi)$,
which by ${\preceq_B}\subseteq{\preceq_\alpha}$ and transitivity of $\preceq_\alpha$
gives even $\sigma\fcov\psi$ (immediately implying $\sigma\mcov\psi$).
Finally, from $\root(T) \preceq_B \root(V)$ (implied by $T\mcov V$),
$\pi\preceq_B\psi$, and transitivity of $\preceq_B$,
$\root(T)\preceq_B\root(W)$. We have shown that $T\mcov W$.

\emph{(ii)} Showing that $\psi\not\in\branch(\sw(W))$ is easy.  In the above
paragraph we have just shown that $\sigma\fcov\psi$,
thus $\psi$ is not a strict weakly covering branch.

\emph{(iii)}
To show that $\bs {\tau_V} {\sw(V)} \sqsubset \bs {\psi} {\sw(W)}$, we will
argue that
(a) for all $1\leq i < k$, it holds that
${\sw(V)\ominus_i \tau_V} \preceq_F^{\forall\exists} {\sw(W)\ominus_i \psi}$
and that
(b)
${\sw(V)\ominus_k \tau_V} \prec_F^{\forall\exists} {\sw(W)\ominus_k \psi}$.
Notice first that for any partial run $X$ of $\A$ and $\tau\in X$ with $\tau\not\in\branch(\sw(X))$,
for all $1\leq j\leq |\tau|$, $\sw(X)\ominus_j \tau \subseteq X\ominus_j\tau$.
Recall that $\tau_V^k = \pi^k$, that ${\sw(V)\ominus_k \tau_V}$ is nonempty, and that
for all $1\leq i < |\pi|$,  ${V\ominus_i \pi} \preceq_F^{\forall\exists} {W\ominus_i \psi}$.

We first show that for all $1\leq i < |\pi|$,
${\sw(V)\ominus_i \pi}\preceq_F^{\forall\exists}{\sw(W)\ominus_i \psi}$.
For every $\phi\in \sw(W)\ominus_i \psi$, there is at least one
$\phi'\in {V\ominus_i \pi}$ with $\phi'\preceq_F\phi$
(because ${V\ominus_i \pi} \preceq_F^{\forall\exists} {W\ominus_i \psi}$ and
${\sw(W)\ominus_i \psi}\subseteq{W\ominus_i \psi}$).
We will show by contradiction that $\phi'\in \sw(V)\ominus_i \pi$ which will imply
${\sw(V)\ominus_i \pi}\preceq_F^{\forall\exists}{\sw(W)\ominus_i \psi}$.
Suppose that $\phi'\not\in \sw(V)\ominus_i \pi$.  Then the branch $\pi^i\phi'$ of $V$
is not strict weakly covering, and as $T\mcov
V$, we have that there is some $\phi''\in\branch(T)$ with
$\phi''\fcov\pi^i\phi'$.  As $\pi\preceq_B\psi$, we have that
$\pi^i\preceq_\alpha\psi^i$. As $\phi'\preceq_F\phi$, we have that
$\phi'\preceq_\alpha\phi$ and $\leaf(\phi')\preceq_F\leaf(\phi)$. This together
with $\phi''\fcov\pi^i\phi'$ gives that
$\phi''\preceq_\alpha\pi^i\phi'\preceq_\alpha\psi^i\phi$  and
$\ext(\phi'')\preceq_F\leaf(\pi^i\phi')\preceq_F\leaf(\psi^i\phi)$. By
transitivity of $\preceq_\alpha$ and $\preceq_F$ and by the definition of
$\fcov$, we obtain $\phi''\fcov\psi^i\phi$. This contradicts with
the fact that $\psi^i\phi$ is strict weakly covering
(as $\phi\in \sw(W)\ominus_i \psi$) and therefore it must be the case that $\phi'\in
\sw(V)\ominus_i \pi$.

\emph{(a)}
The fact that for all $1\leq i < k$,
${\sw(V)\ominus_i \tau_V} \preceq_F^{\forall\exists} {\sw(W)\ominus_i \psi}$
is implied by the result of the previous paragraph, because $\tau_V^k = \pi^k$
(thus ${\sw(V)\ominus_i \tau_V} = {\sw(V)\ominus_i \pi}$).

\emph{(b)}
It remains to show that ${\sw(V)\ominus_k \tau_V} \prec_F^{\forall\exists} {\sw(W)\ominus_k \psi}$.
By the definitions of $\ominus_k$, $\pi$ and $\tau_V$, it holds that
${\sw(V)\ominus_k \tau_V} \supset {\sw(V)\ominus_k \pi}$. (To see this, recall
that $\pi$ is strict weakly covering, but $\tau_V$ is not. Therefore,
${\sw(V)\ominus_k \pi}  = {\sw(V)\ominus_k \tau_V} \setminus
\branch(\sw(V)(\pi^{k+1}))$).
Since $\supset$ implies $\preceq_F^{\forall\exists}$, we have that 
${\sw(V)\ominus_k \tau_V} \preceq_F^{\forall\exists} {\sw(V)\ominus_k \pi}$. 
Moreover, since $\pi'\not\in {\sw(V)\ominus_k \pi}$ and $\pi'$ is a minimal element of ${\sw(V)\ominus_k \tau_V}$,
${\sw(V)\ominus_k \pi} \preceq_F^{\forall\exists} {\sw(V)\ominus_k \tau_V}$ cannot hold (there is no  $\pi''\in {\sw(V)\ominus_k \pi}$ with $\pi''\preceq_F\pi'$), and therefore we have ${\sw(V)\ominus_k \tau_V} \prec_F^{\forall\exists}{\sw(V)\ominus_k \pi}$.
Finally, ${\sw(V)\ominus_k \tau_V} \prec_F^{\forall\exists}{\sw(V)\ominus_k \pi}\preceq_F^{\forall\exists} {\sw(W)\ominus_k \psi}$ gives  ${\sw(V)\ominus_k \tau_V} \prec_F^{\forall\exists} {\sw(W)\ominus_k \psi}$.
This completes
the part (iii) of the proof and we can conclude that $V \morestrong{\tau_V}\psi W$.
\end{proof}

With Lemma~\ref{progress:lemma} in hand, we are finally ready to prove Lemma~\ref{missing:lemma}.

\begin{proof}[Proof of Lemma~\ref{missing:lemma}] If $T\fcov T$, we are done as
in the statement of the lemma, we can take $T$ to be $U$. So, suppose that
$T\npreceq_\ext T$. Observe that $\root(T)\preceq_B\root(T)$, and every branch
of $T$ weakly covers itself, which means that $T\mcov T$.  We construct a run
$U$ strongly covering $T$ as follows. Starting from $T$ and $\epsilon$, we can
construct a chain $T \morestrong{\epsilon}{\tau_1}T_1
\morestrong{\tau_1}{\tau_2}T_2\morestrong{\tau_2}{\tau_3}T_3\ldots$ of partial
runs that more and more strongly cover $T$ by successively applying
Lemma~\ref{progress:lemma} for each $i$, $\tau_i\in T_i$,
$\tau_i\not\in\branch(\sw(T_i))$, and $T\mcov T_i$.  Observe that by the
definition of stronger covering, we have that $\bs {\epsilon} {\sw(T)}
\sqsubset \bs {\tau_{1}} {\sw(T_1)}\sqsubset \bs {\tau_{2}}{\sw(T_2)}\sqsubset
\bs {\tau_{3}} {\sw(T_3)}\ldots$. 

Notice now that for each $i$, since $T\mcov T_i$,
$\height(T_i) \leq \height(T)$. Therefore, since length of $\tau_i$ is bounded by $\height(T)$, the
length of $\bs {\tau_{i}} {\sw(T_i)}$ is bounded by $\height(T)$ too. Since lengths of all paths in the sets within $\bs {\tau_{i}} {\sw(T_i)}$ are obviously bounded by $\height(T)$ as well, $\bs {\tau_{i}} {\sw(T_i)}$ is a $\height(T)$-bounded sequence.
%
%
Therefore, by Lemma~\ref{lemma:chainto0}, the
chain must eventually arrive to its last $T_k$ and $\tau_k$ with $\bs {\tau_{k}}
{\sw(T_k)}= \emptyset\ldots\emptyset$. As $\bs {\tau_{k}} {\sw(T_k)}=
\emptyset\ldots\emptyset$, $\sw(T_k)$ has to be empty, which implies that $T
\fcov T_k$. We can put $U = T_k$ and Lemma~\ref{missing:lemma} is proven.
\end{proof}

We use Lemma~\ref{missing:lemma} to prove Lemma~\ref{adding:lemma}. Informally,
it says that even despite the poorer transition relation and smaller set
of accepting states, $\A$ can answer to any partial run of $\A^+$ by a more
accepting partial run.
To express this formally, we need to define the following weaker version $\moreacc$ of the relation of being more accepting that takes into account $\alpha^+$ on the left and $\alpha$ on the right. This is, for states $q$ and $r$, $q\moreacc r$ iff $q\in\alpha^{+}\implies r\in\alpha$.
For two paths $\pi,\psi\in Q^+$, $\pi\moreacc\psi$ iff
$|\pi| = |\psi|$ and for all $1\leq i \leq |\pi|$,
$\pi_i\in\alpha^{+}\implies\psi_i \in\alpha$. Last, for finite trees $T$ and $U$ over
$Q$, we use $T\moreacc U$ to denote that  $\branch(T)\moreacc^{\forall\exists}\branch(U)$.

\begin{lemma}\label{adding:lemma}
For any partial run $T$ of $\A^+$ on $w\in\Sigma^\omega$, there exists a partial run
$U$ of $\A$ on $w$ such that $\root(T)\preceq_B \root(U)$ and $T\moreacc U$.
\end{lemma}

\begin{proof}
By induction to the structure of $T$, using Lemma~\ref{missing:lemma} within
the induction step.  To make the induction argument pass, we will prove a
stronger variant of the lemma.  Particularly, we will replace the relation
$\moreacc$ within the statement of the lemma by its stronger variant
$\moreaccm$ which is defined as follows. Given paths $\pi$ and $\psi$,
$\pi\moreaccm\psi$ iff $\pi\moreacc\psi$ and $\leaf(\pi)\preceq_M\leaf(\psi)$.
For two partial runs $V$ and $W$, we use $V\moreaccm W$ to denote that
$\branch(V)\mathrel{(\moreaccm)^{\forall\exists}}\branch(W)$.  Apparently,
${\moreaccm}\subseteq{\moreacc}$ for paths as well as for partial runs.
\medskip

\noindent
\emph{A stronger variant of the lemma:} For any partial run $T$ of $\A^+$ on
$w\in\Sigma^\omega$, there exists a partial run $U$ of $\A$ on $w$ such that
$\root(T)\preceq_B \root(U)$ and $T\moreaccm U$.
\smallskip

It is obvious that the above statement implies the statement of the lemma.  We
will prove it by induction to the structure of $T$.  In the base case, $T =
\{q\}$ for some $q\in Q$.  If $q\not\in\alpha^+$, we can put $U = \{q\}$
($\preceq_M$ and $\preceq_B$ are reflexive).  If $q\in\alpha^+$, then by the
definition of $\alpha^+$, there is $p\in \alpha$ such that $p\equiv_M q$.  This
means that $q\preceq_M p$ and $p\preceq_M q$.  By the definition of
$\preceq_M$, there exists a mediator $s$ with $p\preceq_F s\succeq_B q$.  As
${\preceq_F}\subseteq{\preceq_\alpha}$, $s\in\alpha$.  Again by the definition
of $\preceq_M$, $q\preceq_M p\preceq_F s\succeq_B q$ gives us $q\preceq_M s
\succeq_B q $ and we can put $U = \{s\}$.

Suppose now that $T$ is not only a root and that the stronger variant of the
lemma holds for every partial run of $\A^+$ on $w$ that is a proper subset of
$T$.  We choose some $\pi\in T$ such that $\succ T(\pi)\neq\emptyset$ and for
every $p\in \succ T(\pi)$, $\succ T(\pi p) = \emptyset$. Notice that since $T$
is a finite tree, such $\pi$ always exists. Denote $P = \succ T(\pi)$ and $q =
\leaf(\pi)$.  Let $T' = T\setminus \{\pi p\mid p\in P\}$.  $T'$ is a partial
run of $\A^+$ on $w$ which is a proper subset of $T$, therefore we can apply
the induction hypothesis on it.  This gives us a partial run $V$ of $\A$ on $w$
such that $\root(T')\preceq_B \root(V)$ and $T'\moreaccm V$.

Let $\Bad_V\subseteq \branch(V)$ be the set such that $\psi \in \Bad_V$ iff
there is no $\phi \in \branch(T)$ such that $\phi\moreaccm\psi$, and let
$\Good_{V} = \branch(V)\setminus \Bad_V$.  Intuitively, $\Bad_V$ contains the
problematic branches because of which $T\moreaccm V$ does not hold. If $\Bad_V$
it is empty, then the relation holds and we can conclude the proof. We
continue assuming that $\Bad_V\neq\emptyset$.

By the definition of $\delta^+$ and because $q\xtr {w_{|\pi|}} P$ is a
transition of $\A^+$, there must be some $s\in Q, s\preceq_M q$ where $s\xtrlow
{w_{|\pi|}} P$ is a transition of $\delta$. We define an extension function
$\ext_V$  such that $\ext_V(\phi) = s$ for every $\phi\in \Bad_V$ and
$\ext_V(\psi) = \leaf(\psi)$ for every $\psi \in \Good_{V}$. To see that
$\ext_V$ conforms the definition of extension function, one has to show that for every branch $\phi\in\Bad_V$, $s\preceq_M\leaf(\phi)$. We know that
$T'\moreaccm V$ but not $T\moreaccm V$. Therefore, there is some branch $\phi'\in T'$ with $\phi'\moreaccm\phi$ such that $\phi'\not\in\branch(T)$ (if $\phi'$ was a branch of $T$, $\phi$ would not be in $\Bad_V$).
Notice that $\pi$ is
the only branch of $T'$ which is not a branch of $T$, which means that it must be the case that $\phi' = \pi$.
Therefore, since $s\preceq_Mq\preceq_M\leaf(\phi)$,
$s\preceq_M\leaf(\phi)$ holds.

By applying Lemma~\ref{missing:lemma} to $V$ and $\ext_V$, we get a
partial run $W$ of $\A$ on $w$ with $V\fcovx{V} W$.  Now, for each
$\psi\in\branch(W)$, there is $\phi\in\branch(V)$ with $\phi\fcovx{V}\psi$.  As
$T'\moreaccm V$, $\rho\moreaccm\phi$ for some $\rho\in\branch(T')$.  There are
two cases of how $\rho$ and $\psi$ may be related, depending on $\phi$:
\begin{enumerate}

\item
If $\phi\in \Good_V$, then $\ext(\phi) = \leaf(\phi)$.
In this case, by the definitions of $\moreaccm$ and $\fcovx{V}$, we have
$\rho\moreacc\phi\preceq_\alpha\psi$ and
$\leaf(\rho)\preceq_M\leaf(\phi)\preceq_F\leaf(\psi)$,
which gives
$\rho\moreacc\psi$ and
$\leaf(\rho)\preceq_M\leaf(\psi)$ (since $\preceq_M$ is forward extensible), meaning that $\rho\moreaccm\psi$.

\item
To analyse the case when $\phi\in \Bad_V$, recall that $\pi$ is the only branch
of $T'$ which is not a branch of $T$, and therefore $\pi$ is also the only
branch of $T'$ with $\pi \moreaccm\phi$.  Therefore, $\rho = \pi$. According to the
definition of $\ext_V$, $\ext_V(\phi) = s$. Since $\phi\fcovx{V} \psi$, we have
$\pi\moreaccm\phi\preceq_\alpha\psi$ which gives $\pi\moreacc\psi$. However,
since (contrary to the previous case 1.) $\ext_v(\phi)\neq\leaf(\phi)$, we cannot
guarantee any further relation between $\leaf(\phi)$ and $\leaf(\psi)$, and we
cannot derive that $\leaf(\pi)\preceq_M\leaf(\psi)$ and $\pi\moreaccm\psi$ need not hold. 
\end{enumerate}
We define the set $\Bad_W\subseteq\branch(W)$ such as $\psi\in \Bad_W$ iff there is no $\rho\in T$ with
$\rho\moreaccm\psi$ and we let $\Good_W = \branch(W)\setminus \Bad_V$.
This is, $\Bad_W$ contains the branches because of which $T \moreaccm W$ does not hold.
Note that if $\psi\in \Bad_V$, then all the $\phi\in\branch(V)$ with $\phi\fcovx{V}\psi$ are as in the case (2) above, i.e.,
$\pi$ is the only branch of $T'$ with $\pi \moreaccm\phi$.
By the definition of $\fcovx{V}$,  $s = \ext_V(\phi)\preceq_F\leaf(\psi)$.
Therefore, by the definition of $\preceq_F$ and since $s\xtrlow{w_{|\pi|}} P$, there must be some transition
$\leaf(\psi)\xtrlow {w_{|\pi|}} R_\psi$ of $\A$ where
$P\preceq_F^{\forall\exists} R_\psi$. We extend $W$ by firing these
transitions for every $\psi\in \Bad_W$, in which way we obtain a run $X = W \cup \{\psi
R_\psi\mid \psi \in \Bad_{W}\}$ of $\A$ on $w$.

Let us use $\New_X = \{\psi R_\psi\mid \psi \in \Bad_{W}\}$ to denote the
branches of $X$ that arose by firing the transitions.  Observe that
$\branch(X) = \Good_W \cup \New_X$. Recall that for all $\psi\in \Bad_W$,
$\pi\moreacc\psi$ and that for every $\psi \in \New_X$, there is some $p\in P$
such that $p\preceq_F \leaf(\psi)$.
We will define an extension function $\ext_X$ of $X$ as follows:
\begin{enumerate} \item If $\psi\in \Good_W$, $\ext_X(\psi) = \leaf(\psi)$.
\item If $\psi \in  \New_X$ and there is $p\in P$ with $p\preceq_F \leaf(\psi)$
and $p\moreacc \leaf(\psi)$, we let $\ext_X(\psi) = \leaf(\psi)$.  \item If
$\psi \in  \New_X$ and there is no $p\in P$ with $p\preceq_F \leaf(\psi)$ and
$p\moreacc \leaf(\psi)$, we proceed as follows.  By the definition of $\New_X$,
there is some $p'\in P$ such that $p'\preceq_F \leaf(\psi)$.  Since
${\preceq_F}\subseteq{\preceq_\alpha}$, $p'\preceq_F \leaf(\psi)$, and not
$p'\moreacc \leaf(\psi)$, it must be the case that $p'\not\in\alpha$,
$\leaf(\psi)\not\in\alpha$, and $p'\in\alpha^+$. This by the definition of
$\alpha^+$ means that there is some $v\in\alpha$ with $p'\equiv_M v$. We put
$\ext_X(\psi) = v$. 
\end{enumerate}

We apply Lemma~\ref{missing:lemma} to $X$ and $\ext_X$, which gives us a
partial run $U$ of $\A$ on $w$ with $X\fcovx{X} U$.  We will check that $U$
satisfies the statement of the stronger variant of the lemma. We will first
prove that $T\moreaccm U$.  For each $\tau\in\branch(U)$, there is
$\psi\in\branch(X)$ with $\psi\fcovx{X}\tau$.  We will derive that there is
some $\rho\in\branch(T)$ with $\rho\moreaccm\tau$. The argument depends on properties of $\psi$. Particularly, we have the following three cases. 
\begin{enumerate}
\item
If $\psi\in \Good_W$, then there is some $\rho\in T$ with $\rho\moreaccm\psi$.
Recall that $\ext_X(\psi) = \leaf(\psi)$ in this case.  Thus, by the
definitions of $\moreaccm$ and $\fcovx{X}$, we have
$\rho\moreacc\psi\preceq_\alpha\tau$ and
$\leaf(\rho)\preceq_M\leaf(\psi)\preceq_F\leaf(\tau)$, which gives
$\rho\moreacc\tau$ and $\leaf(\rho)\preceq_M\leaf(\tau)$, i.e.,
$\rho\moreaccm\tau$.
\item
If $\psi\in \New_X$ and there is some $p\in P$ with $p\preceq_F\leaf(\psi)$ and
$p\moreacc\leaf(\psi)$, then by the definition of $\ext_X$, $\ext_X(\psi) =
\leaf(\psi)$. Recall that as $\psi^{|\psi|-1}\in \Bad_W$,
$\pi\moreacc\psi^{|\psi|-1}$.  Therefore, also $\pi p\moreacc\psi$.  By the
definition of $\fcovx{X}$, we have that $\psi\preceq_\alpha\tau$ and
$\leaf(\psi)\preceq_F\leaf(\tau)$.  Finally, $\pi
p\moreacc\psi\preceq_\alpha\tau$ and
$p\preceq_F\leaf(\psi)\preceq_F\leaf(\tau)$ together imply that $\pi p
\moreaccm\tau$.
\item
If $\psi\in \New_X$ and there is no $p\in P$ with $p\preceq_F\leaf(\psi)$ and
$p\moreacc\leaf(\psi)$, then by the definition of $\ext_X$, there are $p'\in P$
with $p'\preceq_F \leaf(\psi)$ and $v\in\alpha$ with $v\equiv_M p'$ such that
$\ext_X(\psi) = v$.  By $\psi\fcovx{X}\tau$, we have $\psi\preceq_\alpha\tau$
and $v\preceq_F\leaf(\tau)$. Thus, since $\preceq_M$ is forward extensible,
$p'\equiv_M v \preceq_F \leaf(\tau)$ gives $p'\preceq_M\leaf(\psi)$.  As
${\preceq_F}\subseteq{\preceq_\alpha}$, we have that $\leaf(\tau)\in\alpha$ and
thus $p'\moreacc\leaf(\tau)$. As $\psi^{|\psi|-1}\in \Bad_W$, we have that
$\pi\moreacc\psi^{|\psi|-1}$. Together with $\psi\preceq_\alpha\tau$, this
gives $\pi p'\moreacc\tau$.  Therefore, $\pi p'\moreaccm\tau$.
\end{enumerate}

Since the above three cases cover all possible variants of $\psi$ and thus all branches of $U$, we have proven that
$T\moreaccm U$. Finally, it is easy to show that $\root(T)\preceq_B\root(U)$
since $\preceq_B$ is transitive and we know that
$\root(T)=\root(T')\preceq_B\root(V)\preceq_B\root(W)=\root(X)\preceq_B\root(U)$.
We have verified that the constructed partial run $U$ satisfies the statement
of the stronger variant of the lemma, which concludes the proof.
\end{proof}

With Lemma~\ref{adding:lemma} in hand, we can prove that for each accepting run
of $\A^+$ on a word $w$, there is an accepting run of $\A$ on $w$. This
requires to carry Lemma~\ref{adding:lemma} from finite partial runs to full
infinite runs.

\begin{lemma}\label{k:lemma}
A run $T$ of $\A$ with $\root(T) = \iota$ is accepting if and only if for every
$\pi\in T$, there exists a constant $k_{\pi}\in\nat$ such that every $\psi$
with $\pi \psi\in T$ and $|\psi|\geq k$ contains an accepting state.
\end{lemma}

\begin{proof}
\emph{(if)}
For every $\pi\in\branch(T)$, there is an infinite sequence of $k_0,k_1\ldots$
such that:
\begin{itemize}
\item
$k_0 = 0$ and
\item
for all $i\in\nat$, $k_i = k_{i-1} + k_{\pi^{n}}$ where $n = k_{i-1}+1$.
\end{itemize}
For all $i\in\nat$, every segment of $\pi$ between $k_{i-1}+1$ and $k_{i}$ contains an accepting
state, therefore $\pi$ contains infinitely many accepting states.

\emph{(only if)}
By contradiction. Suppose that there is $\pi \in T$ for which there is no $k_\pi$.
We will show that in this case, there must be $\psi\in Q^\omega$ such that
$\pi\psi\in\branch(T)$ and $\psi$ does not contain an accepting state (which
contradicts the assumption that $T$ is accepting).

We will give a procedure which returns $\psi^i$ for each $i\in\nat$ 
(based on the knowledge of $\psi^{i-1}$). For each $i\in\nat^0$, we will keep the invariant that for
$\pi\psi^i$, $k_{\pi\psi^i}$ does not exists and that $\psi^i$ does not contain an accepting state.
Since $\psi^0 = \epsilon$, the invariant holds for $i=0$.

Let the invariant hold for $i-1,i\in\nat$, and suppose that we have already
constructed $\psi^{i-1}$.  Denote $P$ the subset of $\succ T(\pi\psi^{i-1})$ containing nonaccepting states. $P$ must be nonempty, because if all the states from $\succ T(\pi\psi^{i-1})$ were accepting,  $k_{\pi\psi^{i-1}}$ would equal 1, violating the invariant for $i-1$. Then, there must be a state $q\in P$ such that $k_{\pi\psi^{i-1}q}$ does not exist, since otherwise we could put $k_{\pi\psi^{i-1}} = \max\{k_{\pi\psi^{i-1}p} \mid p\in P\} + 1$, which would also violate the invariant for $i-1$.
We choose $q$ as the continuation and put $\psi^{i} = \psi^{i-1} q$. Observe that this choice 
satisfied the invariant for $i$. 

We have shown that for every
$i\in\nat$, we can construct the $i$th
prefix $\psi^i$ of $\psi$ that does not contain an accepting state.  Therefore, the whole infinite path $\psi$ does not contain an accepting state, and the branch $\pi\psi$ of $T$ does not contain infinitely many accepting states. This contradicts the assumption that
$T$ is accepting.
\end{proof}

\begin{lemma}\label{infinite:lemma}
For every accepting run $T$ of $\A^+$ a word $w\in\Sigma^\omega$, there exists an accepting run $U$ of $\A$ on $w$.
\end{lemma}
\begin{proof}
For a tree  $X$ over $Q$, let $X(i) = \{\pi\in X\mid |\pi|\leq i\}$ be the
$i$th prefix of $X$
($X(0) = \emptyset$).
From Lemma~\ref{adding:lemma}, for each $i\in\nat$, there is a partial run
$U_i$ of $\A$ on $w$ such that $T(i)\moreacc U_i$ and $\root(T(i))\preceq_B\root(U_i)$.
As ${\preceq_B}\subseteq{\preceq_\iota}$, $\root(U_i) = \iota$.
Note that for all $\pi\in \branch(U_i)$, $|\pi|$ equals $i$, because only paths of the same length can be related by $\moreacc$.
Denote $\U = \{U_1,U_2,\ldots\}$.  $\U$ is an infinite set that for each $k\in\nat$ contains a partial run $U_k$ of $\A$ with all the branches of the length $k$. We will use $\U$ to construct the
infinite accepting run $U$.

Observe that for any infinite set $\V$ of partial runs of $\A$ and for
any $i\in\nat$, there has to be at least one partial run $W$ of $\A$ such that
for infinitely many $V\in \V$,  $W = V(i)$.
The reason is that for any $i\in\nat$, there is obviously only finitely many of
possible partial runs of the height $i$ that $\A$ can generate.

We prove the existence of $U$ by giving a procedure, which for every
$k\in\nat$ gives the $k$th prefix $U(k)$ of $U$.
\begin{itemize}
\item
Let $\U_0 = \U$ and let $U(0) = \emptyset$.
\item
For every $k\in\nat$, $U(k)$ is derived from $U(k-1)$ as follows.
Let $\U_k\subseteq\U$ be defined as
the set such that for all $i\in \nat$, $U_i\in \U_k$ iff $U(k-1) = U_i(k-1)$.
In other words, $\U_k$ is the subset of $\U$ of the partial runs with the $i$th prefix equal to $U(k-1)$.
Then, $U(k) = U_n(k)$ for some $n\geq k$ such that $U_n\in \U_k$ and there is
infinitely many $m\in \nat$ such that $U_m\in \U_k$ and $U_n(k) = U_m(k)$.
I other words, $U(k)$ is a tree that appears as the $k$th prefix of infinitely many partial runs in $\U_k$.
\end{itemize}
To see that this construction is well defined, observe that:
\begin{itemize}
\item $\U_0$ is infinite, and
\item for all $k\in\nat$, if $\U_{k-1}$ is infinite, then $U(k-1)$ is defined and
$\U_k$ is infinite.
\end{itemize}
Thus, $U(k)$ is well defined for every $k\in \nat$ and $U$ is a run of $\A$.

It remains prove that $U$ is accepting.
We will show that for every $\pi\in U$, there is $k_{\pi}\in\nat$ such that
every $\psi$ with $\pi \psi\in T$ and $|\psi|\geq k$ contains an accepting
state. By Lemma~\ref{k:lemma}, it will follow that $U$ is accepting.

Let us choose arbitrary $\pi\in U$. Let $n=|\pi|$. By Lemma~\ref{k:lemma}, for every $\pi'\in \branch(T_n)$,
there is $k_{\pi'}\in\nat$ such that
every $\psi'$ with $\pi' \psi'\in T$ and $|\psi'|\geq k_{\pi'}$ contains an accepting
state. Let $k = \max \{k_{\pi'}\mid \pi'\in \branch(T(n))\}$. By the construction of $U$, $T(n+k)\moreacc U(n+k)$.
This implies that
for every $\pi''\in\branch(U(n))$,
every $\psi''$ with $\pi'' \psi''\in T$ and $|\psi''|\geq k$ contains an accepting
state.  As $\pi$ in $\branch(U(n))$, we can put $k_\pi = k$ and we are done.
\end{proof}

\begin{theorem}\label{adding:theorem}
$\L(\A^+) = \L(\A)$.
\end{theorem}

\begin{proof}
The inclusion $\L(\A)\subseteq \L(\A^+)$ is obvious as $\L(\A^+)$ has riches both
transition function and the set of accepting states.
The inclusion $\L(\A^+)\subseteq \L(\A)$ follows immediately from Lemma
\ref{infinite:lemma}.
\end{proof}

\begin{corollary}\label{quotienting:corollary}
Quotienting with mediated equivalence preserves the language.
\end{corollary}

\section{Computing the Relations}\label{section:algorithm}
In this section, we describe algorithms for computing ABA forward and backward simulation, and mediated preorder.
For forward simulation, we use an
algorithm from \cite{fritz:state}, for backward simulation, we present an
algorithm based on a translation to an LTS simulation problem similar to the one
from Chapter~\ref{chapter:ta_reduction} for computing upward TA simulation. Mediated preorder is then computed by
the algorithm presented in Chapter~\ref{chapter:ta_reduction}. For the
mediated preorder to be useful for quotienting, we also need to remove
ambiguity before we start computing the backward simulation. This can
be done by a simple procedure presented in this section too.  
For the rest of the section, we fix an ABA $\A = (\Sigma,Q,\iota,\delta,\alpha)$. 

\paragraph{Forward Simulation.}

The algorithm for computing maximal forward simulation $\preceq_F$ on $\A$ can
be found in Fritz and Wilke's work~\cite{fritz:state} (it is called direct
simulation in their paper). They reduce the problem of computing maximal
forward simulation to a simulation game. Although Fritz and Wilke use a
slightly different definition of ABA, it is easy to translate $\A$ to an ABA
under their definition with $\O(n+m)$ states and $\O(nm)$ transitions and then
use their algorithm to compute $\preceq_F$. The time complexity of the above
procedure is $\O(nm^2)$.

\paragraph{Removing Ambiguity.}

As we have argued in Section~\ref{section:mediated}, $\A$ needs to be
$\preceq_F$-unambiguous for mediated minimisation. Here, we describe how to
modify $\A$ to make it $\preceq_F$-unambiguous. The modification does not
change the language of $\A$ and also the forward simulation relation
$\preceq_F$, therefore we do not need to recompute forward simulation again for
the modified automaton.

The procedure for removing ambiguity is simple. For every transition $p\xtr{a}
P$ with $P=\{p_1,\ldots, p_k\}$ and for each $i\in \{1,\ldots,k\}$, we check if
there exists some $i<j\leq k$ such that $p_j \preceq_F p_i$. If there is one,
remove $p_i$ from $P$. The time complexity of this procedure is obviously in
$\O(n^2m)$. 

We note that an alternative way is quotienting the automaton w.r.t.
forward simulation equivalence.

\subsection{Computing Backward Simulation}

Our algorithm for computing backward simulation is inspired by the algorithms
for computing tree automata simulations---we translate the problem of computing
maximal backward simulation on $\A$ to a problem of computing maximal
simulation on a labelled transition system. 

The reduction is very similar to the
reduction of the problem of computing tree automata backward simulation from
Chapter~\ref{chapter:ta_reduction}.  We first define the notion of an \emph{environment}, which is
a tuple of the form $(p, a, P \setminus \{p'\})$ obtained by removing a state
$p' \in P$ from the transition $p \xrightarrow{a} P$ of $\A$. Intuitively, an
environment records the neighbours of the removed state $p'$ in the transition
$p \xrightarrow{a} P$. We denote the set of all environments of $\A$ by
$\env(\A)$. Formally, we define the LTS $\A^\odot =(\Sigma,Q^\odot,\Delta^\odot)$
as follows:
\begin{itemize}
\item $Q^\odot = \{q^\odot \mid q\in\ Q\} \cup \{(p, a, P)^\odot \mid (p, a, P) \in \env(\A)\}$.
\item $\Delta^\odot = \{(p, a, P \setminus \{p'\})^\odot \xrightarrow{a} p^\odot, {p'}^\odot \xrightarrow{a}(p, a, P \setminus \{p'\})^\odot \mid P \in \delta(p,a), p'\in P\}$.
\end{itemize}

\begin{figure}[t]
\begin{center}
\includegraphics[height=25mm]{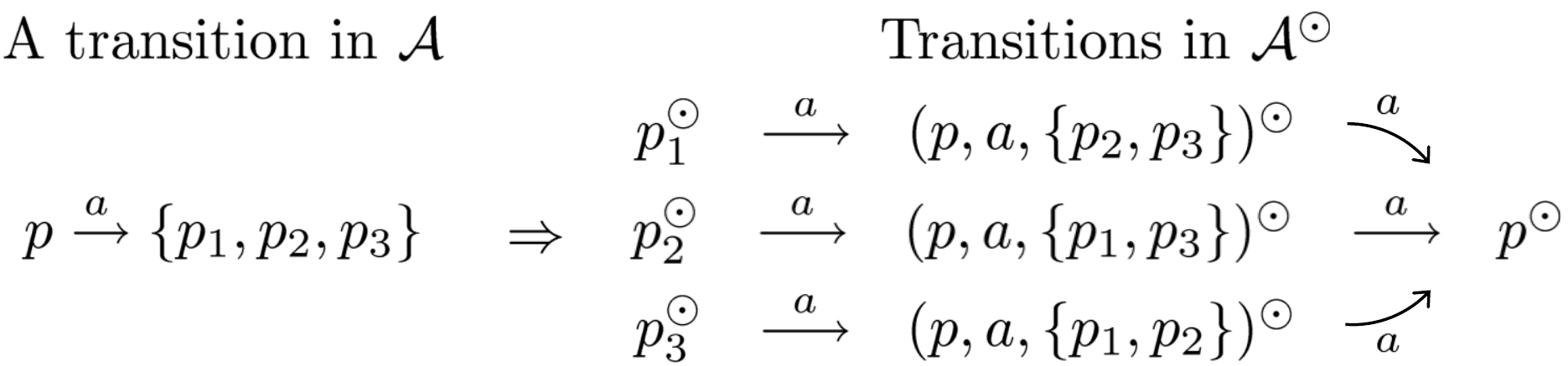}
\caption{An example of the reduction from an ABA transition to LTS transitions\label{fig:reduction}}
\end{center}
\end{figure}

An example of the reduction is given in Figure~\ref{fig:reduction}.  The goal
of this reduction is to obtain a simulation relation on $\A^\odot$ with the
following property: $p^\odot$ is simulated by $q^\odot$ in $\A^\odot$ iff $p
\preceq_B q$ in $\A$. However, the maximal simulation on $\A^\odot$ is not
sufficient to achieve this goal. Some essential conditions for backward
simulation (e.g., $p \preceq_B q \Longrightarrow p \preceq_\alpha q$) are
missing in $A^\odot$. This can be fixed by defining a proper initial preorder
$I$.

Formally, we let $I=\{(q^\odot_1,q^\odot_2)\mid {q_1 \preceq_\iota q_2}
\wedge {q_1 \preceq_\alpha q_2}\} \cup \{((p, a, P)^\odot, (r, a, R)^\odot)\mid
P\preceq_F^{\forall\exists}R \}$. Observe that $I$ is a preorder. Recall that
according to the definition of the backward simulation, $p \preceq_B r$ implies
that (1) $p \preceq_\iota r $, (2) $p \preceq_\alpha r $, and (3) for all
transitions $q \xtr{a} P\cup \{p\},p\not\in P$, there exists a transition $s
\xtr{a} R\cup \{r\},r\not\in R$ such that $q \preceq_B s$ and
$P\preceq_F^{\forall\exists}R$. The set $\{(q^\odot_1,q^\odot_2)\mid {q_1
\preceq_\iota q_2} \wedge {q_1 \preceq_\alpha q_2}\}$ encodes the conditions
(1) and (2) required by the backward simulation, while the set $\{((p, a,
P)^\odot, (r, a, R)^\odot)\mid P\preceq_F^{\forall\exists}R \}$ encodes the
condition (3).  A simulation relation ${\preceq^I}$ can be computed using the
aforementioned procedure with LTS $\A^\odot$ and the {\it initial} preorder
$I$. The following theorem shows the correctness of our approach to computing
backward simulation.

\begin{theorem} \label{theorem:backward_reduce}
For all $q,r\in Q$, we have $q \preceq_B r$ iff $q^\odot \preceq^I r^\odot$.
\end{theorem}
\begin{proof}
{\it (if)}
We define $\preceq$ to be a binary relation on $Q$ such that $p \preceq r$ iff
$p^\odot \preceq^I r^\odot$.  We show that $\preceq$ is a backward simulation
on $Q$ which immediately implies the result.

Suppose that $p \preceq r$ and $p' \xrightarrow{a} \{p\} \cup P$ where $p
\not\in P$ is a transition of $\A$.  Since $p \preceq r$, we know that $p^\odot
\preceq^I r^\odot$; and since $p' \xrightarrow{a} \{p\} \cup P$ is a transition
of $\A$, we know by definition of $\A^\odot$ that $p^\odot\xrightarrow{a}(p',
a, P)^\odot$ and $(p', a, P)^\odot \xrightarrow{a} p'^\odot$ are transitions in
$\A^\odot$.
Since $\preceq^I$ is a simulation, we can find two transitions
$r^\odot\xrightarrow{a}(r', a, R)^\odot$ and $(r', a, R)^\odot \xrightarrow{a}
r'^\odot$ in $\A^\odot$ with $(p', a, P)^\odot \preceq^I (r', a, R)^\odot$ and
$p'^\odot \preceq^I r'^\odot$.
From $p'^\odot \preceq^I r'^\odot$, $(p', a, P)^\odot \preceq^I (r', a,
R)^\odot$, and the definition of the initial preorder $I$, we have $p' \preceq
r'$ and $P \preceq^{\forall\exists}_F R$. It follows that $\preceq$ is in fact
a backward simulation parametrised by $\preceq_F$.

{\it (only if)}
Define $\preceq_\odot$ as a binary relation on $Q^\odot$ such that $p^\odot
\preceq_\odot r^\odot$ iff $p \preceq_B r$ and $(p, a, P)^\odot \preceq_\odot
(r, a, R)^\odot$ iff $P \preceq^{\forall\exists}_F R$ and $p \preceq_B r$. By
definition, $\preceq_\odot \subseteq I$.  We show that $\preceq_\odot$ is a
simulation on $Q^\odot$ which immediately implies the result. In the proof, we
consider two sorts of states in $\A^\odot$; namely those corresponding to
states and those corresponding to ``environments''.

Suppose that $p^\odot \preceq_\odot r^\odot$ and the transition
$p^\odot\xrightarrow{a}(p', a, P)^\odot$ is in $\A^\odot$. Since $p^\odot
\preceq_\odot r^\odot$, we know that $p \preceq_B r$. From the transition
$p^\odot\xrightarrow{a}(p', a, P)^\odot$ and by definition of $\A^\odot$, $p'
\xrightarrow{a} P\cup\{p\}$ is a transition in $\A$. Since $p \preceq_B r$,
there exists a transition $r' \xrightarrow{a} R\cup\{r\}$ in $\A$ such that $p'
\preceq_B r'$ and $P \preceq^{\forall\exists}_F R$. It follows that there
exists a transition $r^\odot\xrightarrow{a}(r', a, R)^\odot$ in $\A^\odot$ such
that $(p', a, P)^\odot \preceq_\odot (r', a, R)^\odot$.

Suppose that $(p, a, P)^\odot \preceq_\odot (r, a, R)^\odot$ and the transition
$(p, a, P)^\odot \xrightarrow{a}p^\odot$ is in $\A^\odot$. Since $(p, a,
P)^\odot \preceq_\odot (r, a, R)^\odot$, we know that $P
\preceq^{\forall\exists}_F R$ and $p \preceq_B r$. By definition of $\A^\odot$,
the transition $(r, a, R)^\odot \xrightarrow{a}r^\odot$ is in $\A^\odot$. Since
$p \preceq_B r$, we have $p^\odot \preceq_\odot r^\odot$. Together we have
there exists a transition $(r, a, R)^\odot \xrightarrow{a}r^\odot$ in
$\A^\odot$ such that $p^\odot \preceq_\odot r^\odot$. It follows that
$\preceq_\odot$ is a simulation on $Q^\odot$.
\end{proof}

\subsection{Complexity of Computing Backward Simulation} The complexity comes
from three parts of the procedure: (1) compiling $\A$ into its corresponding
LTS $\A^\odot$, (2) computing the initial preorder $I$, and (3) running
Algorithm~\ref{algorithm:LRT} from Chapter~\ref{chapter:LTS_simulation} for
computing the LTS simulation relation. 
Let $n$ and $m$ be the
number of states and transitions in $\A$, respectively. 
The LTS $\A^\odot$ has at most $nm$+$n$
states and $2nm$ transitions. It follows that Part (3) has both time complexity and space complexity $\O(|\Sigma|n^2m^2)$. As we will
show, among the three parts, Part (3) has the highest time and space complexity
and therefore computing backward simulation also has time and space complexity
$\O(|\Sigma|n^2m^2)$.  Under our
definition of ABA, every state has at least one outgoing transition for each
symbol in $\Sigma$. It follows that $m \geq |\Sigma|n$. Therefore, we can also
say that the procedure for computing maximal backward simulation has time
and space complexity $\O(nm^3)$.

\paragraph{Initial Preorder for Computing Backward Simulation.}
Recall that the preorder $I$ is the union of two
components: $\{(q^\odot_1,q^\odot_2)\mid {q_1 \preceq_\iota q_2} \wedge {q_1
\preceq_\alpha q_2}\}$ and $\{((p, a, P)^\odot, (r, a, R)^\odot)\mid \forall
{r_j \in R} \exists {p_i\in P}: p_i\preceq_F r_j \}$. It is trivial that the
first set can be computed by an algorithm with time complexity $\O(n^2)$.
However, a na\"ive algorithm (pairwise comparison of all different
environments in $\env(\A)$) for computing the second set has time complexity
$\O(n^4m^2)$. Here, we will describe a more efficient algorithm, which allows
the computation of $I$ in time $\O(n^2m^2)$ and space $\O(n)$.

The main idea of the algorithm is the following. For each pair of transitions
of $\A$, it computes all the pairs of environments that arise from them (by
deleting a right-hand side state) and are to be added to $I$ at once, reusing a
lot of information that a na\"ive algorithm would compute repeatedly for each
pair of environments. For a fixed pair of transitions, this procedure has time
complexity $\O(n^2)$ and space complexity $\O(n)$. Because $\A$ has at most
$m^2$ different pairs of transitions and the $\O(n)$ memory needed for the data
structures for one pair of transitions can then be reused for the other pairs, the second component of
$I$ can be this way computed in time $\O(n^2m^2)$ and space $\O(n)$.

We now explain how to efficiently compute all pairs of environments that arise
from a given pair of transitions and that are related by $I$. Let us fix transitions
$p\xrightarrow{a}P$ and $r\xrightarrow{a}R$. 
We will maintain a function $\beta: R\rightarrow
\{T,F\} \cup P$ such that:
\[
\beta(r') =
\left\{
  \begin{array}{ll}
    T & \mbox{ if at least two states in $P$ are forward smaller than $r'$}.\\
    F & \mbox{ if no state in $P$ is forward smaller than $r'$}.\\
    p' & \mbox{ if $p'$ is the only state in $P$ such that } p' \preceq_F r'.\\
\end{array}
\right.
\]
The function $\beta$ can be computed by lines 1-4 of
Algorithm~\ref{algorithm:genPairs} in time
$\O(n^2)$ and space $\O(n)$.%
Let us consider a pair of states $((p, a, P\setminus\{p'\})^\odot, (r, a,
R\setminus\{r'\})^\odot)$ in $\A^\odot$. This pair can be added to $I$ if and
only if the following two conditions hold:
\begin{enumerate}
  \item \label{c1}$\forall\hat{r} \in (R\setminus\{r'\}).\beta(\hat{r}) \neq F$.
  \item \label{c2}$\forall\hat{r} \in (R\setminus\{r'\}).\beta(\hat{r}) \neq p'$.
\end{enumerate}
The algorithm first pre-processes $p\xrightarrow{a}P$ and $r\xrightarrow{a}R$,
computing certain information that will allow us to check the two conditions in
constant time for every pair of environments arising from the two transitions.

The pre-processing needed for efficient checking of Condition~(\ref{c1}) is the following. We define $\hat{r}
\in R$ as the {\tt KeyState} if $\hat{r}$ is the only one state in $R$ such
that $\beta(\hat{r}) = F$. Given a function $\beta$, the {\tt KeyState}
can be found efficiently (with time complexity $\O(n)$ and space complexity
$\O(1)$) by scanning through $R$ and
\begin{itemize}
  \item if there exist two states $r_1,r_2\in R$ such that
$\beta(r_1)=\beta(r_2)=F$, the algorithm terminates immediately because it
follows that none of the pairs of environments generated from the given pair of
transitions satisfies the requirement of $I$;
  \item if there exists only one state such that $\beta$ maps it to $F$, let it
be the {\tt KeyState}.
\end{itemize}
Then we have Condition~(\ref{c1}) is satisfied if (1) there is no {\tt KeyState}
or (2) $r'$ is the {\tt KeyState}.

For efficient checking of Condition~(\ref{c2}), we maintain a
function $\gamma: P\rightarrow \{T,F\} \cup R$ such that
\[
\gamma(p') =
\left\{
  \begin{array}{ll}
    F       & \text{if } \beta^{-1}(p') = \emptyset \\
    r'      & \text{if } \beta^{-1}(p') = \{r'\}\\
    T       & \text{otherwise.}
\end{array}
\right.
\]
The function $\gamma$ can be found in time $\O(n^2)$ and space $\O(n)$
by scanning once through $\beta$ for each element of $P$.  With the function
$\gamma$, Condition~(\ref{c2}) can easily be verified by checking if
$\gamma(p')\in \{F, r'\}$, which means that for all the states $\hat{r}$ in
$R\setminus \{r'\}$, there is some state $\hat p$  different from $p'$ such that $\hat p \preceq_F
\hat{r}$.

\begin{algorithm}[!h]
    \KwIn{Two transitions $p\xrightarrow{a}P$ and $r\xrightarrow{a}R$ in $\A$.}
    \caption{\textit{Add Pairs of States to I}}
    \label{algorithm:genPairs}
    \tcc{Computing function $\beta$}
    \lForAll{$r'\in R$}{
        $\beta(r'):=F$\;
    }
    \ForAll{$p'\in P, r'\in R$}{
        \If{$p' \preceq_F r'$}{
            \lIf{$\beta(r')=F$}{
                $\beta(r'):=p'$\;
            }
            \lElse{
                $\beta(r'):=T$\;
            }
        }
    }

    \tcc{Preprocessing for Condition (1) (computing {\tt KeyState})}
    \lForAll{$r'\in R$}{
        \If{$\beta(r') = F$}{
            \lIf{there is no {\tt KeyState}}{Let $r'$ be the {\tt KeyState}\;}
            \lElse{Terminate the algorithm\;}
        }
    }

    \tcc{Preprocessing for Condition (2) (computing function $\gamma$)}
    \lForAll{$p'\in P$}{
        $\gamma(p'):=F$\;
    }
    \lForAll{$r'\in R$}{
        \If{$\beta(r') \notin \{T, F\}$}{
            \lIf{$\gamma(\beta(r'))=F$}{
                $\gamma(\beta(r')):=r'$\;
            }
            \lElse{
                $\gamma(\beta(r')):=T$\;
            }
        }
    }
    \tcc{main loop}
    \ForAll{$p'\in P, r'\in R$}{
        \If{there is no {\tt KeyState} or $r'$ is the {\tt KeyState}}{
            \lIf{$\gamma(p')\in \{F, r'\}$}{add $((p, a, P\setminus\{p'\})^\odot, (r, a, R\setminus\{r'\})^\odot)$ to $I$}
        }
    }

\end{algorithm}
In Algorithm~\ref{algorithm:genPairs}, we first find out the {\tt KeyState} if
there is one and compute the function $\gamma$ from $\beta$. Then in the main
loop, for each pair of states $((p, a, P\setminus\{p'\})^\odot, (r, a,
R\setminus\{r'\})^\odot)$, we check if it belongs to $I$ by verifying the
Conditions~(\ref{c1}) and (\ref{c2}). Since it is easy to see that
Algorithm~\ref{algorithm:genPairs} has time complexity
$\O(n^2)$ and space complexity $\O(n)$ (not taking into account the space needed for $I$ itself), we can conclude that the initial preorder $I$ can be computed in
time $\O(n^2m^2)$ and space $\O(m^2)$ (encoding of $I$). This leads to the following
theorem that summarises complexity of computing backward simulation.

\begin{theorem} \label{theorem:backward_complexity}
Maximal backward simulation parametrised by a given transitive and reflexive
forward simulation can be computed with both time and space complexity $\O(|\Sigma|n^2m^2)
\subseteq\O(nm^3)$.
\end{theorem}

\newpage
\section{Experimental Results}\label{section:experiments}

In this section, we evaluate the performance of ABA mediated minimisation by applying it to accelerate the algorithm proposed by Vardi and Kupferman~\cite{kupferman:weak} for complementing nondeterministic B\"uchi automata (NBA). In this algorithm, ABA's are used as an intermediate notion for the complementation. To be more specific, the complementation algorithm has two steps: (1)~it translates an NBA to an ABA that recognises its complement language, and (2) it translates the ABA back to an equivalent NBA. The second step is an exponential procedure (exponential in the size of the ABA), hence reducing the size of the ABA before the second step usually pays off.

The experimentation is carried out as follows. Three sets of 100 random NBA's (of $|\Sigma|=$ 2,4, and 8, respectively) are generated by the GOAL tool \cite{tsay:goal} and then used as inputs of the complementation experiments.
We compare results of experiments performed according to the following different options: (1)~\textbf{Original: }keep the ABA as it is, (2) \textbf{Mediated: } minimising the ABA with mediated equivalence, and (3) \textbf{Forward:} minimising the ABA with forward equivalence.

For each input NBA, we first translate it to an ABA that recognises its complement language. The ABA is (1)~processed according to one of the options described above and then (2) translated back to an equivalent NBA using an exponential procedure~\footnote{For the option ``Original'', we also use the optimisation suggested in~\cite{kupferman:weak} that only takes a consistent subset.}. The results are given in Table~\ref{tab:performance} and Table~\ref{tab:compare}. Table~\ref{tab:performance} is an overall comparison between the three different options and Table~\ref{tab:compare} is a more detailed comparison between \textbf{Mediated} and \textbf{Forward} minimisation.

\begin{table}[h]
\begin{center}
\caption{Combining minimisation with complementation.}\label{tab:performance}
\begin{tabular}{|c|c|c|c|c|c|c|c|}
  \hline
                    & \mr{2}{$|\Sigma|$} & \mc{2}{NBA}                 & \mc{2}{Complemented-NBA} & \mr{2}{Time (ms)} & Timeout\\
\cline{3-6}                    &                    & St. & Tr.                  & St. & Tr.       &                   & (10 min)\\
\hline
Original        	&  \mr{3}{2}         &\mr{3}{2.5}  &\mr{3}{3.3}	  & 13.9& 52.75	    &5500.9             &0 	 \\
Mediated	        &                    &             &              & 6.68& 34.02	    &524.7	            &0       \\
Forward		        &                    &             &              & 9.45& 55.25	    &5443.7	            &1       \\
\hline
\hline	
Original        	&  \mr{3}{4}         &\mr{3}{3.3}  &\mr{3}{6.0}	  & 46.4  & 348.5	&9298.6             &6 	 \\
Mediated	        &                    &             &              & 20.42 & 235.5	&1985.4	            &6       \\
Forward		        &                    &             &              & 26.88 & 325.6	&1900.6	            &7       \\
\hline
\hline
Original        	&  \mr{3}{8}         &\mr{3}{4.7}  &\mr{3}{11.9}  & 127.1.3  & 1723.4	&33429.4            &24 	 \\
Mediated	        &                    &             &              & 57.63 & 1738.3	&12930.6            &21       \\
Forward		        &                    &             &              & 81.23 & 2349.2	&22734.2	        &24       \\
\hline
\end{tabular}
\end{center}
\end{table}

In Table~\ref{tab:performance}, the columns ``NBA'' and ``Complemented-NBA'' are the average statistical data of the input NBA and the complemented NBA. The column ``Time(ms)'' is the average execution time in milliseconds. ``Timeout'' is the number of cases that cannot finish within the timeout period (10 min). Note that in the table, the cases that cannot finish within the timeout period are excluded from the average number. From this table, we can see that minimisation by mediated equivalence can effectively speed up the complementation and also reduce the size of the complemented NBA's.

\begin{table}[h]
\begin{center}
\caption{Comparison: \emph{Mediated} vs. \emph{Forward}}\label{tab:compare}
\begin{tabular}{|l|c|c|c|c|c|}
  \hline

    & \mr{2}{$|\Sigma|$} & \mc{2}{Minimised-ABA} & \mc{2}{Complemented-NBA}\\
\cline{3-6}
    &                    & St. & Tr.& St. & Tr.\\
\hline
Average   &2 &33.54\% & 51.62\% & 63.3\% & 235.56\%\\
\cline{2-6}
Difference & 4 &36.24\% & 51.44\% & 89.9\% & 298.99\%\\
\cline{2-6}
& 8 &27.94\% & 40.88\% & 152.3\% & 412.7\%\\
\hline
\end{tabular}
\end{center}
\end{table}

In Table~\ref{tab:compare}, we compare the performance between \textbf{Mediated} and \textbf{Forward} minimisation in detail. The columns ``Minimised-ABA'' and ``Com\-ple\-men\-ted-NBA'' are the average difference in the sizes of the ABA after minimisation and the complemented BA. From the table, we observe that mediated minimisation results in a much better reduction than forward minimisation.

\section{Conclusion and Future Work}\label{section:discussion}

We have introduced a novel notion of alternating automata backward simulation.
Inspired by our previous work on tree automata simulation reduction, we
combined forward and backward simulation to form a coarser relation called
mediated preorder and showed that quotienting wrt. mediated equivalence
preserves the language of ABA. Moreover, we developed an efficient algorithm
for computing backward simulation and mediated equivalence. Experimental
results show that the mediated reduction of ABA significantly outperforms the
reduction based on forward simulation.

In the future, we would like to extend our experiments to other
applications such as LTL to NBA translation. Furthermore, we would like to
extend the mediated equivalence by building it on top of even coarser forward
simulation relations, e.g., \emph{delayed} or \emph{fair} forward simulation
relations~\cite{fritz:simulation}. Also, we would like to study the possibility
of using mediated preorder to remove redundant transitions (similar to the
approaches described in~\cite{somenzi:efficient}). We believe that the
extensions described above can significantly improve the performance of
mediated reduction.

\chapter{Conclusions and Future Directions}
\label{chapter:conclusions}
Each of the main chapters contains detailed conclusions
concerning the specific topic. Here, we summarise once more the main points
and discuss possible further research directions.

\enlargethispage{2mm}

\section{A Summary of the Contributions}
The main focus of this thesis was on developing efficient methods for handling
nondeterministic tree automata. We have studied simulation based methods for
size reduction of tree automata and methods for universality and language
inclusion testing. We have found efficient algorithms for computing tree
automata simulations that are based on translating problems of computing tree
automata simulations to problems of computing common simulation over LTS. For
this, we developed an efficient LTS simulation algorithm which is an extension
of the fastest Kripke structure simulation algorithm.  The same TA to LTS
translations as for the TA simulations can be used also for computing tree
automata bisimulations.  Thus, all tree automata (bi)simulations can be computed
in a uniform and elegant way, with possibility of using the most efficient LTS
simulation and bisimulation algorithms. We have discovered a new type of
relations that we call mediated equivalences that can be used for quotienting
tree automata as well as for word automata. Mediated equivalence arises from a combination of upward and
downward simulation, it includes downward simulation and thus gives a better
reduction, as we confirm also experimentally. Since the combination principle
allows also combining simulations with bisimulations, we have obtained a scale
of TA mediated equivalences that offer a fine choice between reduction power
and computational cost.

To solve language inclusion problem for tree automata, we have adapted the so
called antichain universality and inclusion checking method for FA
\cite{doyen:antichain}.  According to our experiments, this optimisation of the
classical subset construction method leads to a major speed-up of the TA
language inclusion and universality tests.  We then improve the antichain
method for both FA and TA by interconnecting it with the simulation based
methods. This again significantly improves efficiency of the algorithms.

We have shown practical applicability of the above TA reduction and inclusion
testing methods by applying them in the framework of abstract regular tree
model checking.  These algorithms allowed us to build a version of ARTMC method
purely on nondeterministic tree automata, avoiding determinisation completely.
According to our experiments, this greatly improved efficiency and scalability
of the ARTMC method.

Since our tree automata reduction methods are based on quite simple and general
principles, applying them for other types of automata comes into
consideration.  We have done this for alternating B\"uchi automata, for which
we have introduced a notion of backward simulation and defined the mediated
equivalence analogically as in the case of tree automata. As shown by our
experiments, mediated equivalence gives very good reduction even in the case of
ABA. 

\section{Further Directions}
There is a number of interesting directions of further work.  We have already
started to work on an algorithm for computing simulation on Kripke structures and
LTS that would match the best time complexity of the algorithm
\cite{ranzato:new} and also the best space complexity of the algorithm
\cite{gentiliny:fromBisimulation}.  We are considering extensions of our
simulation reduction methods to other types of automata, such as hedge
automata, weighted tree automata, or nested word automata. Also the mediation
principle itself can be further elaborated.  We already have some preliminary
results suggesting that it is possible to define a hierarchy of coarser and
coarser relations similar to the mediated equivalence (and suitable for
quotienting automata), where a mediated relation of level $i$ is used to induce
a mediated relation of level $i+1$.  The finite automata minimisation/reduction
is an interesting problem itself and we are thinking about reduction techniques
based on other principles than simulation quotienting.  For instance, an
efficient reduction heuristic based on the theory of universal automaton
\cite{arnold:note,polak:minimalization,kameda:state,carez:minimalization} could
possibly be designed. 

Further, we are still working on the tree automata language inclusion problem.
We are developing a universality and language inclusion checking algorithm for
tree automata that proceeds downwards (wrt. tree automata transition relation)
and makes use of downward simulation, in contrary to the upward algorithm from
Chapter~\ref{chapter:fa_ta_inclusion} that exploits only upward simulation.
Similarly as our reduction techniques, our language inclusion and universality
antichain/simulation techniques can be extended for other types of automata.
We have shown this in \cite{abdulla:simulationsubsumption} for the B\"uchi automata language inclusion problem and we are
continuing the work on this topic. Further, 
we do not restrict ourselves to simulation based techniques. One could think
for instance about using some abstraction techniques as in
\cite{ganty:fixpoint}, and
it may also be interesting to look for inspiration at the areas of decision
procedures of logics or solving other hard problems such as QBF. 

Our work on alternating B\"uchi automata simulation reduction can be continued
in the way of looking at more advanced handling of B\"uchi acceptance condition.
More specifically, we would like to study possibilities of constructing a
mediated equivalence from delayed or fair simulation \cite{fritz:simulation},
which could lead to even better reductions. 

Last, we are working towards applying our methods in practice. We are
developing an efficient BDD based library that would provide procedures for
handling nondeterministic tree automata (in the style of \cite{klarlund:MONA}).
This work includes also a development of BDD versions of our algorithms, which is
itself an interesting problem.  We are also working on an ARTMC based method for
verification of pointer manipulating programs that will make use of our TA
reduction and language inclusion checking techniques. 

\section{Publications Related to this Thesis} 
The algorithm for computing simulations over labelled transition systems
appeared in \cite{abdulla:computing}.  The tree automata reduction methods and
algorithms for computing simulations and bisimulations were published in
\cite{abdulla:computing,abdulla:composed,abdulla:uniform}. The generalisation
of the antichain universality and language inclusion method for TA appeared in
\cite{bouajjani:antichain, abdulla:when}. The combination of the antichain and
simulation methods was published in \cite{abdulla:mediating}.  Finally, the
results on ABA simulation reduction are from
\cite{abdulla:simulationsubsumption}.


The following publications are also to a large degree outcomes of work on this
thesis. The work \cite{holik:optimizing} presents optimisations of the
algorithm for computing simulations on LTS from
Chapter~\ref{chapter:LTS_simulation}. In \cite{holik:counterexample}, we fix
some problems in counterexample guided refinement loop for complex systems that
were discovered within the work on the ARTMC tool presented in
Section~\ref{section:rtmc}. The work \cite{abdulla:simulationsubsumption}
presents an application of our simulation based subsumption principle in
B\"uchi automata inclusion testing.

Full versions of the above mentioned papers were published as the technical reports
\cite{abdulla:computing:tr,bouajjani:antichain:tr,abdulla:uniform:tr,abdulla:when:tr,abdulla:composed:tr,abdulla:simulationsubsumption:tr,abdulla:mediating:tr,holik:optimizing:tr}. The works \cite{abdulla:composed} and \cite{abdulla:uniform} first appeared as \cite{abdulla:composed:conf} and \cite{abdulla:uniform:conf}.


\bibliographystyle{alpha}
\bibliography{literature}

\end{document}